\documentclass[
    notitlepage,
    twocolumn,
    pra,
    superscriptaddress,
    showkeys,
    nofootinbib,
    nobibnotes,
    10pt
]{revtex4-2}

\usepackage[T1]{fontenc}
\usepackage[utf8]{inputenc}
\usepackage{lmodern}

\usepackage[german,french,english]{babel}
\usepackage[dvipsnames]{xcolor}
\usepackage{graphicx}
\usepackage{tikz}
\usepackage{pgfplots}
\usepackage{tikz-3dplot}
\usepackage{tkz-euclide}
\usepackage{tikz-cd}

\usetikzlibrary{arrows}
\usetikzlibrary{arrows.meta}
\usetikzlibrary{decorations.markings}
\usetikzlibrary{decorations.pathmorphing}
\usetikzlibrary{calc}
\usetikzlibrary{intersections}
\usetikzlibrary{pgfplots.fillbetween}
\usetikzlibrary{shapes.geometric}
\usetikzlibrary{quotes}

\tikzset{>={Latex[length=6pt,width=5pt]}}

\newlength\tickoffset
\pgfplotsset{compat=1.16}

\definecolor{myblue}{RGB}{34, 130, 226}
\definecolor{myred}{RGB}{220, 44, 61}
\definecolor{myyellow}{RGB}{230, 161, 0}
\definecolor{mygreen}{RGB}{0, 173, 107}
\definecolor{mygray}{RGB}{184, 184, 184}
\definecolor{mypurple}{RGB}{71, 0, 99}

\usepackage{amsmath}
\usepackage{amsthm}
\usepackage{amssymb}
\usepackage{mathrsfs}
\usepackage{mathtools}
\usepackage{tensor}
\usepackage{mathrsfs}
\usepackage{array}
\usepackage{extarrows}
\usepackage{bm}
\usepackage{nicematrix}
\usepackage{slashed}
\usepackage[normalem]{ulem}
\usepackage{siunitx}

\usepackage{titlesec}
\usepackage{wrapfig}
\usepackage{enumitem}
\usepackage{needspace}
\usepackage{tcolorbox}
\usepackage{longtable}
\usepackage{stackengine}
\usepackage{changepage}
\usepackage{scalerel}
\usepackage{stackengine}
\usepackage{placeins}
\usepackage{footnote}
\usepackage{subfigure}

\titleformat{\paragraph}[runin]
{\bfseries}{\theparagraph}{1em}{}
\titlespacing{\paragraph}{0pt}{0pt}{*1.5}



\usepackage{hyperref}
\hypersetup{
	colorlinks,
	linktoc=page,
	linkcolor={myblue},
	citecolor={myred},
	urlcolor={myred}
}
\providecommand{\doi}[1]{}
\providecommand{\eprint}[2][]{}
\renewcommand{\doi}[1]{\href{http://dx.doi.org/#1}{\nolinkurl{doi:#1}}}
\renewcommand{\eprint}[2][arXiv]{#1:#2}

\usepackage[noabbrev]{cleveref}

\setlist{nosep}

\usepackage{xpatch}
\makeatletter
\xpatchcmd{\proof}{\topsep6\p@\@plus6\p@\relax}{}{}{}
\makeatother

\makeatletter
\@ifclassloaded{revtex4-2}{%
    \renewcommand\onecolumngrid{%
        \do@columngrid{one}{\@ne}%
        \def\set@footnotewidth{\onecolumngrid}%
        \def\footnoterule{\kern-6pt\hrule width 1.5in\kern6pt}%
    }
    \renewcommand\twocolumngrid{%
        \def\footnoterule{%
        \dimen@\skip\footins\divide\dimen@\thr@@
        \kern-\dimen@\hrule width.5in\kern\dimen@}
        \do@columngrid{mlt}{\tw@}
    }%
}{}
\makeatother

\usepackage{tcolorbox}
\tcbuselibrary{theorems}
\tcbuselibrary{skins}
\tcbuselibrary{breakable}
\tcbuselibrary{hooks}

\newtcbox{\inlinebox}[1][mygray]{on line,
arc=3pt,colback=#1!3,colframe=#1,
before upper={\rule[-3pt]{0pt}{10pt}},boxrule=1pt,
boxsep=0pt,left=2pt,right=2pt,top=1pt,bottom=.5pt}

\newtcolorbox{mybox}{%
    enhanced,breakable,
    frame hidden,
    colframe=mypurple,
    colback=mypurple!3,
    coltitle=black,
    fonttitle=\bfseries,
    colbacktitle=mypurple!3,
    borderline={0.2mm}{0mm}{mypurple},
    top=1mm,
    bottom=1mm,
    left=1.5mm,
    right=1.5mm,
    before skip=1.5ex,
    after skip=1.5ex
}

\makeatletter
\newtcbtheorem[auto counter]{example}{Example}
{
    theorem style=plain,enhanced,breakable,
    frame hidden,
    colframe=myyellow,
    colback=myyellow!3,
    coltitle=black,
    fonttitle=\bfseries,
    colbacktitle=myyellow!3,
    borderline={0.2mm}{0mm}{myyellow},
    top=1mm,
    bottom=1mm,
    left=1.5mm,
    right=1.5mm,
    before skip=1ex,
    after skip=1ex
}
{exa}
\makeatother

\newtcbtheorem[auto counter]{proposition}{Proposition}
{
    theorem style=plain,enhanced,breakable,
    frame hidden,
    colframe=myblue,
    colback=myblue!3,
    coltitle=black,
    fonttitle=\bfseries,
    colbacktitle=myblue!3,
    borderline={0.2mm}{0mm}{myblue},
    top=1mm,
    bottom=1mm,
    left=1.5mm,
    right=1.5mm,
    before skip=1ex,
    after skip=1ex
}
{prop}

\newtcbtheorem[use counter from=proposition]{lemma}{Lemma}
{
    theorem style=plain,enhanced,breakable,
    frame hidden,
    colframe=myred,
    colback=myred!3,
    coltitle=black,
    fonttitle=\bfseries,
    colbacktitle=myred!3,
    borderline={0.2mm}{0mm}{myred},
    top=1mm,
    bottom=1mm,
    left=1.5mm,
    right=1.5mm,
    before skip=1ex,
    after skip=1ex
}
{lem}

\newtcbtheorem[auto counter]{theorem}{Theorem}
{
    theorem style=plain,enhanced,breakable,
    frame hidden,
    colframe=mypurple,
    colback=mypurple!3,
    coltitle=black,
    fonttitle=\bfseries,
    colbacktitle=mypurple!3,
    borderline={0.2mm}{0mm}{mypurple},
    top=1mm,
    bottom=1mm,
    left=1.5mm,
    right=1.5mm,
    before skip=1ex,
    after skip=1ex
}
{thm}

\newif\ifproofread

\def\bbR{\mathbb R}
\def\bbC{\mathbb C}

\def\bbZ{\mathbb Z}

\newcommand\bb[1]{\boldsymbol{#1}}

\def\id{\mathrm{id}}
\def\sfid{\mathsf{id}}
\def\tr{\mathsf{tr}}

\def\i{\mathrm i}
\def\ex{\mathrm e}

\def\st{\,:\,}

\def\iso{\cong}

\def\der{\mathrm d}

\newcommand{\group}[1]{\mathrm{#1}}

\def\SU{\group{SU}}

\def\gpU{\group{U}}

\newcommand{\A}{\mathcal A}

\renewcommand{\H}{\mathcal H}

\newcommand{\J}{\mathcal J}

\renewcommand{\O}{\mathcal O}
\renewcommand{\P}{\mathcal P}

\newcommand{\Z}{\mathcal Z}

\newcommand{\ket}[1]{\left|#1\right\rangle}
\newcommand{\bra}[1]{\left\langle#1\right|}
\newcommand{\braket}[2]{\left\langle#1\middle|#2\right\rangle}
\newcommand{\ketbra}[2]{\left|#1\right\rangle\!\left\langle#2\right|}
\newcommand{\sket}[1]{\ket{\smash{#1}}}
\newcommand{\sbra}[1]{\bra{\smash{#1}}}
\newcommand{\sbraket}[2]{\braket{\smash{#1}}{\smash{#2}}}
\newcommand{\sketbra}[2]{\ketbra{\smash{#1}}{\smash{#2}}}

\newcommand{\strike}[1]{\ifmmode\text{\sout{\ensuremath{#1}}}\else\sout{#1}\fi}

\stackMath
\newcommand\reallywidehat[1]{%
	\savestack{\tmpbox}{\stretchto{%
			\scaleto{%
				\scalerel*[\widthof{\ensuremath{#1}}]{\kern-.6pt\bigwedge\kern-.6pt}%
				{\rule[-\textheight/2]{1ex}{\textheight}}
			}{\textheight}%
		}{0.5ex}}%
	\stackon[1pt]{#1}{\tmpbox}%
}

\def\ie{\textit{i.e.}}
\def\Ie{\textit{I.e.}}
\def\eg{\textit{e.g.}}
\def\Eg{\textit{E.g.}}

\def\apriori{\textit{a priori}}

\def\viceversa{\textit{vice versa}}

\newtheoremstyle{named}{}{}{}{}{\bfseries}{:}{.5em}{\thmnote{#3 }#1}
\theoremstyle{named}

\makeatletter
\def\namedlabel#1#2{\begingroup
   \def\@currentlabel{#2}%
   \label{#1}\endgroup
}
\makeatother

\begin{document}

\title{The Perspectives of Non-Ideal Quantum Reference Frames}

\author{Sébastien C. Garmier}
\thanks{\url{kontakt@sebastiengarmier.ch}}
\affiliation{Institute for Theoretical Physics, ETH Zurich,
Wolfgang-Pauli-Strasse 27,
8093 Zurich, Switzerland}

\author{Ladina Hausmann}
\affiliation{Institute for Theoretical Physics, ETH Zurich,
Wolfgang-Pauli-Strasse 27,
8093 Zurich, Switzerland}

\author{Esteban Castro-Ruiz}
\affiliation{Institute for Quantum Optics and Quantum Information (IQOQI),
Austrian Academy of Sciences, Boltzmanngasse 3, 1090 Vienna, Austria}

\begin{abstract}
    \noindent
    We define the perspective of any quantum reference frame (QRF)
    and construct reversible transformations between different perspectives.
    Our construction is based on two principles motivated operationally
    by the change from relative to absolute coordinates and leads to
    an incoherent group averaging approach with general symmetry group.
    Thereby, it extends the framework of
    \href{https://doi.org/10.1038/s42005-025-02036-x}{[Commun Phys 8, 187 (2025)]}
    from ideal QRFs, which generally require infinite resources
    like energy or angular momentum, to non-ideal QRFs,
    with only finite resources.
    We find that the perspective of a non-ideal QRF deviates significantly
    from that of an ideal QRF: Firstly, systems described relative to
    a non-ideal QRF appear superselected.
    Secondly, the structure of the perspective of a non-ideal QRF attests
    that successive relational operations on a system lead to
    back-reaction on this QRF.
\end{abstract}

\maketitle


\section{Introduction}
\label{sec:introduction}

\noindent
In a relational approach to physics,
observables are typically defined relative to a \emph{reference frame}.
In practice, such frames are realized by physical systems,
like clocks and rulers~%
\cite{Einstein1905,Misner1973,Taylor1992,Brown2009}.
If one takes the standpoint that quantum theory accurately describes such systems,
one should also seek a formulation of physics relative to
\emph{quantum reference frames} (QRFs)~%
\mbox{\cite{Eddington1949,Aharonov1967,Aharonov1984,Rovelli1991,Angelo2011,Bartlett2007}}.
Recently, there has been much progress towards constructing
quantum versions of \emph{reference frame or coordinate transformations}~%
\mbox{\cite{Aharonov1984,Merriam2005,Giacomini2019,Hamette2020,
Vanrietvelde2020,Glowacki2023,Carette2025,
Hamette2021,Krumm2021,Kabel2025,CastroRuiz2025a,DeVuyst2025,
Ballesteros2025,CastroRuiz2025,Hamette2025}},
giving rise to different frameworks for QRFs,
with applications to quantum foundations~%
\cite{Giacomini2019,AhmadAli2022,Krumm2021,JorqueraRiera2025}
and the overlap between quantum and gravitational physics~\cite{Giacomini2019,CastroRuiz2017,CastroRuiz2020,Hamette2023,Kabel2025}.

The simplest type of QRF are ideal QRFs, \ie, systems where all possible
frame orientations are assigned perfectly distinguishable states.
However, they often do not exist in nature as they
typically need infinite resources for their construction~%
\cite{Fleming1973,Peres1980,Mandelstam1991,Bartlett2007,Loveridge2018}.
For instance, an infinitely precise clock would require infinite energy spread.
In turn, such an infinite spread has important consequences when considering gravity,
as it would strongly disturb spacetime.
Generally, the fact that gravity cannot be shielded
implies ultimate physical limits to the workings of our clocks and rulers~%
\cite{Peres1960,Garay1995,Gorelik2005,Bronstein2012,CastroRuiz2017}.
Understanding non-ideal QRFs is therefore crucial to study situations
at the interface of quantum theory and gravitational physics.
There has been much progress in the study of non-ideal QRFs~%
\cite{Bartlett2007,Hamette2021,Glowacki2023,Hoehn2023a,Carette2025,
AhmadAli2022,DeVuyst2025,DeVuyst2025a,Fewster2024},
but due to their physical importance and added technical difficulty compared to ideal QRFs,
many questions remain open.

In this work, we build a framework for defining the perspectives of
general QRFs for a wide range of symmetry groups
(including non-compact ones),
and obtain unitary QRF transformations between perspectives.
Our framework assumes consistency with the potential existence
of an external laboratory frame and
leads to an \emph{incoherent group averaging approach}.
It extends the so-called \emph{extra-particle framework}
of~\cite{CastroRuiz2025a} to non-ideal QRFs.
Treating non-ideal QRFs is not trivial,
as the straightforward application of well-known QRF ``jumping rules''
(\eg~\cite{Giacomini2019,Hamette2020,CastroRuiz2025a})
would lead to non-unitary transformations,
which are incompatible with an interpretation
of QRF transformations as quantum coordinate changes.
Here, we instead define the notion of a QRF perspective based
on two fundamental assumptions called the
\hyperref[ppl:perspective]{perspective} and \hyperref[ppl:invariance]{invariance} principles.
Intuitively, the perspective principle formalizes the idea that from the QRF perspective,
the dependence of relational operators on the QRF becomes implicit,
and the invariance principle ensures that all QRF perspectives are treated on equal footing.
A QRF perspective is then a factorization of the total Hilbert space
which implements both principles.
Like in the ideal case~\cite{CastroRuiz2025a},
unitarity of the QRF transformations relies crucially on invariant
operators that commute with the relational ones,
the so-called \emph{extra particle}.
Our mathematical framework is compatible with different
physical interpretations and assumptions regarding QRFs.
Along the way, we discuss these interpretations and comment
on their implications for the foundations of quantum theory.

We find that the perspective of a non-ideal QRF differs
from the perspective of an ideal QRF
in qualitatively different ways, and it is not simply a ``blurred'' version of the ideal case.
We find a range of \emph{superselection} when a system $S$ is described
from the perspective of a QRF $A$,
extending from maximal superselection if the QRF is trivial
to no superselection when the QRF is ideal.
Furthermore, $S$ as described from the laboratory's perspective
is not generally isomorphic to $S$ as described from $A$'s perspective.
Particularly, if $S$ is bipartite from the laboratory's perspective,
it might not be bipartite in $A$'s perspective
(see~\cite{AhmadAli2022} for a related result in the
perspective-neutral framework);
yet, a bipartite structure is approximated as the frame becomes ideal.
Relatedly, for two commuting operators on $S$ in the laboratory perspective
the corresponding operators in $A$'s perspective might not commute.
Finally, the algebra containing relational operators between observed
system $S$ and QRF~$A$ can even have a larger number of degrees of freedom
than $S$ described from the perspective of the laboratory.
This counter-intuitive fact can be interpreted as some kind of
\emph{back-reaction} of the system onto the QRF due to the non-ideality of the latter.
We apply the framework in detail for Abelian symmetry groups
to investigate the above phenomena and discuss an explicit example
of two QRFs with different sizes describing each other.

The paper is structured as follows:
In section \ref{sec:QRF_POV} we introduce the notion of a QRF
perspective in general; this is the heart of our framework.
Section \ref{sec:extreme_cases} then derives explicit
expressions for QRF perspectives and transformations in
important extreme cases. Section \ref{sec:non_ideal_case}
discusses the general non-ideal case.
The case of an Abelian symmetry group
is dealt with in section \ref{sec:POVs_Abelian}.
Finally, we explore physical consequences of our framework
as well as open questions in section \ref{sec:conclusion_outlook}.
Technical details can be found in the appendices.

\vspace{-1em}
\section{Quantum Reference Frame Perspectives}
\label{sec:QRF_POV}

\noindent
Consider two quantum systems $A$ and $S$, with respective Hilbert
spaces $\H_A$ and $\H_S$, described from the perspective of a
laboratory $L$. We assume the laboratory is equipped with devices such
that all possible states and measurements of the total system $AS$ can
be characterized operationally. Ultimately, however, we are interested
in $A$'s perspective rather than $L$'s. That is, our goal for this
section is to make sense of and to define the
\emph{perspective of the quantum system $A$ serving as a QRF}.

To begin, consider the following example: 
\begin{example}{}{classical_A}
    Suppose for the moment that $A$ is a classical particle
    with position $x_A^{|L}$.
    Let $S$ be a quantum particle described by the position
    $\smash{\hat x_S^{|L}}$ and momentum $\smash{\hat p_S^{|L}}$,
    all defined from the perspective of $L$ (indicated by ``$|L$'').
    Describing $S$ from the perspective of $A$ entails
    a classical reference frame transformation~\cite{Ballentine2014}:
    we shift $\smash{\hat x_S^{|L}}$ by $\smash{-x_A^{|L}}$, obtaining the relational operator
    $\hat X_S^{|A} := \hat x_S^{|L} - x_A^{|L} \hat\id_S$, which we
    call the position of $S$ relative to $A$
    from the perspective of $A$ (indicated by ``$|A$'').
\end{example}
\noindent
In the example, the position observable $\hat X_S^{|A}$
of the system~$S$ relative to $A$ from the perspective of $A$ is not the same
as the position observable $\hat x_S^{|L}$ of $S$ from the perspective of $L$.
In fact, when we refer to \emph{the system $S$ from the
perspective of $A$}, we, strictly speaking, mean a different subsystem 
of $AS$ than the system $S$ defined from
the perspective of $L$.\footnote{
    Unlike in a generic coordinate transformation,
    the parameter $x_A^{|L}$ of the transformation in the example
    is not an externally fixed value, but the position of the physical system $A$.
    We will later see how this fact is crucial when working with QRFs.
    There, the difference between $S$ from the perspectives of either $L$
    and $A$ will be so pronounced that it makes sense to even give them
    different names, ``$S$'' being reserved for the perspective of $L$.
}
Furthermore, the perspective of $A$ is such that the position~$\smash{\hat X^{|A}_S}$
of $S$ relative to $A$ depends only implicitly on the frame $A$.
This is achieved by the relabelling $\smash{\hat x^{|L}_S - x^{|L}_A \hat\id_S}
\rightsquigarrow \smash{\hat X^{|A}_S}$.
Generally, the dependence of relational quantities, like $\hat X^{|A}_S$,
on the reference frame becomes implicit from the perspective of said frame.
We make this property a defining characteristic also of \emph{quantum} reference frame perspectives:
\begin{mybox}
    \phantomsection\label{ppl:perspective}
    \textbf{Perspective Principle:}
    The dependence of relational quantities on the QRF
    becomes implicit from the perspective of the frame.
\end{mybox}
\noindent
Ideas similar to this, albeit formulated in different
ways and often assumed implicitly,
are present in a wide range of QRF frameworks~%
\cite{Aharonov1967,Aharonov1984,Angelo2011,Loveridge2018,Loveridge2018,
Giacomini2019,Hamette2020,Vanrietvelde2020,Hamette2021,Krumm2021,
Glowacki2023,Kabel2025,Carette2025,CastroRuiz2025a,
DeVuyst2025,Ballesteros2025,CastroRuiz2025,Hamette2025}.
We will take example~\ref{exa:classical_A} as a guide to formalize
the \hyperref[ppl:perspective]{perspective principle}.
Specifically, we will generalize translations to a generic symmetry group,
introduce the notion of a quantum system serving as QRF,
define relational quantities in the context of QRFs as so-called \emph{relational POVMs},
and formalize what it means for the QRF to ``become implicit''
in relational quantities.

\medskip
\paragraph*{Symmetry Group and QRFs.}
A crucial feature of example~\ref{exa:classical_A} is that positions transform
under the translation group $(\bbR,+)$ while the position $\smash{x_A^{|L}}$ of the system~$A$
can also be seen as giving rise to the translation by~$\smash{-x_A^{|L}}$.
Generalizing this feature, we will consider a symmetry group $G$
under which~$A$ and $S$ transform under unitary representations
$\smash{\hat U_A}$ and $\smash{\hat U_S}$~\cite{Wigner1959},
acting on $\H_A$ and $\H_S$ respectively.\footnote{
    $G$ is taken to be a finite-dimensional, unimodular Lie group,
    where discrete groups are included as Lie groups of dimension zero,
    and all representations are continuous.
    See appendix~\ref{app:notations_conventions} for the conventions
    and notations we use throughout.
    We require $\hat U_A$ and $\hat U_S$ to be \emph{non-projective},
    which can always be assumed by \emph{centrally extending}
    the symmetry group~\cite{Ragunathan1994,Bargmann1954}.
}
Generalizing positions, we now speak of \emph{orientations};
they are $G$-valued, \ie\ take values in the group $G$,
analogously to how positions can be seen as $(\bbR,+)$-valued.

In example~\ref{exa:classical_A}, $A$ is a \emph{reference frame}
for translations, since its position $x_A^{|L}$ can keep track
of translations acting on $A$~\cite{Aharonov1984,Bartlett2007,Palmer2014,
Loveridge2018,Giacomini2019,Hamette2020,Hamette2021,
CastroRuiz2025a,Glowacki2023,Fewster2024}.
By analogy with the classical case, we
promote the system $A$ together with an \emph{orientation POVM}\footnote{
    \emph{Positive operator-valued measures} (POVMs) generalize projective
    quantum measurements (PVMs)~\cite{Renes2022}.
} $\{\hat\gamma_A^{|L}(g)\}_{g\in G}$ on~$\H_A$
to a \emph{QRF for $G$}. 
The orientation POVM of the QRF $A$ defines the orientation of~$A$ relative to $L$.
The choice of POVM depends on the physical situation one considers.
Acting globally with the symmetry group $G$
changes the orientation of $A$, so the orientation POVM
must be covariant to be a sensible indication of orientation~%
\cite{Bartlett2007,Loveridge2018,Hamette2021,CastroRuiz2025a,AhmadAli2022,
Glowacki2023,Fewster2024}:
\begin{equation}
    \mathsf U_A(g')\bigl[\hat\gamma_A^{|L}(g)\bigr] = \hat\gamma_A^{|L}(g'g)
    \quad \forall\,g,g' \in G,
    \label{eq:orientation_covariance}
\end{equation}
where $\mathsf U_A(g')[\;\cdot\;] := \hat U_A(g') [\;\cdot\;]
\hat U_A^\dagger(g')$.
Indeed, the probability of obtaining the orientation $g \in G$ for $A$
relative to~$L$ should equal the probability of obtaining the orientation
$g'g$ after $G$ has acted with $g' \in G$.
This generalizes the classical position $x_A^{|L}$
from example~\ref{exa:classical_A}
which transforms under global translations.

A QRF $A$ is called \emph{ideal},
if it can perfectly keep track of orientation~%
\cite{Hamette2020,Hamette2021,Glowacki2023a,CastroRuiz2025a,CastroRuiz2025}.
That is, if there exist (possibly improper) states $\{ \smash{\hat\rho_A^{|L}(g)} \}_{g\in G}$
of the QRF such that $\tr\bigl[\hat\gamma_A^{|L}(g) \hat\rho_A^{|L}(g')\bigr] = 0$
if and only if $g \neq g'$. Note that this definition 
only depends on the state space of the QRF $A$ and 
is independent of the specific state of~$A$
in a concrete experiment.
By the covariance condition in \cref{eq:orientation_covariance}
it is enough to find one state $\hat\rho_A^{|L}(e)$
such that  $\tr\bigl[\hat\gamma_A^{|L}(g) \hat\rho_A^{|L}(e)\bigr] =
0$ if $g \neq e$,
since by acting with $\hat U_A$ one obtains states with the
desired property for every other orientation:  
$\hat\rho_A^{|L}(g) := \mathsf U_A(g)[\hat\rho_A^{|L}(e)]$.
An ideal QRF can emulate a classical reference frame:
if we restrict the states of $A$ from the perspective of $L$
to~$\smash{\{\hat\rho_A^{|L}(g)\}_{g \in G}}$,
the orientation of $A$ becomes definite.
A QRF which is not ideal is called \emph{non-ideal}.

\medskip
\paragraph*{Relational POVMs.}
The relational observable in example~\ref{exa:classical_A},
$\smash{\hat x_S^{|L} - x_A^{|L} \hat\id_S}$,
is obtained from $\smash{\hat x_S^{|L}}$
by subtracting~$\smash{x_A^{|L} \hat\id_S}$, compensating for the position of $A$.
Note that this implies the result is invariant under global translations.
We extend this idea to define relational observables in general:
following~\cite{Bartlett2007,Loveridge2018,CastroRuiz2025a,Hamette2021,
Glowacki2023,Fewster2024}, we introduce a \emph{relationalization map}
$\mathsf R_A$ that takes an observable on $S$ and compensates it with the orientation
of the QRF~$A$, producing a relational observable, which is
$G$-invariant by construction.\footnote{
    $\mathsf R_A$ is often called a \emph{relativization map}
    and denoted with monetary symbols ($\$$, $\yen$)~%
    \cite{Bartlett2007,Loveridge2018}.
    Given the origin of two of the authors,
    perhaps we should have denoted the relationalization map ``Fr.''
    We prefer the terms ``relationalization'' and ``relational''
    as opposed to ``relativization'' and ``relative'',
    because the latter often carry the connotation of ``frame-dependent'',
    but the fact that $\mathsf R_A$ is defined from the perspective of $L$ is not relevant here.
}
Generally, we take an observable on $S$ to be a POVM
$\{\hat m_S^{|L}(i)\}_{i \in I}$ with outcomes~$i \in I$,
and the resulting relational observable will be a
POVM $\{\hat M^{|L}_{AS}(i)\}_{i\in I}$ on $AS$ with the same outcomes.
We call these POVMs the \emph{underlying}
and \emph{relational POVMs}, respectively. Thus, we define:
\begin{gather}
    \hat M_{AS}^{|L}(i) = \mathsf R_A\bigl[\hat m_S^{|L}(i)\bigr],
    \label{eq:relational_POVM_element} \\
    \mathsf R_A[\;\cdot\;] := \int_G\der g\, \hat\gamma_A(g) \otimes
    \mathsf U_S(g)[\;\cdot\;].
    \label{eq:relationalization_map}
\end{gather}
Let us see $\mathsf R_A$ at work:
\begin{example}{}{regular_position}
    Let $A$ and $S$ be quantum particles with positions $\hat x_A^{|L}$
    and $\hat x_S^{|L}$, observables which correspond to the PVMs
    $\{\ketbra{x}{x}_A\}_{x\in\bbR}$ and $\{\ketbra{x}{x}_S\}_{x\in\bbR}$.
    We turn $A$ into an ideal QRF for translations, $G = (\bbR,+)$,
    by choosing the former PVM as the orientation POVM of $A$.
    Applying $\mathsf R_A$ to the latter gives
    $\bigl\{ \int_\bbR\der x'\, \sketbra{x'}{x'}_A
    \otimes \sketbra{x+x'}{x+x'}_S \bigr\}_{x\in\bbR}$,
    which is the PVM of to the observable $\hat X^{|L}_{AS} = \hat x_S^{|L} - \hat x_A^{|L}$.
    As expected, the relational quantity is the position difference
    between $S$ and $A$.
    Furthermore, the PVM of the momentum $\hat p_S^{|L}$
    of $S$ is already relational in the sense that
    $\mathsf R_A[\ketbra{p}{p}_S] = \hat\id_A \otimes \ketbra{p}{p}_S$
    for all $p \in \bbR$.
\end{example}

\medskip
Relational POVMs can be understood in two ways. On the one hand,
they can be operationally obtained from measurements of the
underlying POVM with respect to $L$, see
Figure~\ref{fig:relational_POVM}.
This \emph{laboratory-first view} most compatible with quantum-information approaches
such as~\cite{Bartlett2006,Bartlett2007},
considers non-invariant POVM measurements from the perspective of $L$
as primordial and relational POVMs are derived from them.
On the other hand, one can take a \emph{symmetry-first view},
assuming that in the presence of symmetries,
only invariant quantities are observable.
In this case, relational POVMs have fundamental status.
This is the view advocated for in~\cite{Loveridge2017},
where it is shown that the usual, non-invariant quantities emerge from relational POVMs
in the limit of sharply-localised reference frames.
In figure~\ref{fig:swmeasurement},
we illustrate a relational measurement of position in the symmetry-first view.
Our framework is compatible with both interpretations.

\begin{figure}[t]
    \centering
    \includegraphics{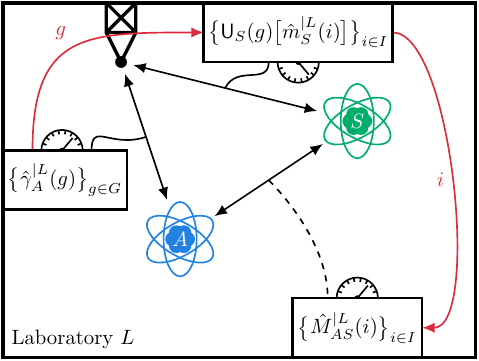}
    \caption{%
        The relational POVM is obtained from the underlying POVM
        by compensating the latter with the orientation of $A$ relative to $L$
        to remove all traces of orientation relative to $L$.
        More precisely, measuring $\{\hat M_{AS}^{|L}(i)\}_{i\in I}$
        can be understood as first measuring the orientation
        $\{\hat\gamma_A^{|L}(g)\}_{g\in G}$ of $A$ relative to $L$, obtaining $g \in G$,
        followed by measuring the orientation-compensated
        POVM $\{\mathsf U_S(g)[\hat m_S^{|L}(i)]\}_{i\in I}$,
        obtaining $i \in I$,
        and finally forgetting $g$ while keeping $i$.
    }
    \label{fig:relational_POVM}
\end{figure}

\begin{figure}[t]
    \centering
    \includegraphics{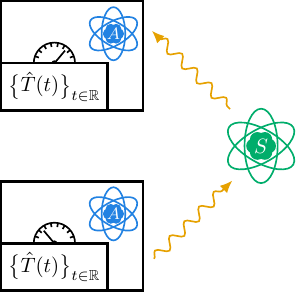}
    \caption{%
        \emph{Salecker-Wigner measurement} of the relative position of $S$ with
        respect to $A$~\cite{Salecker1958}.
        The position QRF $A$ consists of the centre of mass degrees of freedom
        (depicted by the blue atom) of a device capable of emitting
        and detecting a light ray.
        This device can be used to measure the position of $S$
        relative to $A$ by bouncing a light ray off $S$.
        As the light bounces off $S$ and arrives back at the box,
        the measurement of a time POVM $\{\hat T(t)\}_{t\in\bbR}$ on the clock
        determines the light travel time.
        This gives an estimate for the relative position of $S$
        with respect to $A$.
        If $A$ and the clock are ideal, this measurement corresponds
        to the PVM measurement of $\hat X_{AS}^{|L}$ in
        example~\ref{exa:regular_position}.
    }
    \label{fig:swmeasurement}
\end{figure}

Finally, note that we may consider more general operators
than POVM elements as arguments to $\mathsf R_A$.
These can \eg\ be used to model \emph{operations} performed on $S$ in relation to $A$.
We discuss this in sections~\ref{sec:non_ideal_case} and~\ref{sec:conclusion_outlook}.

\medskip
\paragraph*{Making the QRF Implicit.}
The relational observable in 
example~\ref{exa:classical_A},~$\hat x_S^{|L} - x_A^{|L} \hat\id_S$, 
depends only parametrically on the frame
via its classical position.
Introducing the relational position operator
$\hat X_S^{|A} := \hat x_S^{|L} - x_A^{|L} \hat \id_S$
makes this dependence implicit and thus fulfils the
\hyperref[ppl:perspective]{perspective principle}.
More concretely, we defined a new quantum system, on which the relational position
acts as the position operator, by a simple renaming of operators.
In the QRF case, \eg\ example~\ref{exa:regular_position},
we also want to do this, but since $A$ has its own Hilbert space
and a general relational POVM element $\hat M^{|L} = \mathsf R_A[\hat m_S^{|L}]$
acts non-trivially on both $\H_A$ and $\H_S$,\footnote{
    If clear, we will, from now on, occasionally leave out subscripts
    indicating on what Hilbert space factors operators act.
}
renaming operators is not enough to make the dependence on the frame implicit.
In other words, the Hilbert space structure $\H^{|L} := \H_A \otimes \H_S$,
while natural for the description of physics from the perspective of $L$,
is unsuitable for the perspective of $A$.

To implement the \hyperref[ppl:perspective]{perspective principle},
we want to change the Hilbert space structure
such that relational POVM elements act non-trivially
only on their \emph{own minimal subsystem}.\footnote{
    With ``minimal'' we mean that it does not contain a smaller subsystem
    with the same properties.
    Without minimality, one could claim that the problem is trivially solved
    by choosing the total system $AS$ as the subsystem
    where relational POVM elements act non-trivially;
    this is useless as typically many other operators besides
    relational POVM elements act on $AS$,
    including operators of the form $\hat X_A \otimes \hat\id_S$
    attributed to the QRF.
    We will see later that most naturally,
    we need \emph{generalized subsystems},
    \ie\ direct sums of Hilbert space tensor factors~%
    \cite{Barnum2004,Rio2015,Chiribella2018}.
}
Then, the tensor factorization of the Hilbert space 
is naturally adapted to relational POVM elements,
making the dependence on the reference frame implicit.
More precisely, we want to find an isomorphism $\hat V^{\to A} : \H^{|L} \to \H^{|A}$
to a new Hilbert space $\H^{|A}$ equipped with a subsystem structure
such that with respect to this structure,
every transformed relational POVM element
$\hat M^{|A} := \mathsf V^{\to A}[\hat M^{|L}]$
acts non-trivially only on the same minimal subsystem.
Accordingly, we call $\hat V^{\to A} : \H^{|L} \to \H^{|A}$
the \emph{QRF jump} from the perspective of $L$ into that of $A$.

For example, consider the change to relative coordinates
in the context of one-dimensional translations:
\let\fnum\thefootnote
\begin{example}{}{relative_coordinates}
    Take the setup of example~\ref{exa:regular_position}.
    A general pure state in $\H^{|L} = \H_A \otimes \H_S$
    is a wave function $\psi(x_A,x_S)$.
    Introducing $X := x_S - x_A$ we may define
    a new wave function $\Psi(x_A,X)
    := \psi\bigl(x_A,x_A + X\bigr)$.
    The relational operator $\smash{\hat x_S^{|L} - \hat x_A^{|L}}$
    acts on $\Psi$ multiplicatively as $X$.
    Similarly, the other interesting relational operator,
    the momentum of $S$,
    acts on $\Psi$ as~$-\i\partial_X \Psi$.
    We thus see that $\Psi$ is already a description satisfying
    the \hyperref[ppl:perspective]{perspective principle}.

    \medskip
    In our notation, this transformation is implemented by setting
    $\H^{|A} := \H_Y \otimes \H_X$,
    $\H_X,\H_Y \iso L^2(\bbR)$, and
    $\hat V^{\to A} := \int\der x_A\,\der x_S\,
    \sket{x_A}_Y \sbra{x_A}_A \otimes
    \sket{x_S-x_A}_X \sbra{x_S}_S$.
    The subsystem on which relational operators act is~$X$,
    which is minimal.\textsuperscript{\ref{ft:example_3}}
    Hence, the \hyperref[ppl:perspective]{perspective principle}
    is satisfied and $\hat V^{\to A}$ is the QRF jump to $A$.%
    \textsuperscript{\ref{ft:example_3_2}}
\end{example}
\footnotetext{
    \label{ft:example_3}%
    The usual representation of position and momentum on $L^2(\bbR)$
    is irreducible~\cite{Neumann1931}.
    Alternatively, note that all operators on $\H_X$
    are transformed relational operators
    as $(\mathsf V^{\to A})^\dagger[\hat\id_A \otimes (\;\cdot\;)]
    = \mathsf R_A[\;\cdot\;]$.
}
\footnotetext{
    \label{ft:example_3_2}%
    As we will see below,
    it is possible to refine $\hat V^{\to A}$ so that
    besides the \hyperref[ppl:perspective]{perspective principle},
    a second principle is satisfied.
}

\noindent
Notably, example~\ref{exa:classical_A} can be recast in the language
of example~\ref{exa:relative_coordinates}, by restricting in the latter
to states of the classical-quantum form
$\hat\rho_{AS}^{|L} = \sketbra{x_A^{|L}}{x_A^{|L}}_A \otimes \hat\sigma_S$
from the perspective of $L$.
With this it becomes clear that the relabelling
$\hat x^{|L}_S - x^{|L}_A \hat\id_S \rightsquigarrow \hat X^{|A}_S$
and $\hat p^{|L}_S \rightsquigarrow \hat p^{|A}_S$
in example~\ref{exa:classical_A}
is nothing more than a change to a more natural subsystem structure,
like in the QRF case of example~\ref{exa:relative_coordinates}.
Regardless of whether our frame is classical or quantum,
we can implement the \hyperref[ppl:perspective]{perspective principle}
by finding a suitable subsystem structure.

\medskip
\paragraph*{Semi-Simplicity.}
Because $\hat V^{\to A}$ is an isomorphism,
not just relational POVM elements act on their own minimal subsystem
from the perspective of $A$, but the whole algebra they generate.
Our result will thus naturally be stated in terms of operator algebras.
Now, if an operator algebra $\A$ acts irreducibly on a Hilbert space $\H$,
\ie\ as the algebra of all operators on $\H$,
then by definition, there are no other operators that have this property,
making the quantum system defined by the Hilbert space~$\H$
minimal in the required sense.
Most generally, if $\A$ acts \emph{semi-simply} on a given Hilbert space
structure, that is, if it acts \emph{block-diagonally with irreducible blocks}
in the given structure, then the \emph{generalized subsystem}~%
\cite{Barnum2004,Rio2015,Chiribella2018}
defined by the spaces acted-on by the blocks is minimal in the required sense too.
Thus, we are naturally led to consider Hilbert space structures
such that relevant algebras act semi-simply.

Concretely, we care about the operator algebra
\begin{equation}
    \A^{S:A|L} := \O\bigl[ \mathsf R_A[\O_S] \bigr].
    \label{eq:relational_algebra}
\end{equation}
Here and in the following,
$\O_X$ denotes the algebra of all operators on a Hilbert space $\H_X$,
and $\O[\dots]$ is the operator algebra generated by $[\dots]$.\footnote{
    See appendix~\ref{app:notations_conventions} for more details on notation.
}
As before, ``$|L$'' in the superscript indicates that
$\A^{S:A|L}$ is an object described from the perspective of $L$.
Furthermore, ``$S{:}A$'' reads ``$S$ relationalized with respect to the QRF $A$''
and indicates that~$\A^{S:A|L}$ is generated by relational operators.
To implement the \hyperref[ppl:perspective]{perspective principle}
we need to find a unitary jump $\hat V^{\to A} : \H^{|L} \to \H^{|A}$ so that
$\A^{S:A|A} := \mathsf V^{\to A}[\A^{S:A|L}]$
acts semi-simply, that is, block-diagonally with irreducible blocks,
in the Hilbert space structure of~$\H^{|A}$.

\medskip
\paragraph*{$\boldsymbol G$-Invariant Operators.}
Relational operators are invariant under the global action of $G$.
Thus, $\A^{S:A|L} < \mathsf G_{AS}[\O_{AS}]$,
where $\mathsf G_{AS}[\O_{AS}]$ is the subalgebra of all
\emph{$G$-invariant operators}, \ie\ those operators of $\O_{AS}$ which
commute with the global action $\hat U^{|L} := \hat U_A \otimes \hat U_S$ of $G$.\footnote{
    As explained in appendix~\ref{app:representation_theory_groups},
    the projector onto $\mathsf G_{AS}[\O_{AS}]$ is given by the
    \emph{$G$-twirl} super-operator $\mathsf G_{AS}$~\cite{Bartlett2007}
    (hence the notation);
    see particularly proposition~\ref{prop:G_twirl}.
}

From the laboratory's perspective, $G$-invariant operators are the most general operators
which could still be operationally accessible 
without using the laboratory frame~%
\cite{Bartlett2007,CastroRuiz2025a}.%
\footnote{
    Whether any given $\hat O \in \mathsf G_{AS}[\O_{AS}]$ \emph{actually has}
    an operational meaning of course depends on additional assumptions one makes.
}
Consider a setting with two QRFs $A$ and $B$ and a system of interest~$Q$,
Hilbert space $\H^{|L} = \H_A \otimes \H_B \otimes \H_Q$
and global $G$-action $\hat U^{|L} = \hat U_A \otimes \hat U_B \otimes \hat U_Q$.
We define the algebras $\A^{BQ:A|L} = \mathsf R_A[\O_{BQ}]$
and $\A^{AQ:B|L} = \mathsf R_B[\O_{AQ}]$
analogously to the case of a single QRF, treating first $BQ$ and then $AQ$ as $S$.
Generally, the algebra $\A^{BQ:A|L}$ contains operators not contained in
$\A^{AQ:B|L}$ and \viceversa,
but they have in common that both of them are $G$-invariant.
Therefore, $G$-invariant operators
beyond relational operators with respect to a single QRF
are operationally interesting.\footnote{
    For example, the idea that observers corresponding to different QRF perspectives
    can collaborate to effectively access more G-invariant operators
    was introduced in~\cite{Hamette2025,Doat2025}.
    Furthermore,~\cite{Doat2025} showed that the different algebras
    generated by relational operators spans the full algebra of invariant operators
    in certain~cases.
}
The subsystem structure of $L$'s perspective is not adapted to
$G$-invariant operators.
Meanwhile, from $A$'s perspective,
the \hyperref[ppl:perspective]{perspective principle}
only demands the subsystem structure to be adapted to relational operators with respect to $A$
(analogously for other QRF perspectives),
a subset of the invariant operators.
There is \apriori\ no reason why one could not also require the structure
to be adapted to all $G$-invariant operators;
and because of their importance,
we assume it as a second principle:
\begin{mybox}
    \phantomsection\label{ppl:invariance}
    \textbf{Invariance Principle:}
    From a QRF perspective,
    $G$-invariant operators are represented naturally.
\end{mybox}
\noindent
Assuming that the QRF perspective of $A$ describes quantum operations
available to an observer with access to $A$ but not the laboratory,
this principle is natural, because $G$-invariant operators are
the largest class of operators that are, in principle, accessible to
such an observer.

 Analogous to the \hyperref[ppl:perspective]{perspective principle},
we formalize the \hyperref[ppl:invariance]{invariance principle} into the mathematical requirement that
$\mathsf V^{\to A}\bigl[\mathsf G_{AS}[\O_{AS}]\bigr]$ acts semi-simply
in the subsystem structure of $A$'s perspective.
Note that, although both principles formalize to similar mathematical statements,
they are physically quite distinct:
the \hyperref[ppl:perspective]{perspective principle}
was motivated from intuition about classical reference frames,
while the \hyperref[ppl:invariance]{invariance principle}
required further considerations regarding operator algebras.

We will require both principles to hold in QRF perspectives, thereby
defining the notion of a perspective. We will see in theorem \ref{thm:QRF_POV}
that both can indeed be implemented \emph{simultaneously} and \emph{independently}
of each other.

\medskip
\paragraph*{General Result.}
To state theorem~\ref{thm:QRF_POV}, it is useful to introduce the
additional algebras
\begin{align}
    \A^{\overline{S:A}|L} &:= \bigl\{ \hat X \in \mathsf G_{AS}[\O_{AS}]
    \st [\hat X,\A^{S:A|L}] = 0 \bigr\},
    \label{eq:commutant_algebra} \\
    \Z^{S:A|L} &:= \A^{S:A|L} \cap \A^{\overline{S:A}|L}.
    \label{eq:centre_algebra}
\end{align}
Here, $\smash{\A^{\overline{S:A}|L}}$ is the $G$-invariant \emph{commutant} of $\smash{\A^{S:A|L}}$
and $\Z^{S:A|L}$ is its centre.
If both principles are satisfied, then $\A^{\overline{S:A}|A}
:= \mathsf V^{\to A}[\A^{\overline{S:A}|L}]$
and $\Z^{S:A|A} := \mathsf V^{\to A}[\Z^{S:A|L}]$
also act semi-simply.
\begin{theorem}{Perspective of a QRF}{QRF_POV}
    \allowdisplaybreaks
    There exists an isomorphism $\hat V^{\to A}$ such that
    the principles of \hyperref[ppl:perspective]{perspective}
    and \hyperref[ppl:invariance]{invariance} are satisfied.
    Concretely,
    \begin{align}
        \H^{|L} &\iso \H^{|A} := \hat V^{\to A} \H^{|L} \notag\\
        &= \bigoplus_{k,r} \H_r \otimes \H_{r,k} \otimes \H_k,
        \label{eq:QRF_view_Hilbert_space_decomposition} \\
        \A^{S:A|L} &\iso \A^{S:A|A}
        := \mathsf V^{\to A}\bigl[ \A^{S:A|L} \bigr] \notag\\
        &= \bigoplus_k\Bigl(\bigoplus_r \hat\id_r \otimes
        \hat\id_{r,k}\Bigr) \otimes \O_k,
        \label{eq:QRF_view_relational_algebra_decomposition} \\
        \A^{\overline{S:A}|L} &\iso \A^{\overline{S:A}|A}
        := \mathsf V^{\to A}\bigl[ \A^{\overline{S:A}|L} \bigr] \notag\\
        &= \bigoplus_{k,r} \hat\id_r \otimes
        \O_{r,k} \otimes \hat\id_k,
        \label{eq:QRF_view_commutant_decomposition} \\
        \hat U^{|L}(g) &\iso \hat U^{|A}(g)
        := \mathsf V^{\to A} \bigl[ \hat U^{|L}(g) \bigr] \notag\\
        &= \bigoplus_{k,r} \hat U_r(g) \otimes
        \hat\id_{r,k} \otimes \hat\id_k,
        \label{eq:QRF_view_G_decomposition}
    \end{align}
    where some $\H_{r,k}$ may be zero-dimensional and
    \begin{enumerate}[leftmargin=1.5em]
        \medskip
        \item $r$ is the \emph{charge} associated to the global $G$-action;
            it enumerates all irreducible representations of $G$,
            acting as $\hat U_r$ on $\H_r$.

        \medskip
        \item $k$ enumerates the joint eigenspaces of the elements
        in the centre $\Z^{S:A|L}$ of $\A^{S:A|L}$,
        and $\bigoplus_r \H_r \otimes \H_{r,k} \otimes \H_k$
        are those eigenspaces.
    \end{enumerate}
\end{theorem}
\noindent
The mathematical form of the decomposition in theorem~\ref{thm:QRF_POV}
was first proven in~\cite{Bianchi2024} for general invariant algebras.
In this paper, we focus on the algebra generated by relation operators,
defined in \cref{eq:relational_algebra},
motivated by our principles of \hyperref[ppl:perspective]{perspective}
and \hyperref[ppl:invariance]{invariance}.
Similar decompositions are used \eg\ in quantum information~%
\mbox{\cite{Zanardi2001,Zanardi2004,Knill2000,Bianchi2024,
Vanrietvelde2025,Ormrod2025}}
and lattice gauge theory~%
\cite{Casini2014,VanAcoleyen2016,Huang2020}.
A simpler version of theorem~\ref{thm:QRF_POV}
is used in~\cite{CastroRuiz2025a} to define the points of view of
a special class of ideal QRFs
(see section~\ref{sec:extreme_cases}).

The algebras in
\cref{eq:relational_algebra,eq:commutant_algebra,eq:centre_algebra}
are well-defined and the theorem holds rigorously if
$\dim \H^{|L} < \infty$ and if $G$ is compact;
this is the case discussed in~\cite{Bianchi2024},
and we provide our own proof for the theorem in this setting in
appendix~\ref{app:representation_theory_algebras}.
There, we also discuss the mathematical subtleties arising more generally.
Essentially, there is no conceptual barrier to generalizing
the algebraic structure of the decompositions in theorem~\ref{thm:QRF_POV}
to infinite-dimensional Hilbert spaces and non-compact symmetry groups,
if direct sums are appropriately replaced by direct integrals.
One obtains an analogous structure like in the finite, compact case.
For the physical questions we will consider,
only this algebraic structure matters,
and hence the insights from the finite, compact case are sufficient.
In this work, we do not provide a rigorous mathematical treatment
of the generalizations but sketch a rough roadmap for it in
appendix~\ref{app:representation_theory_algebras}.

Theorem~\ref{thm:QRF_POV} fulfils the
\hyperref[ppl:perspective]{perspective principle}:
transformed relativized POVM elements $\hat M^{|A}$ act non-trivially
only on the minimal subsystem represented by the Hilbert spaces $\H_k$,
because $\A^{S:A|A}$ acts semi-simply,
\ie\ block-diagonally on these spaces,
with the algebras of all operators as blocks.\footnote{
    Note the necessity of the brackets around the direct sum over $r$
    in equation~\eqref{eq:QRF_view_relational_algebra_decomposition}.
}
We denote this subsystem by~$S{:}A|A$.
Generically, it is a \emph{generalized subsystem}~\mbox{\cite{Barnum2004,Rio2015,Chiribella2018}},
consisting of a direct sum or integral of tensor factors.
The theorem also implements the
\hyperref[ppl:invariance]{invariance principle}:
$G$ acts semi-simply on the generalized subsystem $\Gamma|A$
defined by the Hilbert spaces $\H_r$,
and hence any operator which commutes with the global $G$-action
must be of the form $\bigoplus_r \hat\id_r \otimes \hat X_{*,r}$,
where $\hat X_{*,r}$ is an arbitrary operator on $\bigoplus_k \H_{r,k} \otimes \H_k$.
These latter spaces are those on which $\mathsf V^{\to A}[\mathsf G_{AS}[\O_{AS}]]$
acts semi-simply.
Furthermore, $\A^{\overline{S:A}|A}$ acts semi-simply on the subsystem
$\overline{S{:}A}|A$ defined by the Hilbert spaces $\H_{r,k}$.
We call this subsystem the \emph{extra particle}
as it generalizes the corresponding concept from \cite{CastroRuiz2025a}.\footnote{
    The name ``extra particle'' has been introduced in~\cite{CastroRuiz2025a},
    because for \emph{regular} QRFs (a kind of ideal QRF, see section~\ref{sec:extreme_cases})
    and for the centrally extended Galilei group,
    $\A^{\overline{S:A}|A}$ resembles the algebra
    of a quantum particle with variable but superselected mass.
}
We now see why introducing the commutant was necessary:
from $A$'s perspective, the algebra of $G$-invariant operators
is generally larger than $\smash{\A^{S:A|A}}$, including also $\smash{\A^{\overline{S:A}|A}}$.
Note that $\smash{\A^{S:A|A}}$ and  $\smash{\A^{\overline{S:A}|A}}$ \apriori\
do not generate the algebra of all $G$-invariant operators;
we will see, however, that this is the case in some situations.
We will often compare the subsystem $\smash{S{:}A|A}$
to the subsystem $\smash{S}$ as described from the perspective of~$L$;
for simplicity, we label the latter with~$S|L$.
Analogously, $A|L$ denotes the QRF as described from $L$.

If one is only interested in the
\hyperref[ppl:perspective]{perspective principle},
then the decomposition over $r$ as well as the distinction between
subsystems $\Gamma|A$ and $\overline{S{:}A}|A$ associated with it can be ignored,
since the remaining subsystem structure still makes
relational operators act semi-simply;
this could \eg\ be the case if, for some physical reason, only relational
operators are operationally interesting.
Similarly, if one only cares about the
\hyperref[ppl:invariance]{invariance principle}
then the decomposition over $k$ and the associated subsystem decomposition can be ignored;
in the context of QRFs this setting is arguably physically less natural,
because it effectively ignores the QRF entirely.
In this sense, the two principles are independent.\footnote{
    Such a decomposition with two independent parameters occurs generally
    when semi-simply representing an operator algebra,
    which is invariant under some symmetry group~\cite{Bianchi2024};
    see also proposition~\ref{prop:algebra_decomposition_full}
    in appendix~\ref{app:representation_theory_algebras}.
}

Note that the perspective of $A$ is fixed by theorem~\ref{thm:QRF_POV}
only up to basis changes in $\H_r$, $\H_{r,k}$ and $\H_k$.
Depending on the physical context, one choice of basis may be preferred.
\Eg, in example~\ref{exa:relative_coordinates}
it makes sense to require the relative position $\hat x_S^{|L} - \hat x_A^{|L}$
to act like a position operator on the subsystem $X = S{:}A|A$;
that is, multiplicatively on position-space wave functions.\footnote{
    However, this does not completely fix the action of the momentum operator $p^{|A}_S$
    as there can be different operators canonically commuting with
    $\smash{\hat x^{|L}_S - \hat x^{|L}_A}$~%
    \cite{Peres2002,Hamette2025,Doat2025,CastroRuiz2025}.
}

\medskip
\paragraph*{QRF Transformations.}
Consider again two QRFs, $A$ and $B$, observing a system $Q$.
We can obtain jumps $\hat V^{\to A} : \H^{|L} \to \H^{|A}$ and
$\hat V^{\to B} : \H^{|L} \to \H^{|B}$
into the perspectives of $A$ and $B$
from theorem~\ref{thm:QRF_POV}
by treating $BQ$ and $AQ$ as $S$, respectively.

We can chain jumps together, obtaining isomorphisms
\begin{align}
    \hat V^{A\to B} &:= \hat V^{\to B} (\hat V^{\to A})^\dagger
    : \H^{|A} \to \H^{|B}, \label{eq:qrfatob} \\
    \hat V^{B\to A} &:= \hat V^{\to A} (\hat V^{\to B})^\dagger
    : \H^{|B} \to \H^{|A},\label{eq:qrfbtoa} 
\end{align}
and $\hat V^{B\to A} = (\hat V^{A\to B})^\dagger$.
Because $\hat V^{A\to B}$ and $\hat V^{B\to A}$
transform from the perspective of $A$ to that of $B$
and \viceversa\ respectively, they are called \emph{QRF transformations}~%
\cite{Aharonov1984,Merriam2005,Angelo2011,Hamette2020,
Vanrietvelde2020,Glowacki2023,Carette2025,
Hamette2021,Krumm2021,Kabel2025,CastroRuiz2025a,DeVuyst2025,
Ballesteros2025,CastroRuiz2025,Hamette2025}.
They are defined up to unitaries on the relevant subsystems,
since this is already the case for the jumps.

Note that no assumption was made regarding the nature of the reference frame;
the QRF transformation is defined in general.
We will see in section~\ref{sec:extreme_cases}
that we recover the QRF transformation of the so-called
\emph{extra-particle framework}~\cite{CastroRuiz2025a}
for the special class of ideal QRFs considered by them;
our framework thus constitutes a generalization of that framework.
As we show in appendix~\ref{app:comparisons},
mathematically, our framework becomes the perspective-neutral framework
of~\cite{Hamette2021} when restricted to the subspace
where global $G$-transformations act with the trivial representation.

Since different algebras act semi-simply in either frame
(\eg, $\smash{\A^{BQ:A|L} \neq \A^{AQ:B|L}}$ in general),
the subsystem structures of either perspective are generally different.
Furthermore, they are generally different from the subsystem structure
from the perspective of $L$.
This phenomenon is known as \emph{subsystem relativity}~%
\cite{CastroRuiz2025a,AhmadAli2022,Hoehn2023a},
which can lead to \emph{relativity of entanglement and coherence}~%
\cite{Giacomini2019,Hamette2020,Savi2021,Cepollaro2025}.
Related to this, note that~$BQ{:}A|A$ does not generally decompose
into subsystems $B{:}A|A$ and~$Q{:}A|A$,
and similarly from the perspective of $B$.
The subsystems $\Gamma|A$ and $\Gamma|B$ on the other hand
are always isomorphic and furthermore
are mapped onto each other by QRF transformations,
since they both encode the action of $G$ relative to the laboratory.
From the point of view of an observer who cannot use
the laboratory frame, these degrees of freedom are not physically
accessible and may be seen as \emph{gauge}~\cite{CastroRuiz2025a}.

Consequently, the QRF transformation maps $G$-invariant
degrees of freedom onto themselves.
That is, the subsystems on which $G$-invariant operators act,
\ie\ the subsystems obtained by combining $BQ{:}A|A$ with $\overline{BQ{:}A}|A$
and $AQ{:}B|B$ with $\overline{AQ{:}B}|B$,
are isomorphic and mapped into each other by the QRF transformations.
Thus, the QRF transformations $\smash{\hat V^{A\to B}}$ and $\smash{\hat V^{B\to A}}$
can be restricted to the subsystems acted on by $G$-invariant operators and remain invertible,
despite subsystem relativity.

Crucially, in general, the QRF transformation does \emph{not} map
$BQ{:}A|A$ to $AQ{:}B|B$.
This even happens if both QRFs are ideal, as we will see below.
To obtain invertible QRF transformations defined at least on
the subsystems $BQ{:}A|A$ and $AQ{:}B|B$,
\emph{one is forced to also consider the subsystems
$\overline{BQ{:}A}|A$ and $\overline{AQ{:}B}|B$.} 
This is consistent with the observation 
that the extra particle is essential 
for the invertibility of general QRF transformations~\cite{CastroRuiz2025a}.

\section{Extreme Cases}
\label{sec:extreme_cases}

\noindent
Having defined the points of view of QRFs in broad generality
using the \hyperref[ppl:perspective]{perspective}
and \hyperref[ppl:invariance]{invariance principle}, 
which we implemented in theorem~\ref{thm:QRF_POV},
we now explore the properties of such perspectives.
First, we consider the extreme cases of \emph{trivial} and \emph{regular QRFs}.

\medskip
\paragraph*{Trivial QRFs.}
The simplest possible QRF is a trivial one:
$\H_A = \bbC$, $\hat U_A = 1_A \in \bbC$ the trivial representation,
and $\smash{\hat\gamma_A^{|L}(g) = 1/|G|}$ for all $g \in G$.
Such a frame is non-ideal (unless $G = \{e\}$).
In fact, it cannot keep track of orientation at all,
and might as well be completely absent since $\H^{|L} \iso \H_S$.
One easily finds $\A^{S:A|L} = 1_A \cdot \mathsf G_S[\O_S]$
and $\smash{\A^{\overline{S:A}|L} = \Z^{S:A|L}}$.
The decomposition of theorem~\ref{thm:QRF_POV} can be obtained from
the structure of the $G$-twirl in proposition~\ref{prop:G_twirl}
of appendix~\ref{app:representation_theory_groups}:
\begin{equation}
    \H^{|L} \iso \H^{|A} = \bigoplus_k \H_{r=k} \otimes \H_k,
\end{equation}
where $k$ ranges over the inequivalent irreducible representations contained in
$\hat U^{|L} \iso \hat U_S$, $\hat U^{|A}$ acts semi-simply
on the spaces $\H_{r=k}$, and $\A^{S:A|A}$
acts semi-simply on the spaces $\H_k$.
Jumping into $A$'s perspective simply consists of finding a basis
adapted to this decomposition.

Essentially, $S{:}A|A$ appears from $A$'s perspective
as $S|L$ does from that of $L$,
but \emph{superselected} by the charge $r=k$ of the representation.
This is the known superselection observed in the
\emph{absence of a reference frame}~\cite{Bartlett2007},
confirming the intuition that using a trivial QRF is equivalent
to using no frame at all.\footnote{
    More generally, a trivial QRF can be defined as
    $\hat U_A = \hat\id_A$, $\hat\gamma_A^{|L} = \hat\id_A/|G|$,
    but with no restrictions on the dimension of $\H_A$.
    In that case, $\H^{|A} = \bigoplus_k \H_{r=k} \otimes \H'_k \otimes \H_k$,
    where the spaces $\H_k$ and $\H_{r=k}$ are as in the simple case,
    but now $\H'_k \iso \H_A$ and $\A^{\overline{S:A}|A}$ acts semi-simply
    on the spaces $\H'_k$.
    $S{:}A|A$ appears superselected in the same way as before,
    but $\overline{S{:}A}|A$ is bigger.
}

\medskip
\paragraph*{Regular QRFs.}
A specific type of ideal QRF is a \emph{regular QRF}:
$\H_A = L^2(G)$ are wave functions on the group manifold,
$\hat U_A^{|L} = \hat L$ is the left-regular representation
and $\hat\gamma_A^{|L}(g) = \ketbra{g}{g}_A$ are rank-one (possibly improper)
projectors obtained from Dirac distributions on~$G$.
The (possibly improper) states $\ket{g}$ are orthogonal:
$\sbraket{g'}{g} = \delta(g'^{-1}g)$.\footnote{
    Here, $g \mapsto \delta(g'^{-1}g)$ is the Dirac-distribution
    on the group manifold $G$, centred on $g' \in G$.
    See appendix~\ref{app:regular} for more details on regular QRFs.
}
Regular QRFs have been widely used in the literature~%
\cite{Aharonov1984,Angelo2011,Giacomini2019,Vanrietvelde2020,Hamette2020,
Hamette2021,Glowacki2023a,CastroRuiz2025a,CastroRuiz2025}.
They are practical due to their mathematical simplicity
and because they are the most straightforward quantum extension
of classical frames: a classical frame has definite
and perfectly distinguishable orientations $\smash{g_A^{|L} \in G}$,
and when one further allows for quantum superpositions,
a regular QRF is obtained.
The QRFs of examples~\ref{exa:regular_position}
and~\ref{exa:relative_coordinates} are regular.

Regular QRFs for infinite groups have infinite resources
in the following sense:
according to the (generalized) \emph{Peter-Weyl theorem},
proposition~\ref{prop:Peter_Weyl}
in appendix~\ref{app:representation_theory_groups},
the group action on a regular QRFs contains all irreducible representations.
For example, a regular QRF for $\SU(2)$ would
contain arbitrarily high angular momentum.
As another example, for a regular QRF for $\gpU(1)$ serving as a periodic clock
the representation charge is energy~\cite{Peres1980},
which would need to be infinitely spread.

Jumps into regular QRFs and hence also the perspective of such frames
have been extensively described in~\cite{CastroRuiz2025a},
so we just briefly outline how their results
derive from our more general theorem~\ref{thm:QRF_POV}.
Central to the result is the observation
(see proposition~\ref{prop:jump} in appendix~\ref{app:regular})
that the unitary map $\hat W : \H_A \otimes \H_S \to \H_A \otimes \H_S$
defined as
\begin{equation}
    \hat W := \int_G \der g\, \sketbra{g}{g}_A \otimes \hat U_S(g)
    \label{eq:regular_inverse_jump_core}
\end{equation}
induces the algebra isomorphism relations
\begin{align}
    \hat\id_A \otimes \O_S &\iso \A^{S:A|L},
    \label{eq:regular_OS_A_isomorphism} \\
    \mathsf G_A[\O_A] \otimes \hat\id_S &\iso \A^{\overline{S:A}|L}, \\
    \hat U_A \otimes \hat\id_S &\iso \hat U_A \otimes \hat U_S.
\end{align}
Using propositions~\ref{prop:Peter_Weyl} and~\ref{prop:G_twirl}
from appendix~\ref{app:representation_theory_groups},
we furthermore obtain an isomorphism
$\hat X_A : \H_A \to \bigoplus_r \H_r \otimes \H_{*,r}$,
inducing the isomorphism relation $\mathsf G_A[\O_A] \iso
\bigoplus_r \hat\id_r \otimes \O_{*,r}$,
and $\hat U_A \iso \bigoplus_r \hat U_r \otimes \hat\id_{*,r}$,
where $r$ enumerates the inequivalent irreducible representations
$\hat U_r$ of $G$ and $\H_{*,r} \iso \H_r$.\footnote{
    Here, ``$*$'' is simply a decoration to distinguish the two factors.
     However, although unimportant for us presently,
     $\H_{*,r}$ is also isomorphic to the complex conjugate of $\H_r$,
     see proposition~\ref{prop:Peter_Weyl}.
}
Thus, the isomorphism $\hat V^{\to A} := (\hat X_A \otimes \hat\id_S)\,
\hat W^\dagger$ induces the isomorphism relations
\begin{alignat}{3}
    \H^{|L} &\iso
    \Bigl(\bigoplus\nolimits_r \H_r \,&\otimes\; \H_{*,r}\Bigr)
    &\otimes \H_S, \label{eq:QRF_view_Hilbert_space_decomposition_ideal}\\
    \A^{S:A|L} &\iso
    \Bigl(\bigoplus\nolimits_r \hat\id_r \,&\otimes\; \hat\id_{*,r}\Bigr)
    &\otimes \O_S, \label{eq:QRF_view_relational_algebra_decomposition_ideal}\\
    \A^{\overline{S:A}|L} &\iso
    \Bigl(\bigoplus\nolimits_r \hat\id_r \,&\otimes\; \O_{*,r}\Bigr)
    &\otimes \hat\id_S, \label{eq:QRF_view_commutant_decomposition_ideal}\\
    \hat U^{|L} &\iso
    \Bigl(\bigoplus\nolimits_r \hat U_r \,&\otimes\; \hat\id_{*,r}\Bigr)
    &\otimes \hat\id_S. \label{eq:QRF_view_G_decomposition_ideal}
\end{alignat}
Clearly, this is of the form described by theorem~\ref{thm:QRF_POV}.
Hence, $\hat V^{\to A}$ is a jump into the perspective of $A$.
Note that $\A^{\overline{S:A}|L}$ is non-trivial and generates the
algebra of $G$-invariant operators together with $\A^{S:A|L}$.

Note also that there is no sum over $k$,
indicating a trivial centre $\Z^{S:A|L}$.
Consequently, $S{:}A|A$ is described by the same Hilbert space $\H_S$
as the subsystem $S|L$, despite representing different physical degrees of freedom.
Essentially, the \emph{structure} of the subsystem $S{:}A|A$ is identical
to that of $S|L$.
For example, if $S$ is a quantum particle described from $L$,
then its relational degrees of freedom with respect to the QRF $A$,
will ``look like'' those of a particle also from the perspective of $A$.
This is a consistency result:
since regular QRFs in definite-orientation states~$\ket{g}$ are quantum
realizations of classical frames,
and since any classical reference frame is equivalent to the laboratory frame,
one expects that regular QRFs describe $S$ in the same way as $L$ does.
Particularly, the subsystem structure of $S$ is recovered in $S{:}A|A$:
if $S|L$ consists of two subsystems $S_1$ and $S_2$,
\ie\ $\H_S = \H_{S_1} \otimes \H_{S_2}$,
then $S{:}A|A$ also splits into two subsystems,
isomorphic to $S_1$ and $S_2$ respectively
(see proposition~\ref{prop:regular_subsystem_structure}
in appendix~\ref{app:regular}).

As we stated above, $\smash{\hat V^{\to A}}$ is unique only up to basis changes
within each tensor factor of the decompositions in theorem~\ref{thm:QRF_POV}.
The choice we made here, particularly, the definition of $\smash{\hat W}$ in equation 
\eqref{eq:regular_inverse_jump_core} and the isomorphism $\hat X_A$ from proposition~\ref{prop:Peter_Weyl}, leads to
\begin{equation}
    \hat V^{\to A} \ket{g}_A \ket{\psi}_S
    = \hat X_A \ket{g}_A \hat U^\dagger_S(g) \ket{\psi}_S.
\end{equation}
Essentially, $\hat V^{\to A}$ acts like a classical reference frame
transformation when acting on states with definite orientation 
relative to $L$, as expected.

Let us consider two regular QRFs $A$ and $B$, and a system of interest $Q$.
If we define jumps $\hat V^{\to A}$ and $\hat V^{\to B}$
as described above, and such that $\hat V^{\to B}
= \mathsf T_{AB}[\hat V^{\to A}]$, where $\mathsf T_{AB}$
exchanges the labels $A$ and $B$,
then the QRF transformation from the perspective of $A$ to that of $B$ is
\begin{multline}
    \hat V^{A\to B} :
    \Bigl(\bigoplus_r \H_r \otimes \H_{*,r} \Bigr) \otimes \H_B \otimes \H_Q \\
    \longrightarrow
    \Bigl(\bigoplus_r \H_r \otimes \H_{*,r} \Bigr) \otimes \H_A \otimes \H_Q,
\end{multline}\vspace{-1em}
\begin{multline}
    \hat V^{A\to B} = \int_G\der g\,
    \Bigl( \bigoplus_r \hat\id_r \otimes \hat U^\dagger_{*,r}(g) \Bigr) \\
    \otimes \sket{g^{-1}}_A \sbra{g}_B \otimes \hat U^\dagger_Q(g)
    \label{eq:regular_QRF_transformation}
\end{multline}
(see proposition~\ref{prop:regular_QRF_transformation}
in appendix~\ref{app:regular}).
Up to relabelling of subsystems, this is the QRF transformation
of~\cite{CastroRuiz2025a},
and there it was shown that $\hat V^{A\to B}$ does not map
$BQ{:}A|A$ to $AQ{:}B|B$,
forcing us to also consider the degrees of freedom in the extra particles.
Note that when acting on states where $B{:}A|A$ has a definite
orientation, $\hat V^{A \to B}$ becomes a classical reference frame transformation.

\medskip
\section{Non-Ideal General Case}
\label{sec:non_ideal_case}

\noindent
Having studied the extreme cases, we now turn to the cases in-between:
a non-ideal but non-trivial QRF.
Although ideal QRFs are a very useful idealization,
non-ideal QRFs are often more realistic
since they tend to have finite resources~%
\cite{Fleming1973,Peres1980,Mandelstam1991,Bartlett2007,Loveridge2018}.
For example, an ideal QRF must have as many perfectly distinguishable
states as there are elements in $G$,
because measuring the orientation perfectly distinguishes the states
$\{\hat\rho_A^{|L}(g)\}_{g\in G}$.
Any finite-dimensional QRF for infinite $G$ is, thus, non-ideal.
And even when allowing infinite dimensions,
as we have seen for regular QRFs, ideal QRFs for uncountable $G$ 
require highly unphysical improper states.
This was the case in examples~\ref{exa:regular_position}
and~\ref{exa:relative_coordinates}, which required position eigenstates.
Let us consider a typical non-ideal QRF:
\begin{example}{}{Galilei}
    A quantum particle $\H_A = L^2(\bbR)$ of mass $m$ transforms under
    the one-dimensional Galilei group $\group{Gal}$
    (translations by $a$ and Galilei boosts by $v$)
    with the mass-$m$ representation
    $\hat U_A(a,v) = \ex^{-\i a \hat p_A + \i v m \hat x_A - \i v t\hat p_A}$~%
    \cite{Bargmann1954,Ballentine2014}.
    Due to Heisenberg uncertainty $\Delta x \Delta \dot x \geq 1/2m$
    (assuming $p = m\dot x$),
    $A$ makes at most a non-ideal QRF for the Galilei group.
    In the limit $m \to \infty$, we can obtain an ideal QRF.%
    \textsuperscript{\ref{ft:Galilei}}
\end{example}
\footnotetext{
    \label{ft:Galilei}%
    Since the mass-$m$ representation is projective,
    we must strictly speaking consider the
    \emph{centrally extended Galilei group} $\group{CGal}$
    to obtain unitary representations.
    See appendix~\ref{app:Galilei} for more details,
    including a derivation of this example.
}

\medskip
\paragraph*{Range of Superselection.}
We saw how a trivial QRF~$A$, which is equivalent to not using a frame at all,
is only able to resolve the $G$-invariant degrees of freedom in~$S$,
resulting in a maximal superselection of the subsystem $S{:}A|A$.
Meanwhile, for a regular QRF, $k$ 
takes on a single value in theorem~\ref{thm:QRF_POV},
and hence there is no superselection in the subsystem $S{:}A|A$;
in fact, it is even isomorphic to $S|L$.

In general, the sum over $k$ is non-trivial,
and we conclude that \emph{non-ideality of the QRF $A$ leads to
some amount of superselection of $S{:}A|A$}.
Between the two previously known extremes of trivial and regular QRFs~%
\mbox{\cite{Aharonov1967,Bartlett2007,CastroRuiz2025a}},
we now discover an entire \emph{range of superselection
due to non-ideality}.
Note that this superselection is a feature of the structure
of $S{:}A|A$ and independent of the particular total quantum state.

\medskip
\medskip
\paragraph*{Role of the Relationalization Map.}
The structure of $A$'s perspective depends on the algebra structure
of $\A^{S:A|L}$, which, according to \cref{eq:relational_algebra},
is, in turn, derived from the relationalization map $\mathsf R_A$.
Although most features of $A$'s perspective are not obvious
when considering $\mathsf R_A$ directly,
some conclusions are still possible.

For example, $S{:}A|A$ is isomorphic to $S|L$ if and only if~$\mathsf R_A$
is an algebra isomorphism; this is the case for regular QRFs.
In the non-ideal case however, $\mathsf R_A$ might not even be an algebra homomorphism,
that is,
\begin{equation}
    \mathsf R_A[\hat A_S] \mathsf R_A[\hat B_S]
    \neq \mathsf R_A[\hat A_S \hat B_S]
    \label{eq:R_no_homomorphism}
\end{equation}
for some $\hat A_S, \hat B_S \in \O_S$.
In fact, it is easy to show that equality in \cref{eq:R_no_homomorphism}
for all possible systems $S$ and all $\hat A_S, \hat B_S \in \O_S$
requires $\hat\gamma_A^{|L}(g) \hat\gamma_A^{|L}(g') =
\delta(g'^{-1}g) \hat\gamma_A^{|L}(g)$
for all $g,g' \in G$,\footnote{
    Here, $\delta$ is the Dirac distribution on the group manifold $G$;
    see appendix~\ref{app:representation_theory_groups}.
    Explicitly, one can see the requirement by taking $S$
    to carry the left-regular representation.
    Then, for $g,g' \in G$ with $g \neq g'$,
    $\mathsf R_A[\ketbra{g}{g}_S] \mathsf R_A[\sketbra{g'}{g'}_S]
    = \mathsf R_A[\sket{g}\!\sbraket{g}{g'}\!\sbra{g'}]
    = \mathsf R_A[0] = 0$ only if $\hat\gamma_A^{|L}(g)\hat\gamma_A^{|L}(g')
    = 0$.
}
which implies that $A$ must be ideal.

\Cref{eq:R_no_homomorphism} implies that $[\hat A_S, \hat B_S] = 0$
does not generally lead to $\bigl[ \mathsf R_A[\hat A_S], \mathsf R_A[\hat B_S] \bigr] = 0$
(see proposition~\ref{prop:algebra_decomposition_S_Abelian} in appendix~\ref{app:POVs_Abelian}
for a complete characterization of this phenomenon for Abelian symmetry groups).
This illustrates how ``the system as seen by $A$'', the subsystem $S{:}A|A$,
may be quite different from ``the system as seen by $L$'', the subsystem $S|L$.

\medskip
\paragraph*{Accessible Degrees of Freedom.}
If $A$ is non-ideal, it may be the case that
\begin{equation}
    \mathsf R_A[\O_S] \subsetneqq \A^{S:A|L}
    \label{eq:back_reaction_condition}
\end{equation}
because $\mathsf R_A[\O_S]$ is not guaranteed to be an algebra;
\cref{eq:R_no_homomorphism} is necessary for this.
Essentially, \cref{eq:back_reaction_condition} means that not every degree of freedom
in the subsystem $S{:}A|A$ is a relational operator, and hence describable
operationally in the sense of figures~\ref{fig:relational_POVM} or~\ref{fig:swmeasurement}.
We will see this situation explicitly with example~\ref{exa:qubit_qutrit}
in section~\ref{sec:POVs_Abelian}.

Notably, an operational interpretation of \cref{eq:back_reaction_condition},
extending that of figure~\ref{fig:relational_POVM}, is conceivable
if we interpret operators in $\mathcal{A}^{S:A|L}$,
including those that are not of the relational form,
as Kraus operators corresponding to quantum operations.
This is consistent with the compatibility of the notions of subsystem
from the point of view of operator algebras and of quantum operations
pointed out in~\cite{Chiribella2018}.
Because the composition of quantum operations implies the multiplication of Kraus operators,
operators in $\mathcal{A}^{S:A|L} \backslash \mathsf R_A[\O_S]$
acquire a natural interpretation in terms of successive operations
performed on the system $S$ relative to the QRF $A$.\footnote{
    Since the Kraus decomposition of a channel is not unique,
    since Kraus operators are not directly measurable quantities,
    and since the operational meaning of channel composition is subtle
    in itself as it relies on independence assumptions~\cite{Finetti1970},
    the operational meaning of Kraus operators is less immediate than
    that of relative POVM elements.
}
This line of thinking provides strong evidence
that in some contexts it may be appropriate to take
the whole algebra $\A^{S:A|L}$ as operationally relevant,
as opposed to just $\mathsf R_A[\O_S]$.
That composition of relational operations leads to a larger set of operations,
\ie\ \cref{eq:back_reaction_condition},
can then be interpreted as a consequence of the \emph{back-reaction}
of $S$ onto the QRF $A$ as described from $L$'s perspective:
intuitively, because $A$ is non-ideal, it will be disturbed by
operations performed on other systems relative to it,
and hence a subsequent relational operation can detect effects of this back-reaction;
this effect only occurs if at least two operations are performed successively.

\medskip
\paragraph*{Outcome Probabilities.}
The effect of non-ideal QRFs can also be   seen 
by considering outcome probabilities:
if $\hat\sigma^{|L} = \hat\rho_A \otimes \hat\varsigma_S$
is a product state described from the perspective of $L$,
and $\hat M^{|L}(i) = \mathsf R_A\bigl[\hat m_S^{|L}(i)\bigr]$
is a relational POVM element,
then one easily finds that the probability $p_i$ of obtaining
the outcome $i$ is~\cite{Garmier2023,Loveridge2018,Carette2025}
\begin{gather}
    p_i = \tr_{AS}\bigl( \hat M^{|L}(i) \hat\sigma^{|L} \bigr)
    = \tr_S\bigl( \hat m^{|L}_S(i) \hat\varsigma'_S \bigr), \notag\\
    \hat\varsigma'_S = \!\int_G\der g\,
    p(g)\, \mathsf U_S^\dagger(g)[\hat\varsigma_S],
    \quad p(g) = \tr_A\bigl( \hat\gamma_A^{|L}(g) \hat\rho_A \bigr).
    \label{eq:fuzziness}
\end{gather}
We see that $p_i$ can be computed directly from the underlying
POVM element and an \emph{effective state} $\hat\varsigma'_S$.
The latter is a mixture of the ensemble
$\smash{\bigl\{p(g), \, \mathsf U_S(g)[\hat\varsigma'_S]\bigr\}_{g \in G}}$,
where $p(g)$ is the probability distribution of $A$'s orientation.
\Ie, $\hat\varsigma_S'$ is a mixture of transformed versions of
$\hat\varsigma_S$,
and so it typically more mixed than $\hat\varsigma_S$, if $A$ is non-ideal, as in this case
the probability distribution $p(g)$ must be spread-out;
otherwise, the state $\hat\rho_A$ would have definite orientation,
contradicting non-ideality.

Figure~\ref{fig:Galilei_fuzziness} shows $p(g)$
for the non-ideal Galilei QRF from example~\ref{exa:Galilei},
in the case where the state from the laboratory perspective is of the form
$\hat\rho_A = \frac{1}{\sqrt{2}} \bigl(\ket{g} + \sket{g'}\bigr) \times (\text{h.c.})$.
Unlike in the ideal case, the peaks of $p(g)$ have finite width
and interference appears, both due the non-orthogonality of the states
$\ket{g}$~\cite{Garmier2023}.
The interference can in principle be accessed by relational operators.\footnote{
    Interference effects in $p(g)$ only translate to interference effects in $p_i$
    if $S$ sufficiently breaks the symmetry of $G$;
    for example, if $S$ carries the regular representation
    and $S$ is prepared in an orientation state,
    then it is always the case.
}
If instead $\frac{1}{2}(\ketbra{g}{g} + \ketbra{g^\prime}{g^\prime})_A
\otimes \hat \varsigma_S$,
then no interference occurs, implying that relational operators can distinguish
between coherent superpositions and probabilistic mixtures of QRF orientations.
In contrast, this distinction is not possible if $A$ were ideal~\cite{Carette2025},
suggesting that the perspective of non-ideal frames is more than a
``blurred'' versions of the perspective of ideal frames;
we will expand on this point in section~\ref{sec:POVs_Abelian}.

Our approach goes beyond mere outcome probabilities,
as it defines the perspective of a QRF through subsystem structures.
This allows us to understand the effect of using a non-ideal QRF
in more detail, through superselection and back-reaction.

\begin{figure}[!tb]
    \centering
    \includegraphics{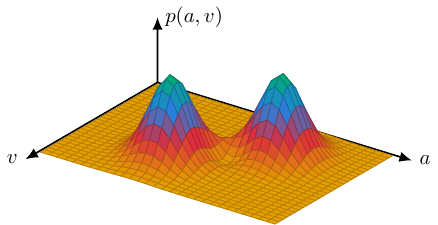}
    \caption{%
        Take example~\ref{exa:Galilei},
        additionally assuming rank-one orientation POVM elements
        $\hat\gamma^{|L}_A(g) = \sketbra{g}{g}$, $g = (a,v)$,
        given by minimal-uncertainty Gaussian states centred
        around $g$ in phase-space.
        If $\hat\rho_A = \frac{1}{\sqrt{2}} \bigl(\ket{g} + \sket{g'}\bigr)
        \times (\text{h.c.})$
        is a superposition of two orientation POVM states
        then we observe two peaks of width $\propto 1/\sqrt{m}$
        as well as wave-like interference between the peaks
        in the probability distribution $p(g)$~\cite{Garmier2023}.
        Figure adapted from~\cite{Garmier2023},
        see appendix~\ref{app:Galilei} for details.
    }
    \label{fig:Galilei_fuzziness}
\end{figure}

\medskip
\paragraph*{QRF Jumps.}
Obtaining the \emph{explicit} form of the QRF jump $\hat V^{\to A}$ is generally difficult,
except for simple situations such as the cases of trivial or regular QRFs studied above,
since finding the explicit decompositions over $k$ and $r$
heavily depends on the systems $A$ and $S$.
The naive idea of replacing $\sketbra{g}{g}_A$ with $\hat\gamma_A^{|L}(g)$
in the map $\hat W$ defined in \cref{eq:regular_inverse_jump_core}
fails because the resulting operator is generally not unitary;
for example, this can be easily seen if $A$ is a trivial QRF.

Nevertheless, the map $\hat W$, as defined in \cref{eq:regular_inverse_jump_core}, 
is still useful beyond regular QRFs. As shown in appendix~\ref{app:embedding},
if the QRF $A$ has rank-one POVM elements, $\hat\gamma_A^{|L}(g) = \frac{1}{c}\ketbra{g}{g}_A$,
$c>0$, then we can \emph{embed} the QRF $A$ into a regular QRF $\tilde A$ with the map
\begin{equation}
        \hat E : \H_A \to \H_{\tilde A},
        \quad\hat E := \frac{1}{\sqrt{c}} \int_G \der g \,
        \ket{g}_{\tilde A} \bra{g}_A.
        \label{eq:embedding}
\end{equation}
Applying $\mathsf W^\dagger[\;\cdot\;] := \hat W^\dagger [\;\cdot\;] \hat W$,
we get
\begin{align}
    \mathsf W^\dagger \circ (\mathsf E \otimes \id_S)
    \bigl[\mathsf R_A[\hat m_S]\bigr]
    &= \hat\pi\bigl( \hat\id_{\tilde A} \otimes \hat m_S \bigr)\hat\pi,
    \label{eq:almost_jump} \\
    \mathsf W^\dagger \circ (\mathsf E \otimes \id_S)
    \bigl[ \hat U^{|L}(g) \bigr]
    &= \hat\pi\bigl( \hat L_{\tilde A}(g) \otimes \hat\id_S \bigr)\hat\pi,
    \label{eq:almost_jump_reorientation}
\end{align}
where $\mathsf E[\;\cdot\;] = \hat E [\;\cdot\;] \hat E^\dagger$,
and $\hat\pi$ is an orthogonal projector on $\H_{\tilde A} \otimes \H_S$
which commutes with $\hat L_{\tilde A} \otimes \hat\id_S$
(see proposition~\ref{prop:general_jump} in appendix~\ref{app:embedding} for details).
A particular case of this result, restricted to the trivial representation sector,
was first found in~\cite{Hamette2021}.
Since $\hat\pi$ commutes with the global action of $G$,
it can be written in block diagonal form with respect to the $G$-representations
labelled by $r$, \ie\ $\hat\pi \cong \oplus_r \hat\pi_r$.
Then, the right-hand sides of \cref{eq:almost_jump,eq:almost_jump_reorientation}
are already compatible with the direct sum decomposition over $r$
in theorem~\ref{thm:QRF_POV}.
Thus, we \emph{almost} succeed at making relational operators act semi-simply,
were it not for the $r$-sectors where $\hat \pi_r$ is neither zero nor the identity.
Thus, to complete the QRF jump to $A$'s perspective,
we must merely find the remaining decomposition over $k$ for each
of these $r$-sectors separately.
If $G$ is Abelian, then $\hat\pi$ is simple enough so that the
complete perspective of $A$ can be obtained, essentially,
by embedding and applying $\mathsf W$;
we discuss this in section~\ref{sec:POVs_Abelian}.

More generally, \cref{eq:almost_jump,eq:almost_jump_reorientation}
can provide insights about the QRF perspective beyond the Abelian case.
Specifically, one can show (see proposition~\ref{prop:general_jump}) that for regular QRFs,
$\hat\pi = \hat\id_{\tilde AS}$. Therefore, to understand non-ideality,
it is enough to focus on those $r$-sectors where $\hat\pi_r$ is not the identity,
as in all other sectors~$\hat\pi$ acts as in the regular case.
For example, if $S = S_1S_2$ is bipartite from the laboratory perspective,
then $S{:}A|A$ is bipartite in those $r$-sectors where $\hat\pi_r$ is the identity,
but not necessarily in other sectors
(see also \cite{AhmadAli2022,Hamette2021} for a similar discussion in the perspective-neutral framework).
As $A$ tends to a regular frame, the $r$-sectors where $\hat\pi$ acts as the identity
become more significant, and $S{:}A|A$ tends to a bipartite system;
for a regular QRF, $S{:}A|A$ is bipartite in all $r$-sectors,
in accordance with our discussion of the regular case in section~\ref{sec:extreme_cases}.

\section{QRF Perspectives for Abelian Symmetry Groups}
\label{sec:POVs_Abelian}

\noindent
If $G$ is Abelian and $A$ a QRF with rank-one POVM elements
then the QRF jump $\hat V^{\to A}$ and the perspective of the QRF $A$
can be computed explicitly.
For simplicity, we only discuss the results for $G = \gpU(1)$ here,
while the theory for any Abelian symmetry group
is developed in appendix~\ref{app:POVs_Abelian}.

We view $\gpU(1) = (\bbR/2\pi,+)$ as addition modulo $2\pi$,
and we write the orientation POVM elements as
$\hat\gamma^{|L}_A(\theta) = \frac{1}{c} \ketbra{\theta}{\theta}_A$, $c > 0$.
The irreducible representations of $\gpU(1)$ are one-dimensional and given by
$\gpU(1) \ni \varphi \mapsto \ex^{\i r \varphi} \in \bbC$, for $r \in \bbZ$;
$r = 0$ is the trivial representation.
Physically,~$\gpU(1)$ can represent the rotation of an object
around a fixed axis, in which case the charge $r$ labels the
\emph{angular momentum} along that axis.
As mentioned earlier, $\gpU(1)$ can also model a periodic quantum clock,
with $r$ labelling its energy levels~\cite{Peres1980}.

The tensor product of two irreducible representations is simply
$\ex^{\i r_1 \varphi} \ex^{\i r_2 \varphi} = \ex^{\i(r_1+r_2)\varphi}$,
hence angular momentum addition is normal addition of the angular momentum
values.\footnote{
    Contrast this with spin-addition in the non-Abelian
    group $\SU(2)$.
}
With $\sigma_A$, $\sigma_S$ and $\sigma_{AS}$ we denote the
\emph{angular momentum spectra} of the systems $A$, $S$ and $AS$
as seen from $L$, \ie\ the sets of angular momentum values
labelling the irreducible representations occurring in those systems.

In this case, the ansatz in \cref{eq:almost_jump} works completely.
One can show (see proposition~\ref{prop:POV_Abelian}
in appendix~\ref{app:POVs_Abelian})
that we obtain the perspective of $A$ in accordance with
theorem~\ref{thm:QRF_POV} if we set
\begin{gather}
    \hat V^{\to A} := \hat W^\dagger (\hat E \otimes \hat\id_S), \label{eq:u1jump1}\\
    \H^{|A} := \hat V^{\to A} \H^{|L} \subset \H_{\tilde A} \otimes \H_S, \label{eq:u1jump2}
\end{gather}
where $\tilde A$ is a regular QRF into which $A$ is embedded
(see proposition~\ref{prop:embedding} in appendix~\ref{app:embedding})
and where $\hat W$ is defined in \cref{eq:regular_inverse_jump_core}.
To describe the perspective of $A$ we define
\begin{align}
    \bb\kappa &: \sigma_{AS} \to \mathcal{P}(\sigma_S), \notag\\
    \qquad \bb\kappa(r) &:= \sigma_S \cap (r - \sigma_A),
    \label{eq:kappa}
\end{align}
where $\P(\;\cdot\;)$ is the power set and  ``$-$'' denotes
Minkowski set subtraction.
With this, the perspective of $A$ is
\begin{alignat}{4}
    \H^{|A} &= \bigoplus\nolimits_{\bb k\in\bb\kappa(\sigma_{AS})}
    &\H_{\bb\kappa^{-1}(\bb k),\tilde A} &\otimes \H_{\bb k,S},
    \label{eq:main_Hilbert_space_A_Abelian}\\
    \A^{S:A|A} &= \bigoplus\nolimits_{\bb k\in\bb\kappa(\sigma_{AS})}
    &\hat\id_{\bb\kappa^{-1}(\bb k),\tilde A} &\otimes \O_{\bb k,S},
    \label{eq:main_relative_operators_A_Abelian}\\
    \A^{\overline{S:A}|A} &= \bigoplus\nolimits_{\substack{
        \bb k\in\bb\kappa(\sigma_{AS})\\r\in\bb\kappa^{-1}(\bb k)}}
    &\O_{r,\tilde A} &\otimes \hat\id_{\bb k,S},
    \label{eq:main_commutant_A_Abelian}\\
    \hat U^{|A}(\varphi) &= \bigoplus\nolimits_{\substack{
        \bb k\in\bb\kappa(\sigma_{AS})\\r\in\bb\kappa^{-1}(\bb k)}}
    &\ex^{\i r\varphi}\, \hat\id_{r,\tilde A} &\otimes \hat\id_{\bb k,S},
    \label{eq:main_representation_A_Abelian}
\end{alignat}
where ``$\H_{\bb x,Y}$'' denotes the subspace of a system $Y$
with angular momentum contained in the set $\bb x$,
and analogously ``$\H_{r,Y}$'' is the subspace with a single angular
momentum value $r$. We see that the function $\bb\kappa$ completely
determines the structure of the perspective of $A$.
Figure~\ref{fig:charge_diagram} illustrates $\bb\kappa$ for two examples.
Note in particular that the simultaneous eigenspaces of $\Z^{S:A|A}$
are labelled by the sets $\bb k \in \bb\kappa(\sigma_{AS})$
and that the $G$-invariant operators are jointly generated by
$\smash{\A^{S:A|A}}$ and $\smash{\A^{\overline{S:A}|A}}$.
Finally, a relational POVM element
$\hat M^{|L} = \mathsf R_A\bigl[\hat m_S^{|L}\bigr]$ becomes
\begin{equation}
    \hat M^{|A}
    = \bigoplus\nolimits_{\bb k \in \bb\kappa(\sigma_{AS})}
    \hat\id_{\bb\kappa^{-1}(\bb k),\tilde A}
    \otimes \hat\pi_{\bb k,S}\, \hat m_S^{|L}\, \hat\pi_{\bb k,S}
    \label{eq:main_relative_POVM_A_Abelian}
\end{equation}
from the perspective of $A$,
with $\hat\pi_{\bb k,S}$ the orthogonal projector onto $\H_{\bb k,S}$.

From \cref{eq:main_relative_operators_A_Abelian} we see that
the subsystem $S{:}A|A$ is a direct sum of sectors isomorphic to ``pieces''
$\H_{\bb k,S} \subset \H_S$ of the system $S|L$,
which may furthermore intersect non-trivially.
This is the superselection we discussed in
section~\ref{sec:non_ideal_case}.
If $\sigma_A$ is smaller than $\sigma_S$
(figure~\ref{fig:charge_diagram_2}),
then none of the sectors is isomorphic to $S|L$,
and, hence, $S{:}A|A$ is itself not isomorphic to $S|L$.
In the opposite case (figure~\ref{fig:charge_diagram_1}),
one of the sectors in $S{:}A|A$, say $\bb k_0$, is isomorphic to $S|L$
(the thick middle sector in the figure).
The larger $\sigma_A$ becomes, the more charges are contained in
$\bb\kappa^{-1}(\bb k_0)$, and the other sectors, not isomorphic to $S|L$,
contain only extreme charges.
Hence: \emph{the larger $\sigma_A$,
the closer $S{:}A|A$ is to being isomorphic to $S|L$}
in the sense that deviations from the ideal case can only be seen at high
absolute angular momentum values.

Let us apply this result to a specific situation.
Suppose that $A$ is a clock, \ie\ a periodic QRF for time translations.
Furthermore, assume that, from the perspective of $L$,
the total state on $AS$ has ``low energy'':
that is, its support is restricted to sectors where $S{:}A|A$
is isomorphic to $S|L$.
In this case, the time evolution of $S{:}A|A$ behaves as if $A$ were an ideal clock.
Conversely, for general states with support on 
``high energy'' sectors, deviations from ideal time evolution will be detected from
the perspective of $A$.
This is a signature of the finite resources of the QRF.
In the limit where $A$ becomes regular, $\sigma_A = \bbZ$,
and hence only the single sector $\bb k_0 = \bb\kappa(r) = \sigma_S$
for all $r \in \bbZ$ survives in $S{:}A|A$.
This is in line with the general theory for regular QRFs
we discussed above in section~\ref{sec:non_ideal_case}.

The most general operator in $\A^{S:A|A}$ has the form
\begin{equation}
    \hat A^{|A} = \bigoplus\nolimits_{\bb k \in \bb\kappa(\sigma_{AS})}
    \hat\id_{\bb\kappa^{-1}(\bb k),\tilde A}
    \otimes \hat A_{\bb k,S},
\end{equation}
where $\hat A_{\bb k,S}$, $\bb k \in \bb\kappa(\sigma_{AS})$,
are \emph{arbitrary} operators.
In contrast, a general relational operator $\hat B^{|A}
= \mathsf V^{\to A} \circ \mathsf R_A[\hat b_S]$ has the \emph{less general}
form given by \cref{eq:main_relative_POVM_A_Abelian},
\ie\ $\hat A_{\bb k,S} = \hat\pi_{\bb k,S}\, \hat b_S\, \hat\pi_{\bb k,S}$.
As we discussed in section~\ref{sec:non_ideal_case},
there is generally a discrepancy between relational operators
and the algebra they generate.
Whenever this discrepancy occurs, relational operators are spread over multiple sectors
of the centre, while $\A^{S:A|L}$ also contains operators
which are constrained to a single sector,
\eg\ the orthogonal projector $\hat\id_{\bb\kappa^{-1}(\bb k),\tilde A}
\otimes \hat\id_{\bb k,S}$ onto a single sector $\bb k$ of the centre.
Consequently, it is in general not possible to determine 
the sector of the centre from measuring a single relational POVM 
if $A$ is non-ideal.

\begin{figure}[!h]
    \vspace{-1em}
    \subfigure[$\sigma_A \supset \sigma_S$.]{\includegraphics{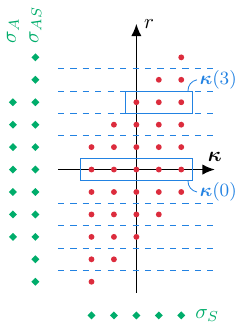}
    \label{fig:charge_diagram_1}}
    \hfill
    \subfigure[$\sigma_A \subset \sigma_S$.]{\includegraphics{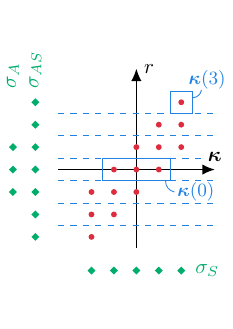}
    \label{fig:charge_diagram_2}}
    \vspace{-0.5em}
    \caption{
        The set $\bb\kappa(r)$ as a function of total angular momentum $r$
        for two example cases (a) and (b):
        a point at vertical position $r$ and horizontal position $z$
        indicates that $z \in \bb\kappa(r)$.
        Simultaneous eigenspaces of $\Z^{S:A|A}$ are labelled by such sets
        and indicated as blue dashed blocks;
        the representation charges contained in the block can be read off
        the vertical axis.
        \Eg, the central block in (a) spans three values of $r$,
        \ie\ the corresponding simultaneous eigenspace of the centre
        carries these three values of angular momentum.
    }
    \vspace{-1.5em}
    \label{fig:charge_diagram}
\end{figure}

Let us illustrate these insights with a concrete example
(see appendix~\ref{app:qutrit_qubit} for detailed calculations):
\begin{example}{}{qubit_qutrit}
    Let $\H_A = \bbC^3$ and $\H_B = \bbC^2$
    both be QRFs for $G = \gpU(1)$
    with rank-one orientation POVMs
    and with angular momentum spectra $\sigma_A = \{-1,0,1\}$
    and $\sigma_B = \{-1,1\}$, respectively.
    We consider no additional system $Q$.
    The perspectives of $A$ and $B$ are summarized
    by the following diagrams:
    
    \begin{center}
        \includegraphics{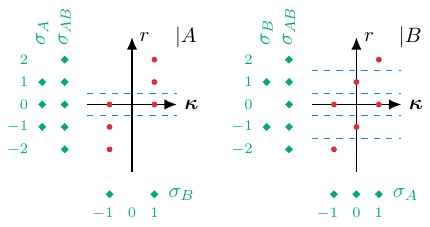}
    \end{center}

    From $A$'s perspective, there is a sector of~$B{:}A|A$,
    the sector $\bb\kappa(0)$, which is isomorphic to~$B|L$. We also
    notice that $\dim\A^{B:A|A} = 6 > 4 = \dim\O_B$.
    If $A$ were regular, then the dimensions would be equal;
    the back-reaction mentioned in section~\ref{sec:non_ideal_case}
    leads to more degrees of freedom in $B{:}A|A$ than in $B|L$.

    \medskip
    From $B$'s perspective,
    there are three sectors of $A{:}B|B$,
    none isomorphic to $A|L$.
    Note that $\bb\kappa(1) = \bb\kappa(-1)$ is one sector,
    despite being disconnected in the diagram.
    Here, $\dim\A^{A:B|B} = 7 < 9 = \dim\O_A$.

    \medskip
    To understand back-reaction better, we return to $A$'s perspective and
    show explicitly that there is information that is not accessible using
    measurements in $\mathsf R_A[\O_B]$, but is with
    measurements in~$\A^{B:A|L}$.
    For this, consider from $L$'s perspective, the state
    $\hat\rho^{|L} = \sketbra{\psi_\varphi}{\psi_\varphi}_A \otimes \hat\sigma_B$,
    where
    \begin{equation}
        \sket{\psi_\varphi}
        = \frac{1}{\sqrt 2}\bigl(\ket{\theta=0} + e^{i \varphi} \ket{\theta=2\pi/3}\bigr).
    \end{equation}
    Here, $\ket{\theta}$, $\theta\in\gpU(1)$, are the states from which the orientation
    POVM of $A$ is built, $\varphi \in \bbR$ is a relative phase
    and $\hat\sigma_B$ is a $G$-invariant state.
    From the perspective of~$L$, the information about the phase $\varphi$ is encoded in $A$ alone
    and cannot be accessed by operations that only act on $B$.
    A straightforward computation shows that for any relational
    operator $\hat F^{|L} = \mathsf R_A[\hat f_B]$, its expectation value
    $\tr\bigl[\hat F^{|L} \hat\rho^{|L}\bigr]$ does not depend on $\varphi$.
    However, one can find operators $\hat X^{|L} \in \A^{B:A|L} \backslash \mathsf R_A[\O_B]$
    such that their expectation values $\tr\bigl[ \hat X^{|L} \hat\rho^{|L} \bigr]$
    do depend on~$\varphi$. We show this in appendix~\ref{app:qutrit_qubit}.
    Hence, the relative phase is not accessible to relational
    operators, but it is accessible to certain polynomials of
    relational operators. As argued in section~\ref{sec:non_ideal_case},
    we can understand this phenomenon in
    terms of the back-reaction of $S|L$ onto the QRF $A|L$ due to the
    finite resources of the latter.
\end{example}

\vspace{-\baselineskip}
\section{Discussion and Outlook}
\label{sec:conclusion_outlook}

\noindent
We have developed a framework that allows us to define
the perspective of any QRF,
extending the framework of~\cite{CastroRuiz2025a} to non-ideal QRFs.
Our extension required refining the definition of
the QRF perspective: it is a factorization of the total Hilbert space,
as seen from a possibly fictitious laboratory frame,
such that invariant operators and operators
relationalized with respect to the QRF act semi-simply.
We physically motivated this definition by the
\hyperref[ppl:perspective]{perspective} and
\hyperref[ppl:invariance]{invariance} principles.
Theorem~\ref{thm:QRF_POV}, based on results by \cite{Bianchi2024}, showed that for any QRF there is a
unitary transformation from the laboratory's perspective into the
perspective of the QRF, the QRF jump.
We prove the theorem rigorously for finite dimensions and compact symmetry groups,
and argue that a generalization to infinite dimensions and non-compact groups should be possible;
actually carrying out this generalization in a rigorous way
is important future work.
As in~\cite{Vanrietvelde2020,Hamette2021,CastroRuiz2025a}, 
the reversible QRF transformation from the perspective of a QRF $A$
to that of a QRF $B$ is constructed by composing the inverse jump
from $A$'s perspective to the laboratory perspective with
the jump from there to $B$'s perspective.

The perspective associated to a non-ideal QRF exhibits
several qualitative differences from that of an ideal QRF.
Firstly, from the perspective of a non-ideal QRF, the observed system is superselected;
ranging from complete superselection, when the QRF is trivial,
to no superselection at all, when the QRF is regular.
This is in line with the well-known result that superselection of the system
is due to the absence of a reference frame and disappears when
an ideal reference frame is introduced~\cite{Aharonov1967,Bartlett2007}.
This means that the system $S|L$ might not be isomorphic to~$S{:}A|A$.
Relatedly, systems decomposing into a tensor product from the perspective
of the laboratory generally no longer do so from the perspective of a non-ideal QRF.
However, they approximate such a decomposition in the infinite-resource limit.
Secondly, the relationalization map is generally not an algebra homomorphism
for non-ideal QRFs, and hence its image $\mathsf R_A[\O_S]$ not always an algebra.
This in turn has two important consequences for the perspective of the reference frame.
The first is that commutation relations are not preserved
by the relationalization map,
and hence commuting operators on $S|L$ do not always correspond
to commuting ones on $S{:}A|A$.
The second consequence is that the algebra $\A^{S:A|L}$ generated by
relational operators is generally a larger space than the
vector space $\mathsf R_A[\O_S]$ of relational operators.
In some cases, like in example~\ref{exa:qubit_qutrit},
$\A^{S:A|L}$ even has a higher dimension than $\O_S$.

Operationally, observing the discrepancy between $\mathsf R_A[\O_S]$
and $\A^{S:A|L}$ requires repeated uses of the reference frame;
it appears when considering elements in the algebra $\A^{S:A|L}$
generated by relational operators $\mathsf R_A[\O_S]$
which are not relational operators themselves. 
Physically, from the laboratory's perspective,
one may understand them as the system $S$ back-reacting on
the QRF $A$, which impacts the second use of the QRF.
A similar phenomenon is known
in other contexts as reference frame degradation~%
\cite{Bartlett2006,Poulin2007}. However, in our case,
this back-reaction is embedded in the structure of the state space
of the system relative to the non-ideal frame.
It is interesting future work to study more concretely
when these phenomena occur and what their dynamical causes~are.

Although theorem~\ref{thm:QRF_POV} provides a way to obtain the QRF jump,
finding its explicit form is difficult.
Besides the simple cases of regular or trivial QRFs,
we determined the explicit form of the jump for all Abelian groups.
A natural next step is to obtain the explicit jumps for non-Abelian groups
with physical relevance such as $\SU(2)$ or the (centrally extended) Galilei group.
In the latter case, we can use our framework to jump into the perspective
of a quantum particle, which is not an ideal QRF
for the Galilei group due to the Heisenberg uncertainty relations.
In this work, we have already made a first step into this direction
by showing that many properties of the perspective
are present if one embeds the non-ideal quantum
reference frame into a regular one.

Our framework is compatible with different physical interpretations of
the perspective of a QRF $A$.
One option is to view the perspective of the QRF $A$
as a mathematical choice of variables,
made by an observer in the laboratory,
that is naturally adapted to relational operators with respect to $A$.
Perspective changes are then quantum generalizations of coordinate transformations.
This is the laboratory-first view expressed in
figure~\ref{fig:relational_POVM}.

A second option is to interpret $A$'s perspective
as the description of an internal observer who can use the QRF $A$
but not the laboratory frame.
This is the symmetry-first view of figure~\ref{fig:swmeasurement}.
This view assumes that the internal observer's inability to use 
the laboratory frame restricts them to invariant operators, $\mathsf G_{AS}[\mathcal{O}_{AS}]$.
Moreover, the view assumes that any operation this internal observer 
performs on the system is correctly described by polynomials of 
relational operators with respect to $A$.
Note that these assumptions imply
that any relational POVM can be implemented
without using the laboratory frame.
This motivates the split of invariant operators
into $S{:}A|A$, the ``system $S$ relationalized with respect to $A$'',
and $\overline{S{:}A}|A$, the ``extra particle''.
We can further assume that the perspective of $A$ is valid even 
without external laboratory frame, because this external structure
is precisely what the internal observer is unable to use,
and not more.
This makes our formalism consistent with the \emph{potential}
existence of an external frame.
Importantly, even in the absence of the external laboratory $L$,
it is possible to operationally test the reference frame
transformation between two internal observers using QRFs
$A$ and $B$ respectively:
for this test the observers just need to perform
measurements on the system from their perspective
and communicate their results~\cite{Hamette2025,Doat2025}.

We have studied situations in which the resources
of a quantum reference frame for group $G$ are limited.
However, we have assumed that all other elements important
in an experiment, such as reference frames for other degrees of freedom
besides $G$-orientation as well as measurement apparatuses,
can have potentially unbounded resources.
Evidence of this assumption can be seen,
for example, in our throughout use of complex Hilbert 
spaces~\cite{Penrose1971, Barbado2025}.
It is therefore not clear whether our framework is still valid
if it is not possible, even in principle, to construct a laboratory.
Understanding what one can do without using macroscopic laboratory equipment
would shine light on foundational questions such as how an observer
which is ``small'' compared to the observed system
should describe this system.
While our framework requires only that the QRF resources are bounded and 
assumes consistency with a (possibly fictitious, infinite-resource) laboratory
frame, it may be used as a stepping stone to push
quantum theory into unexplored realms.

\vspace{-\baselineskip}
\section*{Acknowledgements}
\vspace{-0.7\baselineskip}

\noindent
We thank V.~Kabel for her feedback on a draft of the paper,
and O.~Oreshkov for his opinion on early ideas.
S.C.G.\ thanks A.~Spalvieri and A.~Giovanakis for their feedback on various parts of the paper;
J.~Huber, N.~Ormrod, A.-C.~de la Hamette, F.~Giacomini,
B.~Sahdo, N.~Cohen, P.~Höhn and L.-Q.~Chen for insightful discussions.
\mbox{E.C.-R.}\ thanks E.~Bianchi for pointing out~\cite{Bianchi2024}
and L.C.~Barbado, \v{C}.~Brukner, N.~Cohen and P.~Höhn for insightful discussions.
This work is supported by the SNSF Grant No.~20QU-1\_225171
and by the SNSF Ambizione Grant No.~PZ00P2-208885.
We thank the SwissMAP NCCR and the ETH Zurich Quantum Center for their support.
We acknowledge support from the \href{https://withoutspacetime.org}{WithOut SpaceTime (WOST)} project,
funded by the Grant ID\#~63683 from the John Templeton Foundation (JTF),
as well as by the
\href{https://qiss.fr}{The Quantum Information Structure of Spacetime, Second Phase (QISS 2)} project, funded by the Grant ID\#~62312 from the JTF.
The opinions expressed in this work are those of the authors and do not necessarily
reflect the views of the JTF.
This work is part of the European Union COST Action CA23130
\href{https://cost.eu/actions/CA23130/}{Bridging high and low energies in search of quantum gravity}.

\def\bibsection{\section*{\refname}} 



\onecolumngrid

\appendix

\section{Notations and Conventions}
\label{app:notations_conventions}

\noindent
We choose units with $\hbar = 1$.
Operators are written with a hat.
Given an operator $\hat X$,
we define the superoperator $\mathsf X[\;\cdot\;]
:= \hat X[\;\cdot\;]\hat X^\dagger$.
With an \emph{isomorphism} of Hilbert spaces
we mean an isomorphism of the underlying vector spaces
which also preserves the inner product;
isomorphisms of Hilbert spaces induce isomorphisms
of \emph{operator algebras} (see below) on the involved spaces.

\medskip
A \emph{symmetry group} $G$ is a finite-dimensional, unimodular Lie group.
Integrals of the form $\int_G\der g \dots$ denote a left-, right- and
inversion-invariant Haar integral over $G$
and $|G| := \int_G\der g$ is the measure of $G$.
It is finite if and only if $G$ is compact.
We include the case of a zero-dimensional $G$ which corresponds to
a discrete group with at most countably many elements;
in that case, $\int_G\der g\dots$ is implicitly understood as a series
$\sum_{g\in G}\dots$ and $|G|$ is the cardinality of the group.
See \eg~\cite{Cohn2013} for proofs of the above statements.

\medskip
A \emph{unitary representation} $\hat U : G \to \gpU(\H)$
of a symmetry group $G$ on a Hilbert space $\H$
is always assumed to be \emph{strongly continuous}~\cite{Fell1988}.
The \emph{left-} and \emph{right-regular} representations
$\hat L$ and $\hat R$ act on $L^2(G)$ as
$\varphi \mapsto (\hat L(g)\varphi)(g') = \varphi(gg')$
and $\varphi \mapsto (\hat R(g)\varphi)(g') = \varphi(g'g^{-1})$,
respectively.
Appendix~\ref{app:representation_theory_groups} treats the representation
theory of symmetry groups in more detail.

\medskip
An \emph{operator algebra} $\A$ is an algebra of operators on a Hilbert space
$\H$ closed under operator multiplication, complex linear combinations
and Hermitian transposition, and which contains the identity.
$\O$ is the operator algebra of all operators on $\H$,
and $\O[\hat X_1,\dots]$ denotes the operator algebra generated by
the set of operators $\{\hat X_1, \dots\}$.

If $\dim\H < \infty$, then $\A$ is a \emph{von Neumann algebra}.
If $\dim\H = \infty$, then $\A$ is often not well-defined,
containing \eg\ unbounded operators and improper operators.
In the latter case, $\A$ is a von Neumann algebra only if it
consists solely of bounded operators and if it is closed under
the weak operator topology~\cite{Neumann1930,Davidson1996}.
Appendix~\ref{app:representation_theory_algebras} discusses the
representation theory of von Neumann algebras.
We mostly gloss over subtleties with infinite dimensions in the main text,
knowing that all results are rigorous in finite dimensions,
and can be made rigorous by appropriately modifying definitions and statements
also in the infinite-dimensional case.
What exactly must be modified and to what extent generalizations are possible
is discussed in appendices~\ref{app:representation_theory_groups}
and~\ref{app:representation_theory_algebras}
whenever the modifications impose themselves.

\medskip
With $\bigoplus_k\dots$ we denote a \emph{direct sum} of either Hilbert spaces
or algebras if $k$ takes on at most countably many values,
and it is implicitly understood as a \emph{direct integral}
if $k$ is continuous (this is related to the generalizations of our results
to infinite dimensions as just mentioned).
In the latter case, the integrand consists of subspaces spanned by improper
vectors or subalgebras on such spaces.
For details, see \eg~\cite{Neumann1949}.

\section{Representation Theory of Lie Groups}
\label{app:representation_theory_groups}

\noindent
We gather here some results about the representation theory of symmetry groups,
beginning with the notions of \emph{regular representations},
followed by the \emph{Peter-Weyl} theorem
and the properties of the $G$-twirl.

\medskip
\paragraph*{Regular Representations.}
For a symmetry group $G$,
consider the Hilbert space $L^2(G)$ consisting of
square-integrable functions $\psi : G \to \bbC$
(\ie, ``wave functions on $G$'').\footnote{
    Actually, $L^2(G)$ consists of equivalence classes of functions
    with two functions identified if they agree almost everywhere,
    \ie\ if they disagree on a set of measure zero~\cite{Cohn2013}.
}
The inner product is
\begin{equation*}
    \sbraket{\psi}{\varphi} = \int_G\der g\, \overline\psi(g)\, \varphi(g)
\end{equation*}
and $\braket{\psi}{\psi} < \infty$ for all $\ket{\psi} \in L^2(G)$.
The \emph{left-} and \emph{right-regular} representations are unitary
representations of $G$ defined as~\cite{Fell1988}
\begin{equation}
    \bigl(\hat L(g') \psi\bigr)(g) := \psi(g'^{-1}g),
    \qquad \bigl( \hat R(g') \psi\bigr)(g) := \psi(gg')
    \quad \forall\, g,g' \in G.
    \label{eq:app_regular_representations}
\end{equation}
Clearly, $[\hat L(g), \hat R(g')] = 0$ for all $g,g' \in G$.

A function $\ket{\psi} \in L^2(G)$ is evaluated at $g \in G$ by the
Dirac distribution $\bra{g}$ centred on $g$:\footnote{
    Strictly speaking, $\bra{g}$ is \apriori\ ill-defined on $L^2(G)$
    since equivalence classes of functions which are equal almost everywhere
    do not have well-defined function values.
    However, a class $\ket{\psi} \in L^2(G)$ containing continuous functions
    always contains a single continuous function $\psi_c \in \ket{\psi}$,
    so that we can define $\sbraket{g}{\psi} := \psi_c(g)$.
    Continuous classes are dense in $L^2(G)$~\cite{Cohn2013},
    making $\bra{g}$ densely defined.
}
\begin{equation*}
    \psi(g) = \sbraket{g}{\psi} =: \int_G\der g'\, \delta(g^{-1}g')\, \psi(g').
\end{equation*}
Taking $\ket{g}$ as the improper state corresponding to $\bra{g}$, one finds
\begin{equation}
    \braket{g'}{g} = \delta(g'^{-1}g)
    \quad \forall\,g,g'\in G,
    \label{eq:app_regular_overlaps}
\end{equation}
and thanks to unimodularity of $G$, even
$\delta(g'^{-1}g) = \delta(g^{-1}g') = \delta(gg'^{-1}) = \delta(g'g^{-1})$.
From this it follows that $\int_G\der g\, \ketbra{g}{g} = \hat\id$,
and $\{\ket{g}\}_{g\in G}$ is an improper basis:
$\ket{\psi} = \int_G\der g\, \psi(g) \ket{g}$.
From \cref{eq:app_regular_representations} it follows that
\begin{equation*}
    \hat L(g') \ket{g} = \ket{g'g},
    \qquad \hat R(g') \ket{g} = \sket{gg'^{-1}}.
\end{equation*}
All this is entirely analogous to how position eigenstates
are improper states of $L^2(\bbR)$ and acted upon
by the translation group $(\bbR,+)$.

\medskip
\paragraph*{Compact Symmetry Groups.}
We first consider the case of a compact symmetry group $G$,
and later generalize to non-compact groups.
The representation theory in the compact case is given by the following:
\begin{proposition}{Peter-Weyl Theorem}{Peter_Weyl}
    Let $G$ be a compact symmetry group, and $\H$ a Hilbert space carrying
    a unitary representation $\hat U$ of $G$.
    Then,
    \begin{equation*}
        \H \iso \bigoplus_r \H_r \otimes \H_{*,r},
        \qquad \hat U \iso \bigoplus_r \hat U_r \otimes \hat\id_{*,r},
    \end{equation*}
    where $r$ enumerates the inequivalent irreducible unitary
    representations of $G$, with corresponding representation spaces $\H_r$
    and action $\hat U_r$, and the multiplicity spaces $\H_{*,r}$
    can be zero-dimensional for some $r$.
    All irreducible representations of $G$ are finite-dimensional
    and there are at most countably many up to equivalence.
    
    \medskip
    If $\H = L^2(G)$, then the left- and right-regular representations
    (see appendix~\ref{app:regular} for definitions)
    act as
    \begin{equation*}
        \hat L \iso \bigoplus_r \hat U_r \otimes \hat\id_{*,r},
        \qquad \hat R \iso \bigoplus_r \hat \id_r \otimes \hat U_{*,r},
        \qquad \H_{*,r} = \H_r, \quad \hat U_{*,r} = \hat{\overline U}_r.
    \end{equation*}
\end{proposition}
\noindent
See~\cite{Peter1927} for the original statement and proof,
and~\cite{Fell1988a} for a textbook treatment.
Note that the Peter-Weyl theorem also applies for finite groups,
except that $L^2(G)$ in the regular representation is
replaced by the space of (necessarily finite) sequences $l^2(G)$.
In the compact case, the $G$-twirl is well-defined (because $|G| < \infty$):
\begin{proposition}{Properties of the $\boldsymbol G$-Twirl}{G_twirl}
    Let $G$ be compact and let $\hat U$ be a strongly continuous unitary
    representation of $G$ on a Hilbert space $\H$.
    The $G$-twirl
    \begin{equation}
        \mathsf G := \frac{1}{|G|} \int_G\der g\, \mathsf U(g)
        \label{app:G_twirl}
    \end{equation}
    satisfies the following properties:
    \begin{enumerate}
        \medskip
        \item $\mathsf G$ is a CPTP map.

        \medskip
        \item $\mathsf G$ is a projector:
        $\mathsf G^2 = \mathsf G$.

        \medskip
        \item $\mathsf G$ is self-adjoint in the Hilbert-Schmidt inner product.
        Furthermore,
        $\tr\bigl(\mathsf G[\hat a] \hat b\bigr)
        = \tr\bigl(\hat a \mathsf G[\hat b]\bigr)
        = \tr\bigl(\mathsf G[\hat a] \mathsf G[\hat b]\bigr)$.

        \medskip
        \item $\mathsf G[\hat a] = \hat a$ if and only if
        $\forall\,g\in G: \, [\hat a, \hat U(g)] = 0$.

        \medskip
        \item Assume that $\H \iso \bigoplus_r \H_r \otimes \H_{*,r}$,
        $\hat U \iso \bigoplus_r \hat U_r \otimes \hat\id_{*,r}$ decays into
        inequivalent irreducible representations $\hat U_r$ of $G$
        as in the Peter-Weyl theorem.
        Then,
        \begin{equation*}
            \mathsf G = \bigoplus_r (\mathsf D_r \otimes \id_{*,r}) \circ \Pi_r
        \end{equation*}
        where $\Pi_r$ is the orthogonal projection superoperator onto
        $\H_r \otimes \H_{*,r}$, and $\mathsf D_r : \hat a \mapsto
        \hat\id \cdot (\tr\,\hat a / \tr\,\hat\id)$
        is the trace-preserving de-phasing superoperator.
    \end{enumerate}
\end{proposition}
\noindent
See~\cite{Bartlett2007} for a standard treatment.
\begin{proof}
    1. Complete positivity follows because $\mathsf\id \otimes \mathsf U$
    is completely positive, and probabilistic mixtures of
    completely positive maps are completely positive.
    Preservation of trace follows similarly.

    2. We compute
    \begin{equation*}
        \mathsf G^2 = \frac{1}{|G|^2} \int_G\der g\,\der g'\,
        \mathsf U(g) \circ \mathsf U(g')
        = \frac{1}{|G|^2} \int_G\der g\,\der g'\, \mathsf U(g)
        = \frac{1}{|G|} \int_G\der g\, \mathsf U(g),
    \end{equation*}
    by invariance of the Haar measure.

    3. $\mathsf G$ is self-adjoint because
    \begin{equation*}
        \tr\bigl( \hat a^\dagger \mathsf G[\hat b] \bigr)
        = \frac{1}{|G|}\int_G\der g\,
        \tr\bigl( \hat a \hat U(g) \hat b \hat U^\dagger(g) \bigr)
        = \frac{1}{|G|}\int_G\der g\,
        \tr\bigl( \hat U^\dagger(g) \hat a \hat U(g) \hat b \bigr)
        = \tr\bigl( \mathsf G[\hat a]^\dagger \hat b \bigr),
    \end{equation*}
    where we have used the cyclicity of the trace.
    Using the inversion-invariance of the Haar measure ($G$ is unimodular),
    it follows that $\tr(\hat a \mathsf G[\hat b])
    = \tr(\mathsf G[\hat a] \hat b)$,
    and using the projector property, this is also equal to
    $\tr(\mathsf G[\hat a] \mathsf G[\hat b])$.
    
    4. If $\mathsf G[\hat a] = \hat a$, then
    \begin{equation*}
        [\hat a, \hat U(g)] = \bigl[\mathsf G[\hat a], \hat U(g)\bigr]
        = \frac{1}{|G|} \int_G\der g'\, \bigl(
            \hat U(g') \hat a \hat U^\dagger(g') \hat U(g)
            - \hat U(g) \hat U(g') \hat a \hat U^\dagger(g')
        \bigr)
        = 0,
    \end{equation*}
    where we have used the invariance of the Haar measure
    in one of the two terms.
    If $[\hat a,\hat U(g)] = 0$ for all $g \in G$, then
    \begin{equation*}
        \mathsf G[\hat a] = \frac{1}{|G|} \int_G\der g\,
        \hat U(g) \hat a \hat U^\dagger(a)
        = \frac{1}{|G|} \int_G\der g\, \hat a
        = \hat a.
    \end{equation*}

    5. See the proof of theorem 1 in~\cite[section II.C]{Bartlett2007}.
\end{proof}
\medskip
\paragraph*{Generalizations.}
If $G$ is not compact, then proposition~\ref{prop:Peter_Weyl} fails
already, because now irreducible representations of $G$
can be infinite-dimensional. However, generalizations exist:
the statement about finite dimensions must be removed,
and direct sums must be replaced by combinations of direct integrals and sums.
Specifically, if $G$ is second-countable and of type-I
(every unitary representation generates a type-I von Neumann algebra~%
\cite{Neumann1930}, see also appendix~\ref{app:representation_theory_algebras})
then the first part of the Peter-Weyl theorem
(complete reducibility of a given unitary representation)
generalizes in this sense~\cite{Fell1988,Fell1988a}.
The second part (decomposition of the regular representation)
also generalizes for such groups and in the same sense~%
\cite{Segal1950,Mautner1955},
and explicit descriptions of the irreducible representations
are known in some cases~%
\cite{HarishChandra1965,HarishChandra1966}.
We will mostly need the first part.
Furthermore, essentially all unimodular Lie groups
of interest in physics are type-I.
Especially, since we are interested mostly in algebraic
properties and not so much analytical ones
(\eg, improper states and subspaces are acceptable), for our purposes,
proposition~\ref{prop:Peter_Weyl} can be generalized to a very large class of
Lie groups encountered in physics and not just compact ones.

Proposition~\ref{prop:G_twirl}, strictly speaking, also fails,
because $|G| = \infty$ and hence the integral will generally not converge.
However, accepting improper states and focusing on algebraic properties,
proposition~\ref{prop:G_twirl} can be generalized for our purposes.

\section{Representation Theory of Operator Algebras
in the Presence of a Symmetry}
\label{app:representation_theory_algebras}

\noindent
In this appendix, we prove theorem~\ref{thm:QRF_POV}
using the representation theory of operator (strictly speaking, von Neumann)
algebras, also taking into account the action of the symmetry group.
We first consider the case of finite-dimensional Hilbert spaces
and algebras, as well as compact symmetry groups.
In this case, all operator algebras are von Neumann algebras
and the Peter-Weyl theorem, proposition~\ref{prop:Peter_Weyl},
applies without modification.
After proving the theorem with these assumptions,
we outline how it can be generalized to infinite dimensions
and non-compact symmetry groups.

\medskip
\paragraph*{Commutants.}
We will need two different commutants.
Let $\A < \O$ be a subalgebra of the operators $\O$ on a Hilbert space $\H$.
Then, we denote the commutant relative to $\O$ with a prime:
$\A' := \bigl\{ \hat X \in \O \st [\hat X,\A] = 0 \bigr\}$.
The commutant relative to $G$-invariant operators
(if $\A < \mathsf G[\O]$ applies)
is denoted with a bar:
$\overline\A := \bigl\{ \hat X \in \mathsf G[\O] \st [\hat X,\A] = 0 \bigr\}$.

\medskip
\paragraph*{Theorem~\ref{thm:QRF_POV} for Finite Dimensions and Compact $\boldsymbol G$.}
Still assuming finite-dimensional Hilbert spaces and a compact symmetry groups,
we consider a slight generalization of theorem~\ref{thm:QRF_POV}
by taking an arbitrary operator algebra
rather than an algebra generated by relational operators.
The result we are going to show here is a combination of quite well-known
results in the infinite-dimensional case~%
\cite{Davidson1996,Neumann1930,Neumann1949,Murray1936,Murray1943,Peter1927}.
This result has been derived and discussed first in \cite{Bianchi2024}
to define entanglement entropy for general invariant algebras.
However, because the finite-dimensional special case is rarely discussed
explicitly despite containing valuable algebraic intuition
for the result, we nevertheless provide a proof here.
Specifically, we will show the following:
\begin{proposition}{}{algebra_decomposition_full}
    Let $\H$ be a finite-dimensional Hilbert space carrying a continuous unitary representation
    $\hat U$ of a compact symmetry group $G$, and let $\A < \mathsf G[\O]$
    be a $G$-invariant operator subalgebra.
    Then, there exists an isomorphism inducing the relations
    \begin{gather*}
        \H \iso \bigoplus_{k,r} \H_r \otimes \H_{r,k} \otimes \H_k,
        \qquad \hat U(g) \iso \bigoplus_{k,r} \hat U_r(g)
        \otimes \hat\id_{r,k} \otimes \hat\id_k \notag\\
        \A \iso \bigoplus_k\Bigl(\bigoplus_r \hat\id_r \otimes \hat\id_{r,k} \Bigr)
        \otimes \O_k,
        \qquad\overline\A \iso \bigoplus_{k,r} \hat\id_r
        \otimes \O_{r,k} \otimes \hat\id_k,
    \end{gather*}
    where
    \begin{enumerate}
        \medskip
        \item\label{item:r_labels} $r$ labels the inequivalent irreducible unitary representations
            of $G$, such that the irreducible representation~$\hat U_r$
            acts on $\H_r$.

        \medskip
        \item\label{item:k_labels} $k$ labels the simultaneous eigenspaces
        $\J_k \iso \oplus_r \H_r \otimes \H_{r,k} \otimes \H_k$
        of all elements in the centre
        $\Z := \A \cap \A' = \A \cap \overline\A$.
    \end{enumerate}
\end{proposition}
\noindent
Clearly, this implies theorem~\ref{thm:QRF_POV} if $\dim\H < \infty$
and $G$ is compact.
We prove this by deriving a series of useful results.
\begin{proposition}{}{algebra_decomposition_1}
    Let $\H$ be a finite-dimensional Hilbert space, and let $\A < \O$ be an operator
    subalgebra with trivial centre, \ie\ $\Z = \A \cap \A' = \O[\hat\id]$.
    Then, there exists an isomorphism such that
    \begin{equation*}
        \H \iso \H_{X'} \otimes \H_X,
        \qquad \A \iso \hat\id_{X'} \otimes \O_X,
        \qquad \A' \iso \O_{X'} \otimes \hat\id_X.
    \end{equation*}
\end{proposition}
\noindent
The proof of this proposition will rely on properties of certain projectors
contained in $\A$, and therefore we introduce some nomenclature.
For simplicity, whenever we say ``projector'',
we mean ``orthogonal projector''.
Furthermore, a projector $0 \neq \hat\pi \in \A$
is called \emph{minimal}, if every projector $\hat p \in \A$
with $\hat p\hat\pi = \hat p$
is either $\hat p = 0$ or $\hat p = \hat\pi$.
Note that if $\A = \O$, then the minimal projectors are precisely
the projectors onto one-dimensional subspaces,
and if $\A = \O[\hat\id]$, then the only minimal projector is $\hat\id$.
Such minimal projectors also play a central role in the infinite-dimensional generalization
of the result; see \eg\ \cite{Sorce2024}.
The proof of the proposition relies on the following technical result:
\begin{lemma}{}{minimal_projectors}
    Let $\H$ a be finite-dimensional Hilbert space, let $\A < \O$ be an operator
    subalgebra with trivial centre, \ie\ $\Z = \A \cap \A' = \O[\hat\id]$,
    and let $\hat\pi_1,\dots,\hat\pi_m \in \A$ be a maximal set
    of non-zero and mutually orthogonal projectors.
    Then:
    \begin{enumerate}
        \medskip
        \item\label{item:minimal} The projectors $\hat\pi_1,\dots,\hat\pi_m$ are all minimal
            and $\sum_i \hat\pi_i = \hat\id$.

        \medskip
        \item\label{item:subspaces} The subspaces $\hat\pi_j\A\hat\pi_i$ are one-dimensional.

        \medskip
        \item\label{item:matrix_elements} It is possible to choose non-zero operators
            $\hat V_{ji} \in \hat\pi_j\A\hat\pi_i$ such that
            $\hat V_{jk} \hat V_{li} = \delta_{kl} \hat V_{ji}$,
            $\hat V_{ji}^\dagger = \hat V_{ij}$,
            and $\hat\pi_j = \hat V_{ji} \hat\pi_i \hat V_{ji}^\dagger$.
            In particular, $\hat V_{ii} = \hat\pi_i$
            and $\hat V_{ji} \hat\pi_i \hat V_{ji}^\dagger = \hat\pi_j$.
    \end{enumerate}
\end{lemma}
\begin{proof}[Proof of lemma~\ref{lem:minimal_projectors}]
    We proceed in five steps.
    First, we prove property~\ref{item:minimal}.
    Second, we show that $\hat\pi_i\A\hat\pi_i$ is one-dimensional
    and spanned by $\hat\pi_i$.
    Third, we find that whenever $\hat\pi_j\A\hat\pi_i$ for $i\neq j$
    is not zero-dimensional, it must be one-dimensional.
    Fourth, we establish that the zero-dimensional case is actually forbidden
    by the triviality of $\Z$, showing property~\ref{item:subspaces}.
    Finally, we construct the operators required by property~\ref{item:matrix_elements}.

    \emph{First step.}
    If any of the $\hat\pi_i$ were not minimal,
    \ie\ if there is a projector $0 \neq \hat p\in\A$ with $\hat\pi_i \hat p
    = \hat p \neq \hat\pi_i$,
    then we could replace $\hat\pi_i$ in the list with the two projectors
    $\hat p, \hat\pi_i - \hat p \in \A$.
    The new list would again be mutually orthogonal
    but one element longer, violating maximality.
    Hence, all $\hat\pi_i$ must be minimal projectors.
    Furthermore, $\hat\id - \sum_i \hat\pi_i \in \A$,
    is a projector and orthogonal to all $\hat\pi_i$.
    If it were not zero, it could be added to the list,
    again breaking maximality.
    Thus, $\sum_i \hat\pi_i = \hat\id$.
    This proves property~\ref{item:minimal}.

    \emph{Second step.}
    Clearly, $\hat\pi_i \in \hat\pi_i\A\hat\pi_i$.
    Assume, by contradiction, that there exists
    $\hat\pi_i \not\reflectbox{\,$\propto$\!\!} \hat A \in \hat\pi_i\A\hat\pi_i$.
    Let $\hat A_\pm := \hat A \pm \hat A^\dagger$.
    One easily checks that $\hat A_\pm$ are both normal
    ($\hat A_\pm\hat A_\pm^\dagger = \hat A_\pm^\dagger \hat A_\pm$),
    hence diagonalizable.
    Also, $(\hat\pi_i\A\hat\pi_i)^\dagger = \hat\pi_i\A\hat\pi_i$,
    so $\hat A_\pm \in \hat\pi_i\A\hat\pi_i$.
    We diagonalize $\hat A_\pm$, obtaining a set of spectral projectors
    for each, which, without loss of generality, can be assumed to be minimal
    (just split large eigenspaces into smaller ones with same eigenvalue).
    These minimal spectral projectors cannot be equal to $\hat\pi_i$,
    since otherwise $\hat A_\pm \propto \hat\pi_i$
    ($\hat A_\pm$ being linear combinations of their spectral projectors,
    and if one spectral projector already equals $\hat\pi_i$,
    there cannot be any others),
    and hence $\hat A \propto \hat\pi_i$.
    Thus, the images of the spectral projectors must be strictly contained
    in the image of $\hat\pi_i$,
    implying that $\hat\pi_i$ is not minimal, a contradiction.
    Thus, such an $\hat A$ cannot exist,
    proving that $\hat\pi_i\A\hat\pi_i$ is one-dimensional
    and spanned by $\hat\pi_i$.
    
    \emph{Third step.}
    Assume that $\hat\pi_j\A\hat\pi_i \neq \{0\}$,
    and take $0 \neq \hat X \in \hat\pi_j\A\hat\pi_i$.
    We have $\hat X = \hat X\hat\pi_i = \hat\pi_j\hat X$,
    and hence $\hat X^\dagger\hat X = \hat\pi_i\hat X^\dagger\hat X\hat\pi_i
    \in \hat\pi_i\A\hat\pi_i$.
    Since $\hat X^\dagger\hat X > 0$,
    and taking into account the previous step,
    we find that $\hat X^\dagger\hat X = \lambda \hat\pi_i$.
    Without loss of generality we rescale $\hat X$ so that $\lambda = 1$,
    \ie\ $\hat X^\dagger\hat X = \hat\pi_i$.
    Furthermore, $\hat X\hat X^\dagger \in \hat\pi_j\A\hat\pi_j$,
    and hence analogously there exists $\mu > 0$ such that
    $\hat X\hat X^\dagger = \mu\hat\pi_j$.
    But
    \begin{equation*}
        \hat\pi_i = \hat X^\dagger\hat X = \hat X^\dagger\hat\pi_j\hat X
        = \frac{1}{\mu} \hat X^\dagger\hat X\hat X^\dagger\hat X
        = \frac{1}{\mu} \hat\pi_i\hat\pi_i = \frac{1}{\mu}\hat\pi_i,
    \end{equation*}
    so that $\mu = 1$.
    In conclusion,
    \begin{equation*}
        \hat\pi_j = \hat X\hat\pi_i\hat X^\dagger,
        \qquad \hat\pi_i = \hat X^\dagger\hat\pi_j\hat X.
    \end{equation*}
    We have found that $\hat X$ is an isometry onto the image of $\hat\pi_j$
    when restricted to the image of $\hat\pi_i$.
    The same is true for $\hat X^\dagger$ but with $i$ and $j$ exchanged.
    To show that $\hat\pi_j\A\hat\pi_i$ is one-dimensional,
    we assume towards contradiction that there exists an operator
    $\hat X \not\reflectbox{\,$\propto$\!\!} \hat A \in \hat\pi_j\A\hat\pi_i$.
    Without loss of generality, we can rescale it to have
    the same properties as $\hat X$.
    Because $\hat X = \hat X \hat\pi_i$ and $\hat A = \hat A \hat\pi_i$,
    it follows that $\hat X \hat\pi_i \not\propto \hat A \hat\pi_i$.
    Thus, there exists a state $\ket{\psi}$
    with $\hat\pi_i\ket{\psi} = \ket{\psi}$
    and such that $\hat A \ket{\psi} \not\propto \hat X\ket{\psi}$.
    Now since $\hat A$, $\hat X$ and their adjoints are isometries
    as described above, it follows that
    $\hat X^\dagger\hat A\sket{\psi} \not\propto
    \hat X^\dagger\hat X\sket{\psi} = \hat\pi_i\sket{\psi} = \sket{\psi}$.
    But $\hat X^\dagger\hat A \in \hat\pi_i\A\hat\pi_i$,
    implying $\hat X^\dagger\hat A \propto \hat\pi_i$
    and therefore $\hat X^\dagger\hat A \ket{\psi} \propto \ket{\psi}$,
    a contradiction.
    We conclude that $\hat A$ cannot exist, and hence that
    $\hat\pi_j\A\hat\pi_i$ is one-dimensional, spanned by $\hat X$.

    \emph{Fourth step.}
    Next, we note that if $\hat\pi_j\A\hat\pi_i \neq \{0\}$,
    then so is $\hat\pi_i\A\hat\pi_j = (\hat\pi_j\A\hat\pi_i)^\dagger
    \neq \{0\}$, and if furthermore $\hat\pi_k\A\hat\pi_j \neq \{0\}$,
    then also $\hat\pi_k\A\hat\pi_i \neq \{0\}$.
    The second statement follows, because, according to the previous step,
    we can find $\hat X_{ji} \in \hat\pi_j\A\hat\pi_i$
    and $\hat X_{kj} \in \hat\pi_k\A\hat\pi_j$
    such that $\hat X_{ji} \hat\pi_i \hat X_{ji}^\dagger
    = \hat\pi_j$ and $\hat X_{kj}\hat\pi_j\hat X_{kj}^\dagger = \hat\pi_k$;
    therefore $\hat X_{kj} \hat X_{ji} \hat\pi_i (\hat X_{kj} \hat X_{ji})
    = \hat\pi_k$, meaning that $\hat X_{kj} \hat X_{ji}
    \in \hat\pi_k\A\hat\pi_i$ is non-zero. 
    We now consider the relation $\sim$ on indices $i,j \in \{1,\dots,m\}$
    defined by
    \begin{equation*}
        j \sim i \quad\iff\quad \hat\pi_j \A \hat \pi_i \neq \{0\}.
    \end{equation*}
    From the properties of the subspaces $\hat\pi_j \A \hat\pi_i$
    we have shown so far, it follows that $\sim$ is an equivalence relation.
    Let $[i]$ be the equivalence class of $i$,
    and let $\{i_a\}_{a=1}^q$ be a set of representatives for each class.
    We then define the projectors
    $\hat\Pi_i := \sum_{j\in [i]} \hat\pi_j \in \A$.
    With them, we can decompose any operator $\hat A\in\A$ as
    \begin{equation*}
        \hat A = \sum_{i,j = 1}^n \hat\pi_j \hat A \hat\pi_i
        = \sum_{a=1}^q \, \sum_{j,i \in [i_a]}
        \hat\pi_j\hat A\hat\pi_i
        = \sum_{a=1}^q \hat\Pi_{i_a} \hat A \hat\Pi_{i_a}.
    \end{equation*}
    This decomposition implies
    \begin{equation*}
        \hat\Pi_{i_1} \hat A = \hat\Pi_{i_1} \hat A \hat\Pi_{i_1}
        = \hat A \Pi_{i_1}
    \end{equation*}
    for every $\hat A \in \A$, \ie\ $\hat\Pi_{i_1} \in \Z$
    (the same is true for every other $i_a$ besides $i_1$,
    but one contradiction will suffice).
    Towards contradiction, we now assume that there exists a pair $(k,l)$
    of indices such that $\hat\pi_l\A\hat\pi_k = \{0\}$.
    Then, there must be more than one equivalence class of the relation $\sim$
    since $k \not\sim l$.
    Thus, $\hat\Pi_{i_1} \neq \hat\id$,
    making $\hat\Pi_{i_1}$ an element of the centre $\Z$ of $\A$
    which is not a multiple of the identity.
    This contradicts the triviality of $\Z$,
    and hence the pair $(l, k)$ cannot exist.
    This proves property~\ref{item:subspaces}: all subspaces $\hat\pi_j\A\hat\pi_i$
    are one-dimensional and spanned by an element like $\hat X$
    in the previous step.

    \emph{Fifth step.}
    For all $i \neq 1$, we now choose an operator $\hat V_{i1}
    \in \hat\pi_i\A\hat\pi_1$ with the properties of $\hat X$
    in step three (\ie\ it has the isometry properties discussed previously).
    We also set $\hat V_{11} := \hat\pi_1$.
    Furthermore, we define
    \begin{equation*}
        \hat V_{ji} := \hat V_{j1} \hat V_{i1}^\dagger
    \end{equation*}
    for all pairs $j$, $i$.
    Note that this is well-defined as it matches our previous choices for $i=1$.
    It is now easy to verify that the operators $\hat V_{ji}$ satisfy
    the requirements in property~\ref{item:matrix_elements}
    using the isometry properties we showed in the third step.
\end{proof}

\medskip
\noindent
Because $\A = \sum_{i,j = 1}^m \hat\pi_j \A \hat\pi_i$,
it follows from lemma~\ref{lem:minimal_projectors}
that the operators $\hat V_{ji}$ form a basis of $\A$,
and furthermore, that they behave under multiplication like
the standard basis matrices of $\bbC^{m\times m}$:
the standard basis matrix $E_{ji}$ consists of zeroes except for a one
at the intersection of the $j$th row and $i$th column, and
$E_{ji} \cdot E_{lk} = \delta_{il}\, E_{jk}$
as well as $E_{ji}^\dagger = E_{ij}$.
Hence, we have an algebra isomorphism $\A \to \bbC^{m\times m}$,
$\hat V_{ji} \mapsto E_{ji}$.
Since $\bbC^{m\times m} \iso \O_X$
for some $m$-dimensional Hilbert space $\H_X$,
we even have $\A \iso \O_X$.
Since all projectors $\hat\pi_i$ are isomorphic
(property~\ref{item:matrix_elements} in lemma~\ref{lem:minimal_projectors}) with common dimension $n$,
it is natural to interpret $\A \iso \hat\id_{X'} \otimes \O_X$,
where $\H_{X'}$ is an $n$-dimensional Hilbert space
and $\H \iso \H_{X'} \otimes \H_X$.
Then, $n$ is the multiplicity with which $\O_X$
occurs in the full algebra $\O$.
This is the intuitive reason behind the isomorphism relations of proposition
\ref{prop:algebra_decomposition_1}.

\medskip
\begin{proof}[Proof of proposition~\ref{prop:algebra_decomposition_1}]
    Let $\hat\pi_1,\dots,\hat\pi_m \in \A$
    be a maximal set of mutually orthogonal projectors.
    Applying lemma~\ref{lem:minimal_projectors}, we obtain operators
    $\hat V_{ji} \in \hat\pi_j\A\hat\pi_i$,
    and as already remarked, property~\ref{item:matrix_elements} 
    in lemma~\ref{lem:minimal_projectors} tells us
    that the dimensions of the images of those projectors are all equal:
    \begin{equation*}
        \dim\hat\pi_i = \tr(\hat\pi_i)
        = \tr(\hat V_{i1}\hat\pi_1\hat V_{i1}^\dagger)
        = \tr(\hat\pi_1\hat V_{i1}^\dagger\hat V_{i1})
        = \tr(\hat\pi_1\hat\pi_1)
        = \tr(\hat\pi_1)
        =: n > 0.
    \end{equation*}
    Since the projectors form a complete set 
    (property~\ref{item:minimal} in lemma~\ref{lem:minimal_projectors}),
    $\dim\H = m \cdot n$.
    We thus have $\H \iso \H_{X'} \otimes \H_X$, where $\H_X$ and $\H_{X'}$
    are $m$- and $n$-dimensional Hilbert spaces respectively.
    Guided by the intuition above, we will now construct an explicit
    isomorphism $\hat\phi : \H_{X'} \otimes \H_X \to \H$
    such that the other two (algebra) isomorphism relations of the proposition
    are also satisfied.

    For this, we choose orthonormal bases
    $\{\ket{i}_X\}_{i=1}^m \subset \H_X$,
    $\{\ket{i}_{X'}\}_{i=1}^n \subset \H_{X'}$
    and $\{\ket{i}_1\}_{i=1}^n \subset \text{image}(\hat\pi_1) \subset \H$
    of $\H_X$, $\H_{X'}$ and the image of $\hat\pi_1$ respectively.
    The map $\hat\phi$ is then defined by 
    \begin{equation*}
        \hat\phi \ket{j}_{X'} \ket{i}_X := \hat V_{i1} \ket{j}_1
    \end{equation*}
    and the requirement that $\hat\phi$ is linear.
    It preserves the inner product,
    $\bra{l}_{X'}\bra{k}_X \hat\phi^\dagger\hat\phi \ket{j}_{X'} \ket{i}_X
    = \bra{l}_1 \hat V_{k1}^\dagger \hat V_{i1} \ket{j}_1
    = \delta_{ki} \bra{l}_1 \hat\pi_1 \ket{j}_1
    = \delta_{ki}\delta_{lj}$,
    and, thus, for dimensional reasons, $\hat\phi$ must be an isomorphism.
    Furthermore,
    \begin{equation*}
        \hat\phi \bigl( \hat\id_{X'} \otimes \sketbra{j}{i}_X \bigr)
        \hat\phi^\dagger = \sum_k \hat\phi\bigl(
        \sketbra{k}{k}_{X'} \otimes \sketbra{j}{i}_X
        \bigr) \hat\phi^\dagger
        = \sum_k \hat V_{j1} \ketbra{k}{k}_1 \hat V_{i1}^\dagger
        = \hat V_{j1} \hat\pi_1 \hat V_{i1}^\dagger
        = \hat V_{j1} \hat V_{i1}^\dagger
        = \hat V_{ji}.
    \end{equation*}
    Because the operators $\hat V_{ji}$ form a basis of $\A$
    (property~\ref{item:subspaces} in lemma~\ref{lem:minimal_projectors}),
    it follows that $\hat\phi(\hat\id_{X'} \otimes \O_X)\hat\phi^\dagger = \A$.
    Thus, $\hat\phi$ also induces the second
    isomorphism relation of the proposition.
    To check the third relation, we note that $\hat\phi$ must map the commutant
    of $\hat\id_{X'} \otimes \O_X$ to the commutant $\A'$ of $\A$
    since commutation relations are conserved by isomorphisms.
    But the commutant of $\hat\id_{X'} \otimes \O_X$ is easily seen to be
    $\O_{X'} \otimes \hat\id_X$.
    Hence, $\hat\phi(\O_{X'} \otimes \hat\id_X) \hat\phi^\dagger =
    \A'$, which is the third isomorphism relation.
\end{proof}

\begin{proposition}{}{algebra_decomposition_2}
    Let $\H$ be a finite-dimensional Hilbert space, and let $\A < \O$ be an operator algebra.
    There exists an isomorphism such that
    \begin{equation*}
        \H \iso \bigoplus_k \H_{*,k} \otimes \H_k,
        \qquad \A \iso \bigoplus_k \hat\id_{*,k} \otimes \O_k,
        \qquad \A' \iso \bigoplus_k \O_{*,k} \otimes \hat\id_k,
    \end{equation*}
    where $k$ labels the simultaneous eigenspaces
    $\J_k \iso \H_{*,k} \otimes \H_k$
    of all elements in the centre $\Z := \A \cap \A'$.
\end{proposition}
\begin{proof}
    The centre $\Z$ is Abelian, as it must commute with itself.
    But then for $\hat Z \in \Z$,
    $\hat Z^\dagger\hat Z = \hat Z\hat Z^\dagger$,
    and thus $\hat Z$ is diagonalizable.
    Since all elements of $\Z$ commute, we can simultaneously diagonalize them.

    Let $k$ label the simultaneous eigenspaces $\J_k \subset \H$ of $\Z$:
    $\H = \bigoplus_k \J_k$.
    The projector $\hat\pi_k$ onto the eigenspace corresponding to $k$
    is contained in $\Z$:
    the eigenprojectors of an operator are polynomials
    of the operator, derivable from the characteristic polynomial,
    and $\Z$ is an operator algebra and thus closed
    under taking of polynomials.
    $[\A,\hat\pi_k] = [\A',\hat\pi_k] = 0$ 
    implies $\A\hat\pi_k = \hat\pi_k\A\hat\pi_k =: \A_k$
    and similarly $\A'\hat\pi_k = \hat\pi_k\A'\hat\pi_k =: \A_{*,k}$.
    Both $\A = \bigoplus_k \A_k$ and $\A' = \bigoplus_k \A_{*,k}$
    are thus block-diagonal with respect to $k$,
    and hence the commutant of $\A_k$ with respect to all operators
    on $\J_k$ is $\A_{*,k}$.
    Furthermore, the centre of $\A_k$, $\A_k \cap \A_{*,k}$,
    must be trivial, for otherwise the eigenspace $\J_k$
    could be further divided into simultaneous eigenspaces.

    Within each $k$-sector, the assumptions of proposition
   ~\ref{prop:algebra_decomposition_1} now apply,
    with the simultaneous eigenspace corresponding to $k$
    in the role of $\H$, $\A_k$ in the role of $\A$,
    and $\A_{*,k}$ in the role of $\A'$.
    Applying the proposition for each $k$ separately
    yields the three isomorphism relations,
    as well as $\J_k \iso \H_{*,k} \otimes \H_k$.
\end{proof}

\medskip
\begin{proof}[Proof of proposition~\ref{prop:algebra_decomposition_full}]
    Apply proposition~\ref{prop:algebra_decomposition_2}
    to obtain an isomorphism
    \begin{equation*}
        \H \iso \bigoplus_k \H_{*,k} \otimes \H_k,
        \qquad \A \iso \bigoplus_k \hat\id_{*,k} \otimes \O_k,
        \qquad \A' \iso \bigoplus_k \O_{*,k} \otimes \hat\id_k,
    \end{equation*}
    in line with property~\ref{item:k_labels} in the proposition
    (except that $\H_{*,k}$ is not yet a tensor product).
    
    The algebra $\A'$ contains the representation operators $\hat U(g)$
    because $\A$ is $G$-invariant.
    Thus, under the above isomorphism,
    $\hat U(g) \iso \bigoplus_k \hat U_{*,k}(g) \otimes \hat\id_k$,
    and $\hat U_{*,k}$ must be unitary since $\hat U$ is unitary.
    We can thus apply the Peter-Weyl theorem
    (proposition~\ref{prop:Peter_Weyl})
    to $\H_{*,k}$ with representation $\hat U_{*,k}$.
    This results in a further isomorphism
    \begin{equation*}
        \H_{*,k} \iso \bigoplus_r \H_r \otimes \H_{r,k},
        \qquad \hat U_{*,k}(g) \iso \bigoplus_r \hat U_r(g)
        \otimes \hat\id_{r,k},
    \end{equation*}
    where $r$ enumerates all inequivalent irreducible unitary
    representations of $G$ with representation spaces $\H_r$
    and actions $\hat U_r$, and $\H_{r,k}$ are the multiplicity spaces.
    This is in line with property~\ref{item:r_labels} in the proposition.
    We now concatenate the two isomorphisms to find the decomposition
    of the proposition.
    Note that we can exchange the order of the direct sums over $k$ and $r$
    and hence combine them into a single sum
    only in the decompositions of $\H$, $\hat U(g)$ and $\A'$,
    but not in the decomposition of $\A$:
    choosing an operator in the algebra
    $\bigoplus_k\bigl(\bigoplus_r \hat\id_r \otimes \hat\id_{r,k} \bigr) \otimes \O_k$
    corresponds to choosing an operator in each algebra summand, \ie\
    in $\O_k$ for each $k$ labelling a sector of the centre;
    meanwhile, choosing an operator in the algebra
    $\bigoplus_{k,r} \hat\id_r \otimes \hat\id_{r,k} \otimes \O_k$
    also corresponds to choosing an operator in each algebra summand,
    but now $\hat\id_r \otimes \hat\id_{r,k} \otimes \O_k$ is an algebra summand
    for each pair $(r, k)$,
    and the same $k$-value may be paired with several $r$-values;
    this makes the latter algebra generally larger than the former.

    Finally, let $\bigoplus_r \hat A_r \otimes \hat B_{r,k}
    \in \overline\A_k$ ($\overline\A \subset \A'$ is block-diagonal
    with respect to $k$ because $\A'$ is).
    $\overline\A = \A' \cap \mathsf G[\O]$ commutes with the representation,
    and so $[\hat A_r, \hat U_r(g)] = 0$ for all $r$.
    But because $\hat U_r$ is irreducible, Schur's lemma implies that
    $\hat A_r \propto \hat\id_r$.
    Therefore, $\overline\A_k \iso \bigoplus_r \hat\id_r \otimes \O_{r,k}$.
\end{proof}

\medskip
\paragraph*{Corollaries.}
An important corollary of proposition~\ref{prop:algebra_decomposition_2}
is the von Neumann \emph{double commutant theorem}
(usually stated for infinite dimensions~\cite{Neumann1930,Davidson1996}):
\begin{proposition}{}{double_commutant}
    If $\H$ is a finite-dimensional Hilbert space and $\A < \O$ an operator subalgebra,
    then $\A'' = \A$.
\end{proposition}
\begin{proof}
    $\A''$ must be block-diagonal with respect to $k$ because
    $\Z = \A \cap \A'$ is contained in the centre of $\A''$.
    The commutant of $\O_{*,k} \otimes \hat\id_k$
    on $\H_{*,k} \otimes \H_k$ is simply $\hat\id_{*,k} \otimes \O_k$,
    leading to $\A'' = \A$.
\end{proof}

\medskip
\paragraph*{Generalizations.}
We now provide a rough roadmap to how one could generalize our results
rigorously to infinite-dimensional spaces and non-compact symmetry groups.
However, since following this roadmap would incur considerable mathematical
difficulties and subtleties, we do not attempt it in this work;
we content ourselves with the roadmap.

Proposition~\ref{prop:algebra_decomposition_1}
generalizes to $\dim\H = \infty$,
if $\O$ is the von Neumann algebra of bounded operators
and $\A$ is a von Neumann algebra of \emph{type I}~%
\cite{Neumann1930,Murray1936}.
Intuitively, the finite-dimensional considerations of minimal projectors
still largely work in infinite-dimensional type-I algebras,
because they contain minimal projectors which can be identified with
rank-1 orthogonal projectors in $\O_X$~\cite{Sorce2024}.
Similarly, proposition~\ref{prop:algebra_decomposition_2} generalizes
under the same assumptions ($\A$ bounded and type-I),
but direct sums have to be replaced by more general combinations
of direct sums and direct integrals~\cite{Neumann1949}.
Finally, proposition~\ref{prop:double_commutant},
the von Neumann double commutant theorem,
actually generalizes without restriction on the type~\cite{Neumann1930}.
See~\cite{Sorce2024} for an overview of von Neumann algebras
and their decompositions.
Finally, the generalization of proposition
\ref{prop:algebra_decomposition_full}
follows from the generalizations of the previous propositions,
combined with that of the Peter-Weyl theorem, 
proposition~\ref{prop:Peter_Weyl},
which we already discussed in appendix~\ref{app:representation_theory_groups}.
Notably, here we only need the first part of the Peter-Weyl theorem
(complete reducibility of unitary representations).

Roughly speaking, proposition~\ref{prop:algebra_decomposition_full}
and, hence theorem~\ref{thm:QRF_POV}, generalize with a combination of direct integrals and sums replacing
direct sums if both $G$ and $\A$ are type I in their respective sense.
Particularly, the involved algebras are bounded.
This, of course, excludes some physically interesting operators,
such as the position $\hat x$ of a quantum particle.
However, we can approximate unbounded operators by bounded ones:
\eg, the position $\hat x$, taken as a POVM with improper Dirac-$\delta$
elements can be approximated by a POVM with well-defined projectors,
such as Gaussian states of arbitrarily narrow width.

Essentially, the only real obstacle to generalizing our framework to
infinite dimensions and non-compact symmetry groups
is that the involved algebras must be type I:
for type II and type III algebras, the notion of perspective of a QRF,
or rather, the implementation of the
\hyperref[ppl:perspective]{perspective principle},
would have to be quite significantly extended beyond what we did in this work.
For an algebra $\A$ of type II, it is possible to show that
$\A \iso \O_X \otimes \tilde\A_{X'}$, acting on a bipartite Hilbert space
$\H_X \otimes \H_{X'}$ (with $\H_X$ as large as possible),
and where $\tilde\A_{X'}$ is a type II von Neumann algebra
acting on $\H_{X'}$~\cite{Murray1936,Murray1943};
see again~\cite{Sorce2024} for an overview.
It is conceivable that our interpretation of the
\hyperref[ppl:perspective]{perspective principle}
and hence our notion of perspective of a QRF can be generalized
in this direction.
For the type III case, much less can be said~\cite{Murray1936,Murray1943}.
But type II is perhaps enough:
recent work on so-called \emph{type reduction} in quantum field theories~%
\cite{Witten2022,Chandrasekaran2023,Fewster2024,DeVuyst2025,DeVuyst2025a}
has shown that certain type-III algebras of quantum field observables
can be relationalized through the introduction of an additional quantum system,
specifically a \emph{quantum clock}, resulting in type-II algebras.
The quantum clock can be seen as a specific QRF for time evolution,
\ie\ $G = (\bbR,+)$,
and the relationalization procedure resembles
our relationalization map $\mathsf R_A$.
It is conceivable that type reduction occurs quite generally
for algebras generated from relationalized operators,
regardless of symmetry group.

\section{Regular QRFs}
\label{app:regular}

\noindent
This appendix provides details as well as proofs surrounding regular QRFs.
As explained in the main text, there are many realizations
of an ideal QRF, one of which being \emph{regular} QRFs.
Regular QRFs have the advantages of being the most natural quantum extension
of classical RFs, and of their relatively simple mathematical structure.
Thus, it often makes sense to focus on regular QRFs.
Regular QRFs are often considered in the literature~%
\cite{Aharonov1984,Angelo2011,Giacomini2019,Vanrietvelde2020,Hamette2020,
Hamette2021,Glowacki2023a,CastroRuiz2025a,CastroRuiz2025}.

\medskip
\paragraph*{Construction of a Regular QRF.}
Recall (see section~\ref{sec:QRF_POV})
that a regular QRF for the symmetry group $G$
is a quantum system with Hilbert space $\H \iso L^2(G)$,
the action of $G$ is the left-regular representation $\hat L$,
and one takes the improper projectors $\ketbra{g}{g}$
onto Dirac-distributions as orientation POVM elements.
We call such QRFs ``regular'' because they use the (left)
regular representation.

The improper projectors $\hat \gamma(g) := \ketbra{g}{g}$, $g\in G$,
indeed define a covariant POVM:
\begin{equation*}
    \mathsf U(g')[\hat \gamma(g)] = \hat \gamma(g'g)
    \quad \forall\, g,g' \in G,
\end{equation*}
and completeness $\int_G \der g\, \ketbra{g}{g} = \hat\id$
was shown in appendix \ref{app:representation_theory_groups}.
Thanks to \cref{eq:app_regular_overlaps},
the orientation POVM elements are orthogonal,
and the QRF therefore ideal:
we can just take the improper states corresponding to the POVM elements
to obtain states with definite orientations.

\medskip
\paragraph*{Regular QRF Jumps and Transformations.}
We now turn to show the results surrounding regular QRFs
needed for the main text.
We begin by deriving the properties of the map $\hat W$,
which make $\hat W^\dagger$ a jump into the perspective of a regular QRF;
this result was first shown in~\cite{CastroRuiz2025a}.
\begin{proposition}{Jumps for Regular QRFs}{jump}
    Let $A$ be a regular QRF for the symmetry group $G$,
    and $S$ an arbitrary quantum system with unitary representation
    $\hat U_S$ of $G$,
    both described from the perspective of the laboratory $L$.
    Then the map
    \begin{equation}
        \hat W := \int_G\der g\, \ketbra{g}{g}_A \otimes \hat U_S(g)
        \label{eq:app_regular_inverse_jump_core}
    \end{equation}
    is unitary. Furthermore,
    \begin{equation}
        \mathsf W\bigl[ \hat U_A(g) \otimes \hat\id_S \bigr]
        = \hat U_A(g) \otimes \hat U_S(g)
        \quad \forall\,g \in G,
        \label{eq:regular_jump_representation}
    \end{equation}
    and it induces the following isomorphism relations of operator algebras:
    \begin{align*}
        \mathsf G_A[\O_A] \otimes \O_S &\iso \A^{\text{inv}(AS)|L}, \\
        \hat\id_A \otimes \O_S &\iso \A^{S:A|L}, \\
        \mathsf G_A[\O_A] \otimes \hat\id_S &\iso
        \A^{\overline{S:A}|L}.
    \end{align*}
\end{proposition}
\noindent
We provide our own proof:
\begin{proof}
    Thanks to $\int_G\der g\,\sketbra{g}{g}_A = \hat\id_A$,
    and \cref{eq:app_regular_overlaps}, $\hat W$ is indeed unitary
    with inverse $W^\dagger$:
    \begin{equation*}
        \left( \int_G\der g\, \ketbra{g}{g}_A \otimes \hat U_S(g) \right)
        \left( \int_G\der g'\,
        \sketbra{g'}{g'}_A \otimes \hat U_S(g'^{-1}) \right)
        = \int_G\der g\, \ketbra{g}{g}_A \otimes \hat\id_S
        = \hat\id_{AS}.
    \end{equation*}
    Next, we compute
    \begin{equation*}
        \hat W \bigl( \hat U_A(g) \otimes \hat\id_S \bigr) W^\dagger
        = \int_G\der g'\,\der g''\,
        \sket{g'}\! \sbraket{g'}{gg''}\! \sbra{g''}_A
        \otimes \hat U_S(g'g''^{-1})
        = \int_G\der g''\, \sketbra{gg''}{g''}_A \otimes \hat U_S(g)
        = \hat U_A(g) \otimes \hat U_S(g).
    \end{equation*}
    Because unitaries preserve operator algebra structures,
    proving the three isomorphism relations reduces to showing that
    the superoperator $\mathsf W$ surjectively maps the left-hand sides
    into the right-hand sides.
    We begin with
    \begin{align*}
        \hat W (\mathsf G_A[\O_A] \otimes \O_S) \hat W^\dagger
        &= \frac{1}{|G|} \int_G\der g\,
        \hat W \bigl( \hat U_A(g) \otimes \hat\id_S \bigr) \hat W^\dagger
        \hat W \O_{AS}\hat W^\dagger
        \hat W \bigl( \hat U_A^\dagger(g) \otimes \hat\id_S \bigr)
        \hat W^\dagger \\
        &= \frac{1}{|G|} \int_G\der g\, \hat U_{AS}(g)
        \O_{AS} \hat U^\dagger_{AS}(g)
        = \mathsf G_{AS}[\O_{AS}],
    \end{align*}
    showing the first relation.
    In the second step, we used that $\O_{AS}$ is invariant under any unitary
    as well as \cref{eq:regular_jump_representation}.
    For the second isomorphism relation,
    let $\hat f_S$ be any operator on $S$. Then,
    \begin{equation*}
        \mathsf W[\hat\id_A \otimes \hat f_S]
        = \int_G\der g\,\der g'\, \sketbra{g}{g} \hat\id \sketbra{g'}{g'}_A
        \otimes \hat U_S(g) \hat f_S \hat U_S^\dagger(g')
        = \int_G \der g\, \ketbra{g}{g}_A \otimes \mathsf U_S(g)[\hat f_S].
    \end{equation*}
    But operators of this form generate $\A^{S:A|L}$,
    and hence it follows that $\mathsf W\bigl[\hat\id_A \otimes \O_S\bigr]
    = \A^{S:A|L}$.
    This shows the second isomorphism relation.
    Finally, it is clear that $\mathsf G_A[\O_A] \otimes \hat\id_S$
    is the commutant of $\hat\id_A \otimes \O_S$
    relative to $\mathsf G_A[\O_A] \otimes \O_S$.
    Since $\hat W$ preserves the algebra structure, it must map
    $\mathsf G_A[\O_A] \otimes \hat\id_S$ to $\overline{\A^{S:A|L}}
    = \A^{\overline{S:A}|L}$.
\end{proof}

\medskip
Furthermore, $\hat W$ also conserves subsystem structure in $S$:%
\footnote{
    We will later show a simultaneously more and less general version of this
    result in proposition~\ref{prop:algebra_decomposition_S_Abelian};
    more general since it applies to all rank-one QRFs
    and less general since it requires an Abelian group $G$.
}
\begin{proposition}{Regular QRFs Resolve Subsystem Structures}
    {regular_subsystem_structure}
    Let $A$ be a regular QRF for the symmetry group $G$
    and $S = S_1S_2$ a bipartite quantum system with Hilbert space
    $\H_S = \H_{S_1} \otimes \H_{S_2}$ and carrying a unitary representation
    $\hat U_S = \hat U_{S_1} \otimes \hat U_{S_2}$ of $G$,
    described from the perspective of the laboratory $L$.
    Define
    \begin{equation*}
        \A^{S_1:A|L} := \mathsf R_A[\O_{S_1} \otimes \hat\id_{S_2}],
        \qquad \A^{S_2:A|L} := \mathsf R_A[\hat\id_{S_1} \otimes \O_{S_2}].
    \end{equation*}
    Then, $\hat W$, as defined in \cref{eq:app_regular_inverse_jump_core},
    induces the algebra isomorphism relations
    \begin{align*}
        \hat\id_A \otimes \O_{S_1} \otimes \hat\id_{S_2}
        &\iso \A^{S_1:A|L}, \\
        \hat\id_A \otimes \hat\id_{S_1} \otimes \O_{S_2}
        &\iso \A^{S_2:A|L}.
    \end{align*}
\end{proposition}
\begin{proof}
    For any operator $\hat f_{S_1}$ on $\H_{S_1}$, we compute
    \begin{equation*}
        \hat W\bigl( \hat\id_A \otimes \hat f_{S_1} \otimes \hat\id_{S_2}
        \bigr)\hat W^\dagger
        = \int_G\der g\, \sketbra{g}{g}_A
        \otimes \hat U_{S_1}(g) \hat f_{S_1} \hat U_{S_1}^\dagger(g)
        \otimes \hat\id_{S_2}
        = \mathsf R_A[\hat f_{S_1} \otimes \hat\id_{S_2}]
    \end{equation*}
    and note that $\A^{S_1:A|L}$ is generated by operators of this form.
    The second isomorphism relation follows analogously.
\end{proof}

\medskip
Finally, we derive the QRF transformations.
These transformations were first discussed in~\cite{CastroRuiz2025a}.
\begin{proposition}{QRF Transformations Between Regular QRFs}
    {regular_QRF_transformation}
    Let $A$ and $B$ be regular QRFs for the symmetry group $G$
    and let $Q$ be any quantum system with unitary representation
    $\hat U_Q$ of $G$, all three described from the perspective
    of the laboratory $L$.
    Furthermore, define the QRF jumps
    \begin{gather*}
        \hat V^{\to A} : \H_A \otimes \H_B \otimes \H_Q
        \to \Bigl(\bigoplus_r \H_r \otimes \H_{*,r} \Bigr)
        \otimes \H_B \otimes \H_Q, \\
        \hat V^{\to A} :=
        \int_G\der g\, \hat X_A \ketbra{g}{g}_A \otimes \hat L_B^\dagger(g)
        \otimes \hat U_Q^\dagger(g),
    \end{gather*}
    and
    \begin{gather*}
        \hat V^{\to B} : \H_A \otimes \H_B \otimes \H_Q
        \to \Bigl(\bigoplus_r \H_r \otimes \H_{*,r} \Bigr)
        \otimes \H_A \otimes \H_Q, \\
        \qquad \hat V^{\to B} := \mathsf T_{AB}\bigl[\hat V^{\to A}\bigr]
        = \int_G\der g\, \hat L_A^\dagger(g) \otimes \hat X_B \sketbra{g}{g}_B
        \otimes \hat U_Q^\dagger(g),
    \end{gather*}
    where $\hat X_A : \H_A \to \bigoplus_r \H_r \otimes \H_{*,r}$
    is an isomorphism obtained from the Peter-Weyl theorem
    (proposition~\ref{prop:Peter_Weyl})
    and $\hat X_B \ket{g}_B := \hat X_A \ket{g}_A$ for all $g \in G$.
    Then, the QRF transformation
    $\hat V^{A\to B} = \hat V^{\to B} (\hat V^{\to A})^\dagger$
    is explicitly given by
    \begin{equation*}
        \hat V^{A\to B} = \int_G\der g\,
        \Bigl( \bigoplus_r \hat\id_r \otimes \hat U^\dagger_{*,r}(g) \Bigr)
        \otimes \sket{g^{-1}}_A \sbra{g}_B \otimes \hat U^\dagger_Q(g).
    \end{equation*}
    In particular, it does not depend on $\hat X_A$ and $\hat X_B$.
\end{proposition}
\begin{proof}
    We compute
    
    \begin{multline*}
        \hat X_B^\dagger \hat V^{\to B} (\hat V^{\to A})^\dagger \hat X_A
        = \int_G\der g\,\der g'\, \hat L_A^\dagger(g) \sketbra{g'}{g'}_A
        \otimes \sketbra{g}{g}_B \hat L_B(g')
        \otimes \hat U_Q^\dagger(g) \hat U_Q(g') \\
        = \int_G\der g\,\der g'\, \sketbra{g^{-1}g'}{g'}_A
        \otimes \sketbra{g}{g'^{-1}g}_B \otimes \hat U_Q^\dagger(g'^{-1}g)
        = \int_G\der g\,\der g'\, \sketbra{g^{-1}}{g'}_A
        \otimes \sketbra{g'g}{g}_B \otimes \hat U^\dagger_Q(g) \\
        = \hat T_{AB} \int_G\der g\, \hat R^\dagger_A(g) \otimes
        \sketbra{g^{-1}}{g}_B \otimes \hat U^\dagger_Q(g),
    \end{multline*}
    where $\hat T_{AB} \ket{g}_A \sket{g'}_B := \sket{g'}_A \ket{g}_B$
    for all $g,g' \in G$ and $\hat R_A$ is the right-regular representation.
    Applying $\hat X_B$ and $\hat X_A$ on either side requires us to compute
    \begin{equation*}
        \hat X_B^\dagger \hat T_{AB}
        \bigl(\hat R^\dagger_A(g) \otimes \sketbra{g^{-1}}{g}_B\bigr) \hat X_A
        = \hat X_A^\dagger \hat R^\dagger_A(g) \hat X_A
        \otimes \sket{g^{-1}}_A\sbra{g}_B.
    \end{equation*}
    According to proposition~\ref{prop:Peter_Weyl}, we have
    \begin{equation*}
        \hat X_A^\dagger \hat R^\dagger_A(g) \hat X_A
        = \bigoplus_r \hat\id_r \otimes \hat U_{*,r}^\dagger(g).
    \end{equation*}
    Combining both computations, the result follows.
\end{proof}

\section{Embedding}
\label{app:embedding}

\noindent
As mentioned in the main text, the map $\hat W^\dagger$
(see \eg\ \cref{eq:app_regular_inverse_jump_core})
used for the jump into a regular QRF
can also be used to obtain a map \emph{resembling}
a jump to a non-ideal, rank-one QRF.
This fact is useful as a stepping stone
to construct actual QRF jumps into non-ideal QRFs
(in appendix~\ref{app:POVs_Abelian}, we do this for non-ideal, rank-one QRFs
for the case of an Abelian symmetry group).
Specifically, one first \emph{embeds} the non-ideal, rank-one QRF $A$
into a regular QRF $\tilde A$:
\begin{proposition}{Embedding}{embedding}
    Let $A$ be a QRF for the symmetry group $G$ 
    with rank-one orientation POVM elements
    $\hat\gamma_A(g) = \frac{1}{c}\sketbra{g}{g}_A$, $c > 0$,
    and unitary representation $\hat U_A(g') \sket{g}_A = \sket{g'g}_A$,
    $\forall\,g,g' \in G$.
    Let $\tilde A$ be a regular QRF for $G$.
    Then the \emph{embedding}
    \begin{equation*}
        \hat E := \frac{1}{\sqrt{c}} \int_G\der g\,
        \sket{g}_{\tilde A} \sbra{g}_A
    \end{equation*}
    is the unique linear map $\hat E : \H_A \to \H_{\tilde A}$
    such that
    \begin{enumerate}
        \medskip
        \item\label{item:isometry} $\hat E$ is an isometry: $\hat E^\dagger \hat E = \hat\id_A$.

        \medskip
        \item\label{item:representation} $\hat E \hat U_A(g) = \hat L_{\tilde A}(g) \hat E$
        for all $g \in G$.
        
        \medskip
        \item\label{item:identity} $\hat E^\dagger \sket{e}_{\tilde A} = d\sket{e}_A$
        for some $d > 0$.
    \end{enumerate}
    \medskip
    In fact, $d = 1/\sqrt c$.
    If $A$ is regular,
    then $\hat E$ is the map $\sket{g}_A \mapsto \sket{g}_{\tilde A}$
    which changes the label from $A$ to $\tilde A$.
\end{proposition}
\noindent
The proof is adapted from~\cite{Garmier2023}:
\begin{proof}
    A general map $\hat E : \H_A \to \H_{\tilde A}$ is of the form
    \begin{equation*}
        \hat E = \int_G\der g\, \sket{g}_{\tilde A} \sbra{\varphi_g}_A,
    \end{equation*}
    where $\sbra{\varphi_g}_A$ are linear forms on $\H_A$,
    \ie\ $\sket{\varphi_g}_A$ are not necessarily normalized
    and potentially even improper states of $A$.
    Requirement~\ref{item:identity} is $\int_G\der g\, \sket{\varphi_g}_A
    \sbraket{g}{e}_{\tilde A} = d \sket{e}_A$ for some $d > 0$,
    and thanks to \cref{eq:app_regular_overlaps} we thus have
    $\sket{\varphi_e}_A = d \sket{e}_A$.
    Requirement~\ref{item:representation} implies
    \begin{equation*}
        \hat U_A(g')^\dagger \sket{\varphi_g}_A
        = \hat U_A(g')^\dagger \hat E^\dagger \sket{g}_{\tilde A}
        = \hat E^\dagger \hat L_{\tilde A}(g')^\dagger \sket{g}_{\tilde A}
        = \hat E^\dagger \sket{g'^{-1}g}
        = \sket{\varphi_{g'^{-1}g}}
    \end{equation*}
    for all $g,g' \in G$.
    In other words, $\sket{\varphi_g}_A = \hat U_A(g) \sket{\varphi_e}_A$
    for all $g \in G$.
    Hence, $\sket{\varphi_g}_A = d \sket{g}_A$ for all $g\in G$.
    Finally, requirement~\ref{item:isometry} gives
    \begin{equation*}
        \hat\id_A = d^2 \int_G\der g\,\der g'\,
        \sket{g'}_A \sbraket{g'}{g}_{\tilde A} \sbra{g}_A
        = d^2 \int_G\der g\, \sketbra{g}{g}_A
        = d^2 c\, \hat\id_A.
    \end{equation*}
    It follows that $d = 1/\sqrt c$ and we find that $\hat E$
    must be of the form given in the statement of the proposition.
     Conversely, one easily shows that $\hat E$ of that form
    satisfies the three requirements.
    If $A$ is regular, then $c = 1$ and
    $\hat E = \int_G\der g\, \sket{g}_{\tilde A} \sbra{g}_A$
    is precisely the map described in the proposition statement.
\end{proof}
\medskip
\noindent
The embedding can then be used to obtain a map resembling a jump
into the non-ideal QRF:
\begin{proposition}{}{general_jump}
    Let $A$ be a QRF for the symmetry group $G$
    with rank-one orientation POVM elements
    $\hat\gamma_A^{|L}(g) = \frac{1}{c} \sketbra{g}{g}_A$, $c > 0$,
    and unitary representation $\hat U_A(g') \sket{g}_A = \sket{g'g}_A$,
    $g,g' \in G$.
    Let $\tilde A$ be a regular QRF,
    and let $\hat E$ be the embedding $\H_A \to \H_{\tilde A}$
    of proposition~\ref{prop:embedding}.
    Furthermore, let $S$ be a quantum system
    with unitary representation~$\hat U_S$ of~$G$.
    Finally, let $\hat W$ be defined as in
    \cref{eq:app_regular_inverse_jump_core}.
    Then:

    \begin{enumerate}
        \medskip
        \item\label{item:projector}
            The map $\hat\pi : \H_{\tilde A} \otimes \H_S
            \to \H_{\tilde A} \otimes \H_S$ defined as
            \begin{equation}
                \hat\pi :=
                \hat W^\dagger \bigl(\hat E\hat E^\dagger \otimes \hat\id_S
                \bigr) \hat W
                = \frac{1}{c} \int_G\der g\,
                \hat R_{\tilde A}(g) \otimes \hat U_S(g)
                \sbraket{e}{g}_A
                \label{eq:app_general_jump_projector}
            \end{equation}
            is the orthogonal projector onto the image of the isometry
            $\hat W^\dagger (\hat E \otimes \hat\id_S) :
            \H_A \otimes \H_S \to \H_{\tilde A} \otimes \H_S$.
            If $A$ is regular, then $\hat\pi = \hat\id_{\tilde AS}$.

        \medskip
        \item 
            For any operator $\hat f_S$ on $S$ and its corresponding
            relativized operator $\hat F := \mathsf R_A[\hat f_S]$
            it holds that
            \begin{equation}
                \mathsf W^\dagger \circ (\mathsf E \otimes \sfid_S)[\hat F]
                = \hat\pi \bigl( \hat\id_{\tilde A} \otimes \hat f_S \bigr)
                \hat\pi.
                \label{eq:app_general_jump_relative}
            \end{equation}

        \medskip
        \item\label{item:commutation_almost_jump}
            For all $g\in G$,
            \begin{equation*}
                \bigl[ \hat L_{\tilde A}(g) \otimes \hat\id_S,
                \hat\pi \bigr] = 0
            \end{equation*}
            and
            \begin{equation}
                \mathsf W^\dagger \circ (\mathsf E \otimes \sfid_S)
                \bigl[ \hat U_A(g) \otimes \hat U_S(g) \bigr]
                = \hat\pi \bigl( \hat L_{\tilde A}(g) \otimes \hat\id_S \bigr)
                \hat\pi.
                \label{eq:app_general_jump_representation}
            \end{equation}
        \end{enumerate}
    \medskip
\end{proposition}
A version of this result, restricted to the zero charge sector, was first found in~\cite{Hamette2021}.
\begin{proof}
    1. We first note that $\hat E \hat E^\dagger$ is the orthogonal projector
    onto the image of $\hat E$ in $\H_{\tilde A}$:
    $(\hat E\hat E^\dagger)^\dagger = \hat E\hat E^\dagger$ and
    \begin{equation*}
        (\hat E \hat E^\dagger)^2
        = \hat E\, \hat\id_A \hat E^\dagger = \hat E\hat E^\dagger,
        \qquad \hat E\hat E^\dagger\hat E
        = \hat E\, \hat\id_A = \hat E.
    \end{equation*}
    Furthermore,
    \begin{multline*}
        \hat\pi := \hat W^\dagger(\hat E\hat E^\dagger \otimes \hat\id_S)\hat W
        = \frac{1}{c} \int_G\der g\,\der g'\, \sketbra{g}{g'}_{\tilde A}
        \otimes \hat U^\dagger_S(g) \hat U_S(g') \sbraket{g}{g'}_A
        = \frac{1}{c} \int_G\der g\,\der g'\, \sketbra{g'g}{g'}_{\tilde A}
        \otimes \hat U_S(g^{-1}) \sbraket{g}{e}_A \\
        = \frac{1}{c} \int_G\der g\, \hat R_{\tilde A}(g) \otimes \hat U_S(g)
        \sbraket{e}{g}_A,
    \end{multline*}
    where in the second step we have substituted $g \rightsquigarrow g'g$,
    and in the second-to-last step we have substituted
    $g \rightsquigarrow g^{-1}$.
    Since $\hat W^\dagger$ is unitary and
    $\hat E\hat E^\dagger \otimes \hat\id_S$
    is the orthogonal projector onto the image of $\hat E \otimes \hat\id_S$,
    $\hat\pi$ is the orthogonal projector onto the image of
    $\hat W^\dagger (\hat E \otimes \hat\id_S)$.
    This last map is an isometry because $\hat E \otimes \hat\id_S$
    is an isometry and $\hat W^\dagger$ is unitary.
    If $A$ is regular, then $\hat E\hat E^\dagger$ is the identity
    (because $\hat E$ is the identity up to label change),
    and hence $\hat\pi$ is the identity.

    2. From statement~\ref{item:projector} it follows that
    \begin{multline*}
        \hat\pi \bigl( \hat\id_{\tilde A} \otimes \hat f_S \bigr) \hat\pi
        = \hat W^\dagger\bigl(\hat E\hat E^\dagger \otimes \hat\id_S\bigr)
        \hat W \bigl(\hat\id_{\tilde A} \otimes \hat f_S\bigr)\hat W^\dagger
        \bigl(\hat E\hat E^\dagger \otimes \hat\id_S\bigr)\hat W \\
        = \hat W^\dagger\bigl(\hat E\hat E^\dagger \otimes \hat\id_S\bigr)
        \left(\int_G\der g\, \sketbra{g}{g}_{\tilde A}
        \otimes \mathsf U_S(g)[\hat f_S]\right)
        \bigl(\hat E\hat E^\dagger \otimes \hat\id_S\bigr)\hat W
        = \hat W^\dagger (\hat E \otimes \hat\id_S) \hat F
        (\hat E^\dagger \otimes \hat\id_S) \hat W.
    \end{multline*}

    3. Using \cref{eq:regular_jump_representation}
    and the properties of $\hat E$ we compute
    \begin{equation*}
        \mathsf W^\dagger \circ (\mathsf E \otimes \sfid_S)
        \bigl[ \hat U_A(g) \otimes \hat U_S(g) \bigr]
        = \mathsf W^\dagger[ \hat U_{\tilde A}(g) \hat E\hat E^\dagger
        \otimes \hat U_S(g)]
        = \mathsf W^\dagger\bigl[ \hat L_{\tilde A}(g) \otimes \hat U_S(g)\bigr] \cdot
        \mathsf W^\dagger\bigl[\hat E\hat E^\dagger \otimes \hat\id_S\bigr]
        = \bigl( \hat L_{\tilde A}(g) \otimes \hat\id_S \bigr) \hat\pi.
    \end{equation*}
    Now, $[\hat L_{\tilde A}(g) \otimes \hat\id_S,\hat\pi] = 0$
    because the left- and right-regular representations on $\tilde A$ commute.
    Hence,
    \begin{equation*}
        \mathsf W^\dagger \circ (\mathsf E \otimes \sfid_S)
        \bigl[ \hat U_A(g) \otimes \hat U_S(g) \bigr]
        = \bigl( \hat L_{\tilde A}(g) \otimes \hat\id_S \bigr) \hat\pi^2
        = \hat\pi \bigl( \hat L_{\tilde A}(g) \otimes \hat\id_S \bigr) \hat\pi.
    \end{equation*}
\end{proof}

To gain intuition about $\hat\pi$, let us sketch how $\hat\pi$ can be decomposed using the representation theory of $G$.
For simplicity and illustrative purposes, we will assume that $G$ is compact.
In appendix~\ref{app:POVs_Abelian}, we will carry out the details of this decomposition for all Abelian symmetry groups $G$,
including non-compact ones.
The main idea is as follows: 
the various representations of $G$ ($\hat R_{\tilde A}$, $\hat U_S$ and $\hat U_A$ within the overlap $\braket{e}{g}_A$)
occurring in \cref{eq:app_general_jump_projector} can be decomposed into irreducible representations.
Then, tensor products of those irreducible representations couple together to form direct sums
of irreducible representations (according to the Peter-Weyl theorem, proposition~\ref{prop:Peter_Weyl}).
Finally, the integral $\int_G\der g$ leads to a further simplification via
\emph{Schur orthogonality}~\cite{Fell1988a}.
Concretely, we first use the Peter-Weyl theorem to decompose the right-regular representation
$\hat R_{\tilde A}(g)$ into irreducible representations,
producing a direct sum over the total charge $r$:
\begin{equation*}
    \hat\pi \iso \bigoplus_r \hat\id_r^\Gamma \otimes \int_G \der g\, \hat U_{r^*}(g) \otimes \hat U_S(g)
    \braket{e}{g}_A,
\end{equation*}
where $r^*$ labels the irreducible representation complex-conjugate to $r$.\footnote{
    We adopt this new notation here, instead of $*,r$ (which was also used more generally for commutants),
    as it will prove useful in the final expression of $\hat\pi$.
}
The superscript-$\Gamma$ on the multiplicity spaces indicates that this is
where the left-regular representation acts irreducibly with $\hat U_r$.
Next, we decompose $\hat U_{r^*} \otimes \hat U_S$ into irreducible representations:
\begin{equation*}
    \hat\pi \iso \bigoplus_r \hat\id_r^\Gamma \otimes \!\bigoplus_{q \in I(r)} \int_G \der g\, \hat U_q(g) \otimes \hat\id_{*,q}
    \braket{e}{g}_A.
\end{equation*}
Here, $I(r)$ is the set of charges $q$ occurring in $\hat U_{r^*}
\otimes \hat U_S$ in a given $r$-sector,
and the irreducible representations $\hat U_q$ occur with multiplicity $\hat\id_{*,q}$.
Note that at this point, the original tensor product structure
$\H_{\tilde A} \otimes \H_S$ is generally lost,
although we will see below that this is not the case when $G$ is Abelian.
Now, since $\hat\pi$ is an orthogonal projector, it is clear that $\int_G\der g\, \hat U_q(g) \braket{e}{g}_A =: \hat\pi_q$
must also be an orthogonal projector, possibly zero.
Next, we use \emph{Schur orthogonality}~\cite{Fell1988a}:
if $\hat U_r$ and $\hat U_{r'}$ are two irreducible unitary representations of $G$
with dimensions $\dim(r)$ and $\dim(r')$ respectively,
and $\{\ket{i}_r\}_{i=1}^{\dim(r)}$, $\{\ket{i}_{r'}\}_{i=1}^{\dim(r')}$
are orthogonal bases for the representation spaces, then
\begin{equation*}
    \frac{1}{|G|} \int_G \der g\, \bra{i}_r \hat U_r(g) \ket{j}_r\,
    \bra{k}_{r'} \hat U_{r'}(g) \ket{l}_{r'}
    = \frac{1}{\dim(r)} \delta_{r,r'^*} \delta_{i,k} \delta_{j,l},
\end{equation*}
where $r'^*$ labels the irreducible representation complex-conjugate to $r'$.
The overlap $\braket{e}{g}_A$, being an expectation value of a representation,
can always be written as a sum of matrix elements of the irreducible representations contained in $A$,
and thus Schur orthogonality can be used to compute $\hat\pi_q$.
Let $J$ denote the representation charges contained in~$A$ and $J^*$ the charges corresponding
to the complex-conjugate counterparts, then Schur orthogonality implies that $\hat\pi_q = 0$ whenever $q \notin J^*$.
The exact form of the non-zero $\hat\pi_q$ depend on the matrix elements contained in the overlap,
and are \eg\ particularly easily computed if $A$ is a so-called \emph{L-R system}~\cite{Hamette2021}.
Hence, we generally have
\begin{equation}
    \hat\pi = \bigoplus_r \hat\id_r^\Gamma \otimes \!\bigoplus_{q \in I(r) \cap J^*} \hat \pi_q \otimes \hat\id_{*,q}.
    \label{eq:general_jump_projector_decomposition_sketch}
\end{equation}
In conclusion, we see that if $I(r) \cap J^* = I(r)$ for a given value of $r$,
then $\hat\pi$ reduces to the identity in this sector,
and the QRF jump behaves like in the regular case.
If $A$ is regular, then $J$ and $J^*$ contain all possible charges, and $\hat\pi = \hat\id$,
as expected.

\section{Details on Points of View for Abelian Groups}
\label{app:POVs_Abelian}

\noindent
In this appendix, we provide details for section~\ref{sec:POVs_Abelian}.
As in the main text, we assume that $A$ is a QRF
for an Abelian group $G$ with rank-one
orientation POVM elements $\hat\gamma_A(g) = \frac{1}{c} \sketbra{g}{g}_A$,
$g \in G$, $c > 0$, and unitary representation
$\hat U_A(g')\sket{g}_A = \sket{g'g}_A$, $g,g' \in G$.
Furthermore, let $S$ be a quantum system
with unitary representation $\hat U_S$ of $G$.
We now derive the perspective of $A$.

\medskip
\paragraph*{Strategy.}
Our main tools for the task are proposition~\ref{prop:general_jump}
and the representation theory of $G$.
Since $G$ is Abelian, the latter is particularly simple,
making the problem tractable.
We will come back to it further down.
Let us briefly detail how we will accomplish the task.

Firstly, we note that according to
\cref{eq:app_general_jump_relative}, applying the isometry
$\hat V^{\to A} := \hat W^\dagger (\hat E \otimes \hat\id_S)
= \hat\pi \hat W^\dagger (\hat E \otimes \hat\id_S)$
makes relative operators act almost semi-simply on the subsystem $S$,
up to the application of the projector $\hat\pi$.
Similarly, \cref{eq:app_general_jump_representation}
shows that $\hat V^{\to A}$ maps the global $G$-action
onto the subsystem $\tilde A$, up to $\hat\pi$.
Importantly, $S$ and $\tilde A$ are subsystems of the Hilbert space
$\H_{\tilde A} \otimes \H_S$
which is typically (whenever $\hat\pi \neq \hat\id$)
\emph{not} isomorphic to $\H^{|L} = \H_A \otimes \H_S$.
We write ``$\hat V^{\to A}$'', foreshadowing that this will nevertheless
turn out to be the sought-after QRF jump; but this remains to be shown.
Accordingly, we set $\H^{|A} := \hat V^{\to A} \H^{|L}$.
$\H^{|A}$ and $\H^{|L}$ are isomorphic because $\hat V^{\to A}$
as defined above is an isometry. 

Secondly, we wish to apply the projector $\hat\pi$,
if possible \emph{without destroying the subsystem structure
on $\H_{\tilde A} \otimes \H_S$}.
Using the representation theory of $G$,
we will see that $\hat\pi$ can be decomposed as
$\hat\pi = \sum_i \hat\pi_{i,\tilde A} \otimes \hat\pi_{i,S}$,
where $\hat\pi_{i,\tilde A}$ and $\hat\pi_{i,S}$
are orthogonal projectors onto subspaces of
$\H_{\tilde A}$ and $\H_S$, respectively,
such that the projectors $\hat\pi_{i,\tilde A} \otimes \hat\pi_{i,S}$
are mutually orthogonal for different index values $i$.
Hence, the subsystem structure $\H_{\tilde A} \otimes \H_S$
stays intact when applying $\hat\pi$ for each $i$,
but becomes superselected due to the sum over $i$.
Furthermore, the $\hat\pi_{i,\tilde A}$'s project on invariant
(in the sense of the global $G$-action)
subspaces, so that the representation structure is also conserved
(at most pruned).
Hence, we reach a decomposition like that of theorem~\ref{thm:QRF_POV}.
This step importantly only works because $G$ is Abelian
and, hence, its representation theory is fairly straightforward.

Thirdly, one may wonder if the commutant is also represented correctly,
since so far we have only technically found a decomposition adapted to
relative operators and the global $G$-action.
However, we will see that the decomposition will automatically also fit
the commutant, again essentially because $G$ is Abelian.

The decomposition found here is the decomposition of
\eqref{eq:general_jump_projector_decomposition_sketch} when $G$ is
Abelian. In contrast to general groups $G$, we will see that the
subsystem structure between $\tilde A$ and $S$ is preserved by $\hat
\pi$ as given by \eqref{eq:general_jump_projector_decomposition_sketch}.
The preservation of this structure is, as far as we know, specific to
the Abelian case.

\medskip
\paragraph*{Representation Theory for Abelian $\boldsymbol G$.}
Since $G$ is Abelian, it is natural to write the group multiplication
as addition: $g+g' := gg' = g'g$ and $e=0$.
Furthermore, all irreducible representations are one-dimensional~%
\cite{Weil1965} (this follows from Schur's lemma).
We consider only unitary representations and will not explicitly
mention that they are unitary every time.

Irreducible representations thus act on $\bbC$
and are generally of the form $\hat U(g) = \ex^{\i r(g)}$
for some function $r$ taking values in $\bbR/2\pi$
and such that $r(ag + g') = ar(g) + r(g') \mod 2 \pi$, $a \in \bbZ$.
Two such representations are equivalent (isomorphic)
if and only if they are equal.
Essentially, $r$ is the \emph{charge}
which uniquely labels the irreducible unitary representation.
For any charge $r$, $a r$ for $a \in \bbZ$, particularly $-r$,
are also charges labelling irreducible representations.
The zero charge $r=0$ is the trivial representation.
Consider two Hilbert spaces $\H_1,\H_2 = \bbC$
with irreducible representations given by functions $r_1$ and $r_2$.
Then, $\H_1 \otimes \H_2 = \bbC$ carries the representation
$\ex^{\i r_1(g)} \ex^{\i r_2(g)} = \ex^{\i(r_1+r_2)(g)}$ which is thus
characterized by the function $r_1 + r_2$.
It is necessarily irreducible because it is one-dimensional.
Note that two different functions $r_1$ and $r_2$ may lead to the same
irreducible representation, as is \eg\ the case in finite groups;
in that case we identify $r_1$ and $r_2$.
In summary, the charges $\Sigma$ labelling
the irreducible representations of $G$
form a module over $\bbZ$ (with addition and scalar multiplication defined
as the addition and scaling of functions modulo identifications),
and charge addition corresponds to $\Sigma$-addition.
This makes $\Sigma$ into a group, called the \emph{Pontryagin dual} of $G$~%
\cite{Weil1965}.
Contrast this to the non-Abelian groups where charge addition
is generally more involved, \eg\ spin-addition for $\SU(2)$.

The (non-compact generalization of) the Peter-Weyl theorem
(proposition~\ref{prop:Peter_Weyl}) applies:
\begin{equation}
    L^2(G) \iso \bigoplus_{r\in \Sigma} \H_r,
    \qquad \H_r = \bbC,
    \qquad \hat L(g) \iso \bigoplus_{r\in \Sigma} \ex^{\i r(g)}\, \hat\id_r,
    \qquad \hat R(g) \iso \bigoplus_{r\in \Sigma} \ex^{-\i r(g)}\, \hat\id_r,
    \label{eq:Peter_Weyl_Abelian}
\end{equation}
where $\bigoplus$ indicates as usual the appropriate combination of direct sum
and integral needed for the given group $G$.
Note that the multiplicities are all trivial, and we could leave them out.
Because $\Sigma$ is a group, we can choose the measure in
$\bigoplus_{r\in\Sigma}$ as a Haar measure of $\Sigma$.
If this is done, then one can generalize Fourier theory and particularly
important for us, obtain a generalization of \emph{Schur orthogonality}~%
\cite{Weil1965}:
\begin{equation}
    \frac{1}{C} \int_G\der g\, \ex^{\i r(g) - \i r'(g)} = \delta(r-r')
    \label{eq:Schur_orthogonality}
\end{equation}
for some constant $C > 0$ dependent on the choice of Haar measure on $\Sigma$,
and where the ``$\delta$-function'' is to be interpreted as
$\int_\Sigma\der r\, \delta(r-r')\, f(r) = f(r')$ for $\der r$
the Haar measure on $\Sigma$.
In particular, $\bigoplus_{r\in\Sigma} \delta(r-r')\, \H_r = \H_{r'}$.

If $G$ is compact, then $\Sigma$ is discrete, and $\bigoplus_{r\in\Sigma}$
truly a direct sum.
In that case, it makes sense to take the counting measure on $\Sigma$,
so that $C = |G|$ and $\delta(r-r') = \delta_{r,r'}$ is the Kronecker-$\delta$.
If $G = (\bbR^n,+)$, then \cref{eq:Schur_orthogonality}
is the well-known Fourier-formula for the Dirac-$\delta$,
and hence $C = (2\pi)^n$.
We will however not need the actual value of $C$.

\medskip
\paragraph*{Decomposition of $\boldsymbol{\hat\pi}$.}
We will consider explicitly the compact case,
and indicate at the end of this step how our findings
can be generalized to non-compact symmetry groups.

The Hilbert space $\H_{\tilde A}$ together with the left- and right-regular
representations $\hat L_{\tilde A}$ and $\hat R_{\tilde A}$
acting on it decompose according to \cref{eq:Peter_Weyl_Abelian}.
Furthermore, $\H_S$ carrying the representation $\hat U_S$ can be decomposed as
\begin{equation}
    \H_S \iso \bigoplus_{q\in\sigma_S} \H_q,
    \qquad \hat U_S(g) \iso \bigoplus_{q\in\sigma_S}
    \ex^{\i q(g)}\, \hat\id_q.
\end{equation}
Here, $\sigma_S \subset \Sigma$ is the set of charges occurring in $\hat U_S$,
and $\H_q$ is the multiplicity space of $q \in \sigma_S$;
because the irreducible representation spaces
are one-dimensional, we could leave them out,
and for notational simplicity, we labelled the multiplicity spaces directly
with ``$q$'' rather than ``$*,q$'' like in proposition~\ref{prop:Peter_Weyl}.
Note that the multiplicity spaces may have any dimension
because $\hat U_S$ is arbitrary.

Consequently, the projector $\hat\pi$, as defined in \eqref{eq:app_general_jump_projector}, decomposes as follows:
\begin{equation*}
    \hat\pi \iso \frac{1}{c}
    \sum_{\substack{r\in\Sigma\\q\in\sigma_S}}
    \hat\pi_{r,\tilde A} \otimes \hat\pi_{q,S}
    \int_G\der g\,
    \ex^{-\i r(g) + \i q(g)}
    \braket{0}{g}_A.
\end{equation*}
To simplify $\hat\pi$ further, we must also decompose $\H_A$ into
irreducible representations using the Peter-Weyl theorem:
\begin{equation*}
    \H_A \iso \bigoplus_{p\in\sigma_A} \H_p,
    \qquad \hat U_A(g) \iso \bigoplus_{p\in\sigma_A}
    \ex^{\i p(g)}\, \hat\id_p,
\end{equation*}
where $\sigma_A\subset\Sigma$ is the set of charges contained in $A$
and $\H_p$ is the multiplicity subspace for $p\in\sigma_A$.
Most generally, the state $\ket{0}_A$ is of the form
$\ket{0}_A \iso \sum_{p\in\sigma_A} \alpha_p \sket{\psi_p}$,
with $\sket{\psi_p} \in \H_p$ normalized so that
$\sum_{p\in\sigma_A} \sketbra{\psi_p}{\psi_p} = \hat\id_A$,
and~$\alpha_p \in \bbC$.
Accordingly, $\ket{g}_A \iso \sum_{p\in\sigma_A} \alpha_p\,\ex^{\i p(g)}\,
\sket{\psi_p}$ and completeness $\frac{1}{c}\int_G\der g\,\ketbra{g}{g}_A
= \hat\id_A$ of the orientation POVM requires
\begin{equation*}
    \sum_{p\in\sigma_A} \hat\pi_p = \hat\id_A
    = \frac{1}{c} \sum_{p,\tilde p\in\sigma_A}\int_G\der g\,
    \alpha_p \overline{\alpha_{\tilde p}}\,
    \sketbra{\psi_p}{\psi_{\tilde p}} \ex^{\i p(g) - \i \tilde p(g)}
    = \frac{C}{c} \sum_{p\in\sigma_A} |\alpha_p|^2 \ketbra{\psi_p}{\psi_p},
\end{equation*}
where we have used Schur orthogonality, \cref{eq:Schur_orthogonality}.
As we are in the compact case $C = |G|$;
however, since we later wish to generalize, we keep $C$ unspecified
and use the Dirac-$\delta$ notation rather than Kronecker-$\delta$.
Thus, we must have that $\sketbra{\psi_p}{\psi_p} = \hat\pi_p$,
forcing $\dim\H_p = 1$, and $|\alpha_p|^2 = c/C$.
Hence, $c^{-1}\sbraket{0}{g}_A = C^{-1}\sum_{p\in\sigma_A} \ex^{\i p(g)}$,
and
\begin{equation*}
    \hat\pi \iso \,\sum_{\mathclap{\substack{r\in\Sigma\\
    q\in\sigma_S,\,p\in\sigma_A}}}\,
    \hat\pi_{r,\tilde A} \otimes \hat\pi_{q,S}\,
    \frac{1}{C} \int_G\der g\,
    \ex^{-\i r(g) + \i q(p) + \i p(g)}
    = \,\sum_{\mathclap{\substack{r\in\Sigma\\
    q\in\sigma_S,\,p\in\sigma_A}}}\,
    \hat\pi_{r,\tilde A} \otimes \hat\pi_{q,S}\,
    \delta(p + q - r),
\end{equation*}
where we have again used Schur orthogonality.

Now is a good time discuss the changes needed for the non-compact case.
Chiefly, all (direct) sums over charges must be replaced by the appropriate
combinations of (direct) integrals and sums, as dictated by the non-compact
generalization of the Peter-Weyl theorem.
Furthermore, projectors $\hat\pi_{r,\tilde A}$ and $\hat\pi_{q,S}$
onto charges typically become improper,
as do the states $\sket{\psi_p}$.
However, all completeness relations remain valid if the sums were
correctly replaced as just described.
All $\delta$-functions are now generally products of Kronecker-
and Dirac-$\delta$ distributions on $\Sigma$.
Finally, the constant $C$ cancels just like in the compact case,
and the decomposition of $\hat\pi$ has the same algebraic form.
One easily checks that this translation from compact to non-compact
works if \eg\ $G = (\bbR^n,+)$, using Fourier integrals.

Continuing in parallel for both cases using the notation for the compact case,
we notice that the $p+q-r = 0$ is only possible if two conditions are met:
firstly, $r \in \sigma_A + \sigma_S$, but according to charge addition,
$\sigma_A + \sigma_S = \sigma_{AS}$ is simply the charge spectrum of
the total system $AS$;
secondly, $q \in \sigma_S \cap (\sigma_{AS} - \sigma_A)$
because $q \in \sigma_S$, $r \in \sigma_{AS}$ and $p \in \sigma_A$.
The conditions are shown in figure~\ref{fig:charge_diagram} in the main text.
Using them, we can simplify the decomposition of $\hat\pi$:
\begin{equation*}
    \hat\pi \iso \,\sum_{\mathclap{\substack{r\in\sigma_{AS}\\
    q\in\sigma_S \cap (r - \sigma_A)}}}\,
    \hat\pi_{r,\tilde A} \otimes \hat\pi_{q,S}.
\end{equation*}
So indeed, $\hat\pi$ is of the form described above in the strategy outline.

To simplify the notation, we denote by
$\bb\kappa(r) := \sigma_S \cap (r-\sigma_A)$
the set of charges occurring in the above sum over $q$,
and introduce the spaces and corresponding orthogonal projectors
\begin{equation*}
    \H_{\bb k,S} := 
    \,\bigoplus_{q \in \bb k}\, \H_{q,S},
    \qquad \hat\pi_{\bb k,S} :=
    \sum_{q \in \bb k} \hat\pi_{q,S}
\end{equation*}
obtained from any given set $\bb k$
(boldface distinguishes set from single charges).
This allows us to write
\begin{equation}
    \hat\pi \iso \sum_{\substack{r\in\sigma_{AS}}}
    \hat\pi_{r,\tilde A} \otimes \hat\pi_{\bb \kappa(r),S}
    = \sum_{\substack{\bb k \in \bb\kappa(\sigma_{AS})\\
    r \in \bb\kappa^{-1}(\bb k)}}
    \hat\pi_{r,\tilde A} \otimes \hat\pi_{\bb k,S}
    = \sum_{\bb k \in \bb\kappa(\sigma_{AS})} \hat\pi_{\bb\kappa^{-1}(\bb k),\tilde A}
    \otimes \hat\pi_{\bb k,S},
    \label{eq:pi_decomposition_Abelian}
\end{equation}
where we see $\bb\kappa$ as a function
with image $\bb\kappa(\sigma_{AS})$, and $\bb\kappa^{-1}(\bb k)$
is the pre-image of the set $\bb k$ under $\bb\kappa$.
Notice that different charges $r_1 \neq r_2$ may result in the same set
$\bb\kappa(r_1) = \bb\kappa(r_2)$,
but different sets $\bb k_1 \neq \bb k_2$
are always due to different charges $r_1 \neq r_2$
because $\bb\kappa$ is a function.
As we will see, the perspective of $A$ is dictated by this function
$\bb\kappa$; if $G$ is not Abelian,
there is no reason to believe that such a function exists.

\medskip
\paragraph*{Perspective of $\boldsymbol A$.}
With proposition~\ref{prop:general_jump} and
\cref{eq:pi_decomposition_Abelian} it follows that
via the QRF jump $\hat V^{\to A} = \hat W^\dagger (\hat E \otimes \hat\id_S)
= \hat\pi\hat W^\dagger (\hat E \otimes \hat\id_S)$,
we get decompositions for the Hilbert space $\H^{|A}$,
for the global $G$-action $\hat U^{|A}$,
and for relative operators of the form
$\hat F = \frac{1}{c} \int_G\der g\,\sketbra{g}{g}_A \otimes
\mathsf U_S(g)[\hat f_S]$, \eg\ relative POVM elements:
\begin{align}
    \H^{|A} := \hat V^{\to A}\H^{|L} &=
    \bigoplus_{\substack{\bb k \in \bb\kappa(\sigma_{AS})\\
    r \in \bb\kappa^{-1}(\bb k)}}
    \H_{r,\tilde A} \otimes \H_{\bb k,S},
    \label{eq:Hilbert_space_A_Abelian}\\
    \qquad \hat U^{|A}(g) := \mathsf V^{\to A}[\hat U^{|L}(g)] &= 
    \bigoplus_{\substack{\bb k \in \bb\kappa(\sigma_{AS})\\
    r \in \bb\kappa^{-1}(\bb k)}}
    \ex^{\i r(g)}\, \hat\id_{r,\tilde A} \otimes \hat\id_{\bb k,S},
    \label{eq:global_reorientations_A_Abelian}\\
    \qquad \mathsf V^{\to A}[\hat F] &=
    \bigoplus_{\bb k \in \bb\kappa(\sigma_{AS})}
    \hat\id_{\bb\kappa^{-1}(\bb k),\tilde A} \otimes \hat\pi_{\bb k,S} \hat f_S \hat\pi_{\bb k,S}.
    \label{eq:relative_POVM_A_Abelian}
\end{align}
Now, the elements of $\A^{S:A|A}$ are most generally polynomials
of operators with the form given in \cref{eq:relative_POVM_A_Abelian},
but finding the algebra directly by computing polynomials is non-trivial.
The $G$-invariant commutant $\A^{\overline{S:A}|A}$ on the other hand
is easy to find, consisting of the $G$-invariant operators
which commute with all operators of the form
of \cref{eq:relative_POVM_A_Abelian},
as the latter generate $\A^{S:A|A}$.
We find
\begin{equation}
    \A^{\overline{S:A}|A} =
    \mathsf V^{\to A}\bigl[\A^{\overline{S:A}|L}\bigr] =
    \bigoplus_{\substack{\bb k\in\bb\kappa(\sigma_{AS})\\
    r\in\bb\kappa^{-1}(\bb k)}}
    \O_{r,\tilde A} \otimes \hat\id_{\bb k,S},
    \label{eq:commutant_A_Abelian}
\end{equation}
where we have used that $\hat f_S$ is arbitrary and that
$G$-invariant operators are block-diagonal in $r$
(apply Schur's lemma to the fact that $G$ acts according to
\cref{eq:global_reorientations_A_Abelian}).

It remains to find the operator algebra $\A^{S:A|A}$
generated by operators of the form given in \cref{eq:relative_POVM_A_Abelian}.
Notice that \cref{eq:relative_POVM_A_Abelian} is block-diagonal in $\bb k$,
and hence the same is true for $\A^{S:A|A}$. This leads to
\begin{equation}
    \A^{S:A|A} = \mathsf V^{\to A}[\A^{S:A|L}] <
    \bigoplus_{\bb k\in\bb\kappa(\sigma_{AS})}
    \hat\id_{\bb\kappa^{-1}(\bb k),\tilde A} \otimes \O_{\bb k,S}.
    \label{eq:larger_algebra_A_Abelian}
\end{equation}
The goal will be to show that this is actually an equality.
For this, we take the intersection of \cref{eq:commutant_A_Abelian}
and of the right-hand side of \cref{eq:larger_algebra_A_Abelian}
to find that
\begin{equation*}
    \Z^{S:A|A} = \A^{S:A|A} \cap \A^{\overline{S:A}|A}
    < \O\Bigl[ \hat\pi_{\bb\kappa^{-1}(\bb k),\tilde A} \otimes
    \hat\pi_{\bb k,S} \st \bb k \in \bb\kappa(\sigma_{AS}) \Bigr].
\end{equation*}
Because any element of an operator algebra can be written
as some polynomial of generators
it follows that $\Z^{S:A|A}$ consists of certain polynomials of the
pairwise orthogonal projectors $\hat\pi_{\bb\kappa^{-1}(\bb k),\tilde A}
\otimes \hat\pi_{\bb k,S}$, $\bb k \in \bb\kappa(\sigma_{AS})$.
But polynomials of pairwise orthogonal projectors are linear combinations
of them, and such a linear combination
$\smash{\sum_i \lambda_i \hat\pi_{\bb\kappa^{-1}(\bb k_i),\tilde A}
\otimes \hat\pi_{\bb k_i}}$ is itself an orthogonal projector if and only if
$\lambda_i \in \{0,1\}$ for all $i$.

Thus, if $\hat\Pi_k \in \Z^{S:A|A}$ is the orthogonal projector
onto a simultaneous eigenspace of $\Z^{S:A|A}$,
\ie\ onto a $k$-sector in the notation of proposition
\ref{prop:algebra_decomposition_full} or equivalently
theorem~\ref{thm:QRF_POV}, then it must have the form
\begin{equation}
    \hat\Pi_k = \bigoplus_i \hat\pi_{\bb\kappa^{-1}(\bb k_i),\tilde A}
    \otimes \hat\pi_{\bb k_i,S},
    \label{eq:centre_projector_A_Abelian}
\end{equation}
where $\bb k_i \in \bb\kappa(\sigma_{AS})$ are distinct sets of charges.
Note the difference between the subscript-$\bb k$ (bold font),
which labels a set of charges,
and subscript-$k$ (regular font),
which labels a simultaneous eigenspace of the centre $\Z^{S:A|A}$.
Given orthogonal projectors $\hat\Pi_k$, $\hat\Pi_{k'}$
onto \emph{different} simultaneous eigenspaces of $\Z^{S:A|A}$,
the sets $\{\bb k_i\}_i$ and $\{\bb k'_j\}_j$ of charge sets
characterizing the projectors
in the sense of \cref{eq:centre_projector_A_Abelian}
must be \emph{disjoint},
as otherwise the projectors are not orthogonal to each other.
Because $\bb\kappa$ is a function,
the representation charges contained in either eigenspace are
also disjoint: any given representation charge value $r$
may only occur in a single simultaneous eigenspace of $\Z^{S:A|A}$,
\ie\ for a single $k$-value.
This observation will be useful later.

\Needspace{2cm}
To continue, we need the following result, which we assume for now and will prove later:
\begin{lemma}{}{app_A_Abelian_technical}
    Let $K$ label the set of simultaneous eigenspaces
    of $\Z^{S:A|A}$ (the set over which the $k$-sums in
    proposition~\ref{prop:algebra_decomposition_full} 
    and theorem~\ref{thm:QRF_POV} run).
    There exists a bijection
    $\jmath : K \to \bb\kappa(\sigma_{AS})$,
    and the orthogonal projector in \cref{eq:centre_projector_A_Abelian}
    onto the simultaneous eigenspace of $\Z^{S:A|A}$ labelled by $k \in K$ is
    \begin{equation*}
        \hat\Pi_k = \hat\pi_{(\bb\kappa^{-1} \circ \jmath)(k),\tilde A}
        \otimes \hat\pi_{\jmath(k),S}.
    \end{equation*}
\end{lemma}
\noindent
Thus, $\bb k$ labels precisely the simultaneous eigenspaces of
the centre $\Z^{S:A|A}$ (hence the choice of symbol).
Let now $\hat\pi_{\bb k} :=
\hat\pi_{\bb\kappa^{-1}(\bb k),\tilde A} \otimes \hat\pi_{\bb k,S}$
be the orthogonal projector onto a $\bb k$-sector.
According to lemma~\ref{lem:app_A_Abelian_technical}, $\hat\pi_{\bb k} = \hat\Pi_{\jmath^{-1}(\bb k)}$
is an orthogonal projector onto a simultaneous eigenspace of $\Z^{S:A|A}$.
Hence, $\hat\pi_{\bb k} \in \Z^{S:A|A} \subset \A^{S:A|A}$.
Therefore, for any relational operator $\mathsf V^{\to A}[\hat F]
\in \A^{S:A|A}$ of the form of \cref{eq:relative_POVM_A_Abelian}, also
\begin{equation*}
    \hat\pi_{\bb k} \mathsf V^{\to A}[\hat F] \hat\pi_{\bb k}
    = \hat\id_{\bb\kappa^{-1}(\bb k),\tilde A}
    \otimes \hat\pi_{\bb k} \hat f_S \hat\pi_{\bb k}
    \in \A^{S:A|A}.
\end{equation*}
But since $\hat f_S$ is arbitrary, we have equality in
\cref{eq:larger_algebra_A_Abelian}:
\begin{equation}
    \A^{S:A|A} = \mathsf V^{\to A}[\A^{S:A|L}]
    = \bigoplus_{\bb k\in\bb\kappa(\sigma_{AS})}
    \hat\id_{\bb\kappa^{-1}(\bb k),\tilde A} \otimes \O_{\bb k,S},
    \label{eq:relative_operators_A_Abelian}
\end{equation}
showing what we wanted to show.

\medskip
\begin{proof}[Proof of lemma~\ref{lem:app_A_Abelian_technical}]
    Consider the orthogonal projector $\hat\Pi_k$ onto a simultaneous
    eigenspace of $\Z^{S:A|A}$,
    as defined in \cref{eq:centre_projector_A_Abelian}.
    According to proposition~\ref{prop:algebra_decomposition_full},
    this eigenspace and the related mathematical objects decompose as
    \begin{gather}
        \hat\Pi_k \H^{|A} = \bigoplus_{r\in\Omega} \H_{r,k} \otimes \H_k,
        \qquad \hat\Pi_k \A^{S:A|A} \hat\Pi_k
        = \Bigl(\bigoplus_{r\in\Omega} \hat\id_{r,k}\Bigr) \otimes \O_k, \notag\\
        \hat\Pi_k \A^{\overline{S:A}|A} \hat\Pi_k
        = \bigoplus_{r\in\Omega} \O_{r,k} \otimes \hat\id_k,
        \qquad \hat\Pi_k \hat U^{|A}(g) \hat\Pi_k
        = \bigoplus_{r\in\Omega} \ex^{\i r(g)}\hat\id_{r,k} \otimes \hat\id_k,
        \label{eq:decompositions_k_fixed_A_Abelian}
    \end{gather}
    for some Hilbert spaces $\H_k$, $\H_{r,k}$, and where
    $\Omega$ is the set of representation charges contained in the eigenspace.
    We have also used that all irreducible representations are one-dimensional
    and act as $\bbC$-valued exponentials.
    According to \cref{eq:centre_projector_A_Abelian},
    $\Omega = \bigcup_i \bb\kappa^{-1}(\bb k_i)$.

    The simultaneous eigenspace corresponding to $\hat\Pi_k$,
    which is characterized by the set of distinct charge sets $\{\bb k_i\}_i$,
    is also characterized by the set of pre-image charge sets
    $\{\bb\kappa^{-1}(\bb k_i)\}_i$,
    simply because $\bb\kappa$ is a function.
    We can thus restrict our previous decompositions in
    \cref{eq:Hilbert_space_A_Abelian,eq:global_reorientations_A_Abelian,%
    eq:commutant_A_Abelian}
    to those values of $r$ in the pre-images of the sets $\{\bb k_i\}_i$
    and thereby obtain the corresponding decompositions in
    \cref{eq:decompositions_k_fixed_A_Abelian}.
    Particularly, we must have that
    \begin{equation*}
        \Bigl(\bigoplus_{\mathclap{\substack{i\\r\in\bb\kappa^{-1}(\bb k_i)}}}
        \H_{r,k} \Bigr) \otimes \H_k
        = \bigoplus_{\mathclap{\substack{i\\r\in\bb\kappa^{-1}(\bb k_i)}}}
        \H_{r,\tilde A} \otimes \H_{\bb k_i},
        \qquad
        \Bigl(\bigoplus_{\mathclap{\substack{i\\r\in\bb\kappa^{-1}(\bb k_i)}}}
        \O_{r,k} \Bigr) \otimes \hat\id_{k}
        = \bigoplus_{\mathclap{\substack{i\\r\in\bb\kappa^{-1}(\bb k_i)}}}
        \O_{r,\tilde A} \otimes \hat\id_{\bb k_i,S}.
    \end{equation*}
    We have added brackets to emphasize that the second factors on the
    left-hand sides do not depend on $r$;
    this will be a crucial point.

    From \cref{eq:commutant_A_Abelian}, we know that
    $\O_{r,\tilde A}$ is the restriction of the $G$-invariant commutant
    to the charge $r$:
    $\pi_{r,\tilde A} \otimes \mathsf\id_{\bb\kappa(r)}
    \bigl[\A^{\overline{S:A}|A}\bigr]
    = \O_{r,\tilde A} \otimes \hat\id_{\bb\kappa(r),S} \iso \O_{r,\tilde A}$.
    At the same time, proposition~\ref{prop:algebra_decomposition_full}
    tells us that the restriction of the $G$-invariant commutant
    to the charge $r$ \emph{and} to the simultaneous centre eigenspace
    labelled by $k$ is $\O_{r,k}$.
    But we remarked above that any given charge value cannot occur in two
    different simultaneous eigenspaces of $\Z^{S:A|A}$.
    Hence, $\O_{r,k}$ is actually also the restriction of the $G$-invariant
    commutant to \emph{just} the charge $r$.
    Hence, we may identify (up to a basis change)
    $\O_{r,k} = \O_{r,\tilde A}$ and also $\H_{r,k} = \H_{r,\tilde A}$.
    But then this means that $\H_k = \H_{\bb k_i,S}$ for all $i$.
    And since the $\bb k_i$'s are distinct sets, this is possible only
    if the sum over $i$ runs over a single value.

    With this, we have shown that every $k$-sector is a $\bb k$-sector,
    and hence there is a function $\jmath : K \to \bb\kappa(\sigma_{AS})$
    implementing this mapping.
    This gives the form of $\hat\Pi_k$ in the lemma statement.
    For different simultaneous eigenspaces of $\Z^{S:A|A}$
    to be different, $\jmath$ must furthermore be injective.
    Finally, $\jmath$ must also be surjective,
    because \cref{eq:relative_POVM_A_Abelian} generally has support on
    all $\bb k$-sectors, and consequently the same must be true for
    $\Z^{S:A|A}$.
\end{proof}

\medskip
\paragraph*{Summary.}
We have found a decomposition as in theorem~\ref{thm:QRF_POV}:
\begin{proposition}{QRF Perspective for Abelian Groups}{POV_Abelian}
    Let $G$ be an Abelian group, $A$ a QRF with rank-one orientation POVM elements
    $\hat\gamma_A^{|L}(g) = \frac{1}{c} \sketbra{g}{g}_A$
    and representation $\hat U_A(g') \sket{g}_A = \sket{g'g}_A$,
    and let $S$ be a quantum system with representation $\hat U_S$.
    Consider the map $\hat V^{\to A} : \H_A \otimes \H_S = \H^{|L}
    \to \H_{\tilde A} \otimes \H_S$, defined as
    $\hat V^{\to A} := \hat W^\dagger(\hat E \otimes \id_S)$,
    where $\hat E$ is the embedding of $A$ into a regular QRF $\tilde A$
    given in proposition~\ref{prop:embedding},
    and $\hat W$ is defined in \cref{eq:app_regular_inverse_jump_core}.
    
    \medskip
    Then, $\hat V^{\to A}$, seen as a unitary map onto its image,
    is the QRF jump into the perspective of $A$
    in accordance with theorem~\ref{thm:QRF_POV}:
    \begin{align*}
        \allowdisplaybreaks
        \H^{|A} := \hat V^{\to A} \H^{|L} &=
        \bigoplus_{\substack{\bb k\in\bb\kappa(\sigma_{AS})\\
        r\in\bb\kappa^{-1}(\bb k)}}
        \H_{r,\tilde A} \otimes \H_{\bb k,S},
        \qquad \bb\kappa(r) := \sigma_S\cap(r-\sigma_A),
        \qquad \H_{\bb k,S} := \bigoplus_{q\in\bb k} \H_{q,S},\\
        \A^{S:A|A} := \mathsf V^{\to A}[\A^{S:A|L}] &=
        \bigoplus_{\bb k\in\bb\kappa(\sigma_{AS})}
        \hat\id_{\bb\kappa^{-1}(\bb k),\tilde A} \otimes \O_{\bb k,S},\\
        \hat F^{|A} := \mathsf V^{\to A}[\hat F^{|L}] &=
        \bigoplus_{\bb k\in\bb\kappa(\sigma_{AS})}
        \hat\id_{\bb\kappa^{-1}(\bb k),\tilde A} \otimes
        \hat\pi_{\bb k,S} \hat f_S \hat\pi_{\bb k,S},
        \qquad \hat F^{|L} := \frac{1}{c}\int_G\der g\, \sketbra{g}{g}_A
        \otimes \mathsf U_S(g)[\hat f_S],\\
        \A^{\overline{S:A}|A} := \mathsf V^{\to A}
        [\A^{\overline{S:A}|A}] &=
        \bigoplus_{\substack{\bb k\in\bb\kappa(\sigma_{AS})\\
        r\in\bb\kappa^{-1}(\bb k)}}
        \O_{r,\tilde A} \otimes \hat\id_{\bb k,S},\\
        \hat U^{|A}(g) := \mathsf V^{\to A}[\hat U^{|L}(g)] &=
        \bigoplus_{\substack{\bb k\in\bb\kappa(\sigma_{AS})\\
        r\in\bb\kappa^{-1}(\bb k)}}
        \ex^{\i r(g)}\,\hat\id_{r,\tilde A} \otimes \hat\id_{\bb k,S}.
    \end{align*}
    In particular, $r$ takes on the same role as in the theorem,
    and $\bb k$ here takes on the role of $k$ in the theorem.
    Note that all $\H_{r,\tilde A}$ are one-dimensional.
\end{proposition}

\medskip
\paragraph*{Subsystems from the Perspective of $\boldsymbol A$.}
Consider now the case where $S = S_1S_2$ consists of two subsystems
in the perspective of $L$.
We know from proposition~\ref{prop:regular_subsystem_structure}
that if $A$ is regular $S{:}A|A$ will also 
factor into two subsystems~$S_1{:}A|A$ and $S_2{:}A|A$.
However, this does not hold in general for non-ideal QRFs.
In the case of Abelian $G$ with rank-one orientation POVM elements for $A$,
which we are currently discussing, more can be said:
\begin{proposition}{}{algebra_decomposition_S_Abelian}
    Let $G$ be an Abelian group, $A$ a QRF for the group $G$
    with rank-one orientation POVM elements, and $S=S_1S_2$.
    Then, in the notation of proposition~\ref{prop:POV_Abelian},
    $S{:}A|A$ factorizes for a given charge set~$\bb k$, \ie
    \begin{equation*}
        \H_{\bb k,S} = \H_{\bb x_1,S_1} \otimes \H_{\bb x_2,S_2},
        \qquad \O_{\bb k,S} = \O_{\bb x_1,S_1} \otimes \O_{\bb x_2,S_2}
    \end{equation*}
    for some sets $\bb x_1$ and $\bb x_2$ of charges,
    if and only if the set of pairs $(q_1,q_2)
    \in \sigma_{S_1} \times \sigma_{S_2}$ satisfying $q_1 + q_2 \in \bb k$
    is~$\bb x_1 \times \bb x_2$.
    We denote the set of values of $\bb k$ for which this
    occurs by $\Xi \subset \bb\kappa(\sigma_{AS})$.
    It particularly occurs if~$\bb k = \sigma_{S_1S_2}$,
    and then $\bb x_1 = \sigma_{S_1}$, $\bb x_2 = \sigma_{S_2}$.

    \medskip
    Let furthermore $\hat X_1, \hat X_2 \in \A^{S:A|A}$ be of the form
    \begin{align*}
        \hat X_1 &= \bigoplus_{\substack{\bb k\in\bb\kappa(\sigma_{AS})\\
        r\in\bb\kappa^{-1}(\bb k)}}
        \hat\id_{r,\tilde A} \otimes \hat\pi_{\bb k,S}
        \bigl( \hat X^{\bb k}_{S_1} \otimes \hat\id_{S_2} \bigr)
        \hat\pi_{\bb k,S}, \\
        \hat X_2 &= \bigoplus_{\substack{\bb k\in\bb\kappa(\sigma_{AS})\\
        r\in\bb\kappa^{-1}(\bb k)}}
        \hat\id_{r,\tilde A} \otimes \hat\pi_{\bb k,S}
        \bigl( \hat\id_{S_1} \otimes \hat X^{\bb k}_{S_2} \bigr)
        \hat\pi_{\bb k,S},
    \end{align*}
    with $\hat X^{\bb k}_{S_1} \in \O_{S_1}$,
    $\hat X^{\bb k}_{S_2} \in \O_{S_2}$ arbitrary for all $\bb k$.
    Then,
    \begin{equation*}
        \bigl[ \hat X_1, \hat X_2 \bigr]
        = \bigoplus_{\substack{\bb k \in \bb\kappa(\sigma_{AS}) \backslash\Xi\\
        r\in\bb\kappa^{-1}(\bb k)}} \hat\id_{r,\tilde A} \otimes
        \bigoplus_{\substack{q_1+q_2\in\bb k\\q_1'+q_2'\in\bb k}}
        \Bigl( \chi_{\bb k}(q_1'+q_2) - \chi_{\bb k}(q_1+q_2') \Bigr)\,
        \hat\pi_{q_1,S_1} \hat X^{\bb k}_{S_1} \hat\pi_{q_1',S_1}
        \otimes \hat\pi_{q_2,S_2} \hat X^{\bb k}_{S_2} \hat\pi_{q_2',S_2},
    \end{equation*}
    where $\chi_{\bb k}$ is the characteristic function
    \begin{equation*}
        \chi_{\bb k}(q) := \begin{cases}
            1, & \quad \text{if $q \in \bb k$}, \\
            0, & \quad \text{otherwise}.
        \end{cases}
    \end{equation*}
    Particularly, $\hat X_1$ and $\hat X_2$ commute if they are
    restricted to $\bb k \in \Xi$.
\end{proposition}
\medskip
\begin{proof}
    Let $\bb k \in \bb\kappa(\sigma_{AS}) \subset \sigma_{S_1S_2}$
    be fixed; note that $\sigma_S = \sigma_{S_1S_2}
    = \sigma_{S_1} + \sigma_{S_2}$.
    We have that
    \begin{equation*}
        \H_{\bb k,S} = \hat\pi_{\bb k,S} \H_S
        = \hat\pi_{\bb k,S} \bigl( \hat H_{S_1} \otimes \H_{S_2} \bigr)
        = \bigoplus_{q_1 + q_2 \in \bb k}
        \bigl( \hat\pi_{q_1,S_1} \H_{S_1} \bigr) \otimes
        \bigl( \hat\pi_{q_2,S_2} \H_{S_2} \bigr),
    \end{equation*}
    where the direct sum ranges over all $q_1 \in \sigma_{S_1}$,
    $q_2 \in \sigma_{S_2}$ such that $q_1 + q_2 \in \bb k$.
    This factorizes if and only if the sum can be split into
    two direct sums over $q_1$ and $q_2$ respectively.
    But this is precisely the condition under which the set of solutions
    $(q_1,q_2)$ is a Cartesian product $\bb x_1 \times \bb x_2$
    for some sets of charges $\bb x_1 \subset \sigma_{S_1}$
    and $\bb x_2 \subset \sigma_{S_2}$.
    And since the sums then run over $\bb x_1$ and $\bb x_2$ respectively,
    this implies the form of $\H_{\bb k,S}$ and 
    the form of $\O_{\bb k,S}$ in the proposition also follows.
    Clearly, this is the case if $\bb k = \sigma_{S_1S_2}$,
    since then $q_1$ runs over $\sigma_{S_1}$
    independently of $q_2$ which runs over $\sigma_{S_2}$.

    Let $\hat X_1$ and $\hat X_2$ be as described.
    Then,
    \begin{equation*}
        \bigl[ \hat X_1, \hat X_2 \bigr]
        = \bigoplus_{\substack{\bb k \in \bb\kappa(\sigma_{AS})\\
        r \in \bb\kappa^{-1}(\bb k)}} \hat\id_{r,\tilde A} \otimes
        \bigl[
            \hat\pi_{\bb k,S} (\hat X^{\bb k}_{S_1} \otimes \hat\id_{S_2})
            \hat\pi_{\bb k,S},
            \hat\pi_{\bb k,S} (\hat\id_{S_1} \otimes \hat X^{\bb k}_{S_2})
            \hat\pi_{\bb k,S}
        \bigr].
    \end{equation*}
    According to what we just showed, the commutator vanishes for those
    values $\bb k$ for which the set of solutions to $q_1 + q_2 \in \bb k$
    is a product,
    because in those cases the projector $\hat\pi_{\bb k,S}$ factorizes.
    Otherwise,
    \begin{align*}
        \bigl[
            \hat\pi_{\bb k,S} (\hat X^{\bb k}_{S_1} &\otimes \hat\id_{S_2})
            \hat\pi_{\bb k,S},
            \hat\pi_{\bb k,S} (\hat\id_{S_1},\hat X^{\bb k}_{S_2})
            \hat\pi_{\bb k,S}
        \bigr] \\
        &= \bigoplus_{\substack{q_1+q_2\in\bb k\\q_1'+q_2'\in\bb k\\
        q_1''+q_2''\in\bb k}}
        (\hat\pi_{q_1,S_1} \otimes \hat\pi_{q_2,S_2})
        (\hat X^{\bb k}_{S_1} \otimes \hat\id_{S_2})
        (\hat\pi_{q_1'',S_1} \otimes \hat\pi_{q_2'',S_2})
        (\hat\id_{S_1} \otimes \hat X^{\bb k}_{S_2})
        (\hat\pi_{q_1',S_1} \otimes \hat\pi_{q_2',S_2})
        - (\leftrightarrow)
        \\
        &= \bigoplus_{\substack{q_1+q_2\in\bb k\\q_1'+q_2'\in\bb k}}
        \chi_{\bb k}(q_1' + q_2)\,
        \hat\pi_{q_1,S_1} \hat X^{\bb k}_{S_1} \hat\pi_{q_1',S_1}
        \otimes \hat\pi_{q_2,S_2} \hat X^{\bb k}_{S_2} \hat\pi_{q_2',S_2}
        - (\leftrightarrow),
    \end{align*}
    where ``$(\leftrightarrow)$'' is the second term of the commutator.
    We find that the second term is identical to the first except
    that the characteristic function must be replaced with
    $\chi_{\bb k}(q_1+q_2')$.
\end{proof}

\section[Explicit Example: Qutrit and Qubit as QRFs for $\gpU(1)$]
{Explicit Example: Qutrit and Qubit as QRFs for $\boldsymbol{\gpU(1)}$}
\label{app:qutrit_qubit}

\noindent
The purpose of this appendix is to provide the technical details for example~\ref{exa:qubit_qutrit}.

\medskip
\paragraph*{Setup.}
Let $A$ and $B$ be a qutrit and qubit respectively,
\ie\ $\H_A \iso \bbC^3$ and $\H_B \iso \bbC^2$, that carry angular
momentum $J_A \in \{-1,0,1\}$ and $J_B \in \{-1,1\}$, respectively.
Technically, we introduce orthonormal bases
$\{\ket{-1}_A,\ket{0}_A,\ket{1}_A\} \subset \H_A$
and $\{\ket{-1}_B,\ket{1}_B\} \subset \H_B$,
and define the angular momentum operators
\begin{equation*}
    \hat J_A := -\ketbra{-1}{-1}_A + 0 \ketbra{0}{0}_A + \ketbra{1}{1}_A,
    \qquad \hat J_B := -\ketbra{-1}{-1}_B
    + \ketbra{1}{1}_B.
\end{equation*}
Exponentiating the angular momentum produces representations of the
Abelian group $G = \gpU(1)$:
\begin{equation*}
    \hat U_A(\theta) := \ex^{-\i \hat J_A \theta},
    \qquad \hat U_B(\theta) := \ex^{-\i \hat J_B \theta}
\end{equation*}
for $\theta \in \gpU(1) = (\bbR/2\pi,+)$.
We fix the Haar measure for $\gpU(1)$ to be the Lebesgue measure
$\der\theta$ on $\bbR/2\pi$.
Hence, $|\gpU(1)| = \int_{\gpU(1)} \der g
= \int_0^{2\pi}\der\theta = 2\pi$.

We then define the states
\begin{align}
    \ket{\theta}_A &:= \hat U_A(\theta)\,
    \frac{1}{\sqrt{3}}\bigl( \ket{-1}_A + \ket{0}_A + \ket{1}_A \bigr)
    = \frac{1}{\sqrt 3}\bigl( \ex^{\i\theta} \ket{-1}_A
    + \ket{0}_A + \ex^{-\i\theta} \ket{1}_A \bigr), \label{eq:app_qutrit_qubit_orientation_1} \\
    \ket{\theta}_B &:= \hat U_B(\theta)\,
    \frac{1}{\sqrt{2}} \bigl( \ket{-1}_B + \ket{1}_B \bigr)
    = \frac{1}{\sqrt 2} \bigl( \ex^{\i\theta} \ket{-1}_B
    + \ex^{-\i\theta} \ket{1}_B \bigr).
    \label{eq:app_qutrit_qubit_orientation_2}
\end{align}
They satisfy
\begin{equation*}
    \int_0^{2\pi}\der\theta\, \ketbra{\theta}{\theta}_A
    = \frac{2\pi}{3}\, \hat\id_A,
    \qquad \int_0^{2\pi}\der\theta\, \ketbra{\theta}{\theta}_B
    = \pi\,\hat\id_B,
\end{equation*}
and by definition $\hat U_A(\theta')\ket{\theta}_A = \sket{\theta'+\theta}_A$,
$\hat U_B(\theta')\ket{\theta}_B = \sket{\theta'+\theta}_B$
for all $\theta,\theta' \in \gpU(1)$.
Thus, by setting
\begin{equation*}
    \hat\gamma_A(\theta) := 
    \frac{3}{2\pi} \ketbra{\theta}{\theta}_A,
    \qquad \hat\gamma_B(\theta) :=
    \frac{1}{\pi} \ketbra{\theta}{\theta}_B,
\end{equation*}
$A$ and $B$ become rank-one QRFs for the Abelian symmetry group $\gpU(1)$.
Hence, proposition~\ref{prop:POV_Abelian} applies.
Note that
\begin{equation*}
    \braket{\theta'}{\theta}_A = \frac{1}{3} + \frac{2}{3}\cos(\theta - \theta'),
    \qquad
    \braket{\theta'}{\theta}_B = \cos(\theta - \theta'),
\end{equation*}
and so both $A$ and $B$ are non-ideal QRFs.

\medskip
\paragraph*{Perspective of $\boldsymbol A$.}
The irreducible representations of $\gpU(1)$ are labelled by
the integer angular momentum values.
Specifically, we have $\sigma_A = \{-1,0,1\}$,
$\sigma_B = \{-1,1\}$ and hence
$\sigma_{AB} = \{-2,-1,0,1,2\}$.
We now apply proposition~\ref{prop:POV_Abelian}
with $B$ in the role of $S$.
From \cref{eq:kappa}, we can deduce that the charge diagram
is given by the left diagram in example~\ref{exa:qubit_qutrit}.
There are three $\bb k$-sectors:
$\bb k_\pm := \{\pm 1\}$ and $\bb k_0 := \{-1,1\}$.
They are in one-to-one correspondence with the $k$-sectors
of the centre of $\A^{B:A|L} \iso \A^{B:A|A}$.
The angular momentum charges contained in those sectors are
$\bb\kappa^{-1}(\bb k_{\pm}) = \{\pm 1,\pm 2\}$
and $\bb\kappa^{-1}(\bb k_0) = \{0\}$.

Consider a generic relative operator:
\begin{equation}
    \hat F^{|L} = \frac{3}{2\pi}
    \int_0^{2\pi}\der\theta\, \ketbra{\theta}{\theta}_A \otimes
    \mathsf U_B(\theta)[\hat f_B],
    \qquad \hat f_B = \sum_{\pm,\pm'} f_{\pm\pm'} \ketbra{\pm 1}{\pm'1}_B,
    \qquad f_{\pm\pm'} :=
    \sbra{\pm 1}_B \hat f_B \sket{\pm'1}_B.
    \label{eq:app_qutrit_qubit_tmp_0}
\end{equation}
Without the need of any computation,
we can read off proposition~\ref{prop:POV_Abelian} that
\begin{align}
    \hat F^{|A} = \mathsf V^{\to A}[\hat F^{|L}]
    = \bigl( \ketbra{-1}{-1} + \ketbra{-2}{-2} \bigr)_{\tilde A}
    &\otimes f_{--} \ketbra{-1}{-1}_B \notag \\
    + \bigl( \ketbra{1}{1} + \ketbra{2}{2} \bigr)_{\tilde A}
    &\otimes f_{++} \ketbra{1}{1}_B \notag \\
    + \ketbra{0}{0}_{\tilde A}
    &\otimes \bigl( f_{--}\ketbra{-1}{-1} + f_{-+} \ketbra{-1}{1}
    + f_{+-} \ketbra{1}{-1} + f_{++}\ketbra{1}{1} \bigr)_B
    \label{eq:app_qutrit_qubit_tmp_1}
\end{align}
where $\tilde A$ is the regular QRF into which $A$ is embedded.
The first two lines are the sectors $\bb k_\mp$
and the third line is the sector $\bb k_0$.
It also follows from proposition~\ref{prop:POV_Abelian}
that the most general element of $\A^{B:A|A}$ is of the form
\begin{align}
    \hat A^{|A} =
    \bigl( \ketbra{-1}{-1} + \ketbra{-2}{-2} \bigr)_{\tilde A}
    &\otimes \mu \ketbra{-1}{-1}_B \notag \\
    + \bigl( \ketbra{1}{1} + \ketbra{2}{2} \bigr)_{\tilde A}
    & \otimes \nu\ketbra{1}{1}_B \notag \\
    + \ketbra{0}{0}_{\tilde A} 
    &\otimes \bigl( \alpha\ketbra{-1}{-1} + \beta\ketbra{-1}{1}
    + \gamma\ketbra{1}{-1} + \delta\ketbra{1}{1} \bigr)_B
    \label{eq:app_qutrit_qubit_tmp_2}
\end{align}
with $\alpha,\beta,\gamma,\delta,\mu,\nu \in \bbC$.
Clearly, not all elements in $\A^{B:A|A}$ are relative operators,
but relative operators generate $\A^{B:A|A}$.
Note also that $\dim\A^{B:A|A} = 6 > 4 = \dim\O_B \geq \dim\mathsf R_A[\O_B]$.
So when observed relative to the QRF $A$,
$B$ is in some sense ``back-reacted'' onto by the QRF
to reveal further degrees of freedom.

Let us explicitly see how relational operators generate $\A^{B:A|A}$.
For this, we consider the relational operators $\hat T^{|A}$,
$\hat M^{|A}$, $\hat N^{|A}$ and $\hat X^{|A}$
obtained when inserting the operators
\begin{equation*}
    \hat t = \alpha\ketbra{-1}{-1} + \beta\ketbra{-1}{1}
    + \gamma\ketbra{1}{-1} + \delta\ketbra{1}{1}
    \qquad \hat m = \mu\ketbra{-1}{-1},
    \qquad \hat n = \nu\ketbra{1}{1},
    \qquad \hat x = \ketbra{-1}{1} + \ketbra{1}{-1}
\end{equation*}
in place of $\hat f_B$;
$\hat x$ is the Pauli-$x$ operator.
They are each of the form given by equation \eqref{eq:app_qutrit_qubit_tmp_1}.
Since $\hat x$ has vanishing diagonal elements,
$\hat X^{|A}$ has support only on the $\bb k_0$-sector.
Furthermore,
$(\hat X^{|A})^2 = \ketbra{0}{0}_{\tilde A} \otimes
\bigl( \ketbra{-1}{-1} + \ketbra{1}{1} \bigr)_B$
is the orthogonal projector onto the $\bb k_0$-sector.
Now it is not hard to see that
\begin{equation*}
    \hat A^{|A} = \bigl( \hat\id - (\hat X^{|A})^2 \bigr) \cdot
    \bigl( \hat M^{|A} + \hat N^{|A} \bigr)
    + (\hat X^{|A})^2\, \hat T^{|A}.
\end{equation*}
Note that $\hat\id$ is also a relative operator,
obtained from $\hat f_B = \hat\id_B$.
Thus, indeed, every element of $\A^{B:A|A}$ can be written as a polynomial
of relative operators.

Let us now turn to the specific example of back-reaction from example~\ref{exa:qubit_qutrit}.
There, we studied the state 
\begin{equation}
    \hat\rho^{|L} = \ketbra{\psi_\varphi}{\psi_\varphi}_A \otimes \hat\sigma_B,
    \qquad \ket{\psi_\varphi} = \frac{1}{\sqrt 2}(\ket{\theta_1} + e^{i \varphi} \ket{\theta_2}),
\end{equation}
where specifically, $\theta_1 = 0$ and $\theta_2 = 2\pi/3$
and $\hat\sigma_B = q\ketbra{-1}{-1}_B + (1-q)\ketbra{1}{1}_B$,
$0 \leq q \leq 1$, is a general $G$-invariant state on $B$.

We now analyse this state from the perspective of $A$.
Since we are only going to consider $G$-invariant degrees of freedom in our state,
we can, without loss of generality, transform the $G$-twirled version of $\hat\rho^{|L}$
into the perspective of $A$ instead of $\hat\rho^{|L}$.
Applying the $G$-twirl gives
\begin{equation*}
    \mathsf G_{AB}[\hat\rho^{|L}] = \frac{1}{2\pi} \int_0^{2\pi} \der\theta\,
    \mathsf U_A(\theta)\bigl[ \sketbra{\psi_\varphi}{\psi_\varphi}_A \bigr]
    \otimes \hat\sigma_B
    = \frac{1}{3} \bigl( p_+ \ketbra{-1}{-1}
    + p_0 \ketbra{0}{0}
    + p_- \ketbra{1}{1} \bigr)_A \otimes \hat\sigma_B,
\end{equation*}
where
\begin{equation*}
    p_\pm = 1 + \cos(\theta_2 - \theta_1 \pm \varphi),
    \qquad p_0 = 1 + \cos\varphi.
\end{equation*}
Here we have used statement 5 in proposition~\ref{prop:G_twirl}
(for Abelian groups, $\mathsf G = \bigoplus_r \Pi_r$,
where $\Pi_r$ is the orthogonal projection superoperator
onto the charge-$r$ subspace)
and \cref{eq:app_qutrit_qubit_orientation_1}.
We could now explicitly apply the QRF jump $\mathsf W^\dagger \circ \mathsf E$
to compute the state $\mathsf G_{\tilde A}[\hat\rho^{|A}]$ from the perspective of $A$ directly.
However, it will be easier to do so in each $r$-sector separately,
given that we have already decomposed our state into sectors of total charge.
The lowest total charge is $r=-2$, stemming from $-1$ in $A$ and $-1$ in $B$.
With the expressions for $\hat W^\dagger$ and $\hat E$
in \cref{eq:regular_inverse_jump_core} and \cref{eq:embedding} respectively,
we compute
\begin{equation*}
    \mathsf W^\dagger \circ \mathsf E \Biggl[
        \frac{1}{3} p_+ \ketbra{-1}{-1}_A \otimes q \ketbra{-1}{-1}_B
    \Biggr]
    = \frac{3}{2\pi} \int_0^{2\pi} \der\theta\,\der\theta' \frac{\ex^{2 \i(\theta-\theta')}}{3}
    \sketbra{\theta}{\theta'}_{\tilde A} \otimes
    \frac{1}{3} p_+ q \ketbra{-1}{-1}_B,
\end{equation*}
where we have used $\braket{\theta}{-1}_A = \ex^{-\i\theta} / \sqrt 3$.
But now we do not have to compute the integrals:
we already know that the first tensor factor must be proportional to $\ketbra{-2}{-2}_{\tilde A}$, since we are in the $r = 2$ charge sector;
and from the unitarity of the QRF jump and the normalization of $\mathsf G_{AB}[\hat\rho^{|L}]$,
we know that the proportionality factor must be one, \ie, 
\begin{equation*}
    \frac{3}{2\pi} \int_0^{2\pi} \der\theta\,\der\theta' \frac{\ex^{2 \i(\theta-\theta')}}{3}
    \sketbra{\theta}{\theta'}_{\tilde A}
    = \sketbra{-2}{-2}_{\tilde A}.
\end{equation*}
Similarly, we see that for $r=-1$, the term in $\mathsf G_{\tilde A}[\hat\rho^{|A}]$ must be $\ketbra{-1}{-1}_A \otimes \frac{1}{3} p_0 q \ketbra{-1}{-1}_B$,
and so on.
Overall,
\begin{align}
    \mathsf G_{\tilde A}[\hat\rho^{|A}] =
    \frac{1}{3} \bigl( p_+ \ketbra{-2}{-2} + p_0 \ketbra{-1}{-1} \bigr)_{\tilde A}
    &\otimes q \ketbra{-1}{-1}_B \notag\\
    + \frac{1}{3} \ketbra{0}{0}_{\tilde A} &\otimes
    \bigl( p_- q \ketbra{-1}{-1}_B + p_+ (1-q) \ketbra{1}{1}_B \bigr) \notag\\
    + \frac{1}{3} \bigl( p_- \ketbra{2}{2} + p_0 \ketbra{1}{1} \bigr)_{\tilde A}
    &\otimes (1-q) \ketbra{1}{1}_B.
\end{align}
Clearly, $\mathsf G_{\tilde A}[\hat\rho^{|A}] \notin \A^{B:A|A}$,
as it is not an expression of the form of \eqref{eq:app_qutrit_qubit_tmp_2},
meaning that $\mathsf G_{\tilde A}[\hat\rho^{|A}]$ is non-trivial in the extra particle
$\overline{S{:}A}|A$.
To isolate the information in $\mathsf G_{\tilde A}[\hat\rho^{|A}]$ accessible using
$\A^{B:A|A}$, we consider the reduced density operator $\hat\rho^{B:A|A} \in \A^{B:A|A}$
defined implicitly as
\begin{equation}
    \tr\bigl[ \hat\rho^{B:A|A} \hat X^{|A} \bigr]
    = \tr\bigl[ \hat\rho^{|A} \hat X^{|A} \bigr]
    \quad \forall\, \hat X^{|A} \in \A^{B:A|A}.
    \label{eq:app_qutrit_qubit_trace_condition}
\end{equation}
Since $\A^{B:A|A}$ is $G$-invariant, also $\tr\bigl[ \hat\rho^{|A} \hat X^{|A} \bigr]
= \tr\bigl[ \mathsf G_{\tilde A}[\hat\rho^{|A}] \hat X^{|A} \bigr]$.
For a general operator $\hat X^{|A} \in \A^{B:A|A}$ of the form given in
\cref{eq:app_qutrit_qubit_tmp_2}, we get
\begin{equation*}
    \tr\bigl[ \mathsf G_{\tilde A}[\hat\rho^{|A}] \hat X^{|A} \bigr]
    = \frac{1}{3} \mu (p_+ + p_0) q
    + \frac{1}{3} \nu (p_- + p_0) (1-q)
    + \frac{1}{3} \alpha p_- q + \frac{1}{3} \delta p_+ (1-q).
\end{equation*}
Generally, a state $\hat\rho^{B:A|A} \in \A^{B:A|A}$ also has the form given in \cref{eq:app_qutrit_qubit_tmp_2},
\begin{align*}
    \hat\rho^{B:A|A} =
    \bigl( \ketbra{-1}{-1} + \ketbra{-2}{-2} \bigr)_{\tilde A}
    &\otimes \tilde\mu \ketbra{-1}{-1}_B \notag \\
    + \bigl( \ketbra{1}{1} + \ketbra{2}{2} \bigr)_{\tilde A}
    & \otimes \tilde\nu\ketbra{1}{1}_B \notag \\
    + \ketbra{0}{0}_{\tilde A} 
    &\otimes \bigl( \tilde\alpha\ketbra{-1}{-1} + \tilde\beta\ketbra{-1}{1}
    + \tilde\gamma\ketbra{1}{-1} + \tilde\delta\ketbra{1}{1} \bigr)_B,
\end{align*}
where $\tilde\alpha, \tilde\beta, \tilde\gamma, \tilde\delta, \tilde\mu, \tilde\nu \in \bbC$
such that furthermore $(\hat\rho^{B:A|A})^\dagger = \hat\rho^{B:A|A}$, $\hat\rho^{B:A|A} > 0$
and $\tr\bigl[\hat\rho^{B:A|A}\bigr] = 1$,
\ie\ such that $\hat\rho^{B{:}A|A}$ is a state.
This gives
\begin{equation*}
    \tr\bigl[ \hat\rho^{B:A|A} \hat X^{|A} \bigr]
    = \tilde\mu \mu + \tilde\nu \nu
    + \tilde\alpha \alpha + \tilde\gamma\beta + \tilde\beta\gamma + \tilde\delta\delta,
\end{equation*}
leading us to
\begin{align}
    \hat\rho^{B:A|A} =
    \frac{1}{3} \bigl( \ketbra{-2}{-2} + \ketbra{-1}{-1} \bigr)_{\tilde A}
    &\otimes (p_+ + p_0) q \ketbra{-1}{-1}_B \notag\\
    + \smash{\frac{1}{3}} \ketbra{0}{0}_{\tilde A} &\otimes
    \bigl( p_- q \ketbra{-1}{-1}_B + p_+ (1-q) \ketbra{1}{1}_B \bigr) \notag\\
    + \frac{1}{3} \bigl( \ketbra{2}{2} + \ketbra{1}{1} \bigr)_{\tilde A}
    &\otimes (p_- + p_0) (1-q) \ketbra{1}{1}_B.
\end{align}
It is not hard to check that this is a state.
Since it satisfies the trace condition in \cref{eq:app_qutrit_qubit_trace_condition}
and is a state contained in $\A^{B:A|A}$, $\hat\rho^{B:A|A}$ is the sought-after density operator.
Note that finding the reduced density operator on $B{:}A|A$
amounts to performing an average of $p_+$ and $p_0$ and of $p_-$ and $p_0$.

From \cref{eq:app_qutrit_qubit_tmp_1} it follows that for a general relational operator
$\hat F^{|L} = \mathsf R_A[\hat f_B]$,
\begin{equation*}
    \tr\bigl[ \hat\rho^{B:A|A} \hat F^{|A} \bigr]
    = \frac{1}{3} \bigl(
        f_{--} (p_+ + p_0 + p_-) q
        + f_{++} (p_- + p_0 + p_+) (1-q)
    \bigr),
\end{equation*}
where $f_{\pm,\pm'}$ are defined as in \cref{eq:app_qutrit_qubit_tmp_0}.
If we insert the specific values $\theta_1 = 0$ and $\theta_2 = 2\pi/3$, \ie\ $\Delta\theta = 2\pi/3$, then
\begin{equation*}
    p_+ + p_- + p_0 = 3 + \cos\varphi + \cos(2\pi/3 + \varphi) + \cos(2\pi/3 - \varphi) = 3,
\end{equation*}
so that the trace above becomes independent of $\varphi$.
Relational operators therefore cannot be used to measure the phase.
However, consider the relational operator $\hat Q^{|A}$ obtained by inserting
$\hat f_B = \ketbra{+}{+}_B$, where $\ket{+} = \smash{\frac{1}{\sqrt{2}} (\ket{+1} + \ket{-1})}$,
in equation~\eqref{eq:app_qutrit_qubit_tmp_0} and then moving to the perspective of $A$
as in equation~\eqref{eq:app_qutrit_qubit_tmp_1}.
It is easy to see that $(\hat Q^{|A})^2$ is not a relational operator and its
expectation value depends on the phase:
\begin{equation*}
    \tr\bigl[ \hat\rho^{B:A|A} (\hat Q^{|A})^2 \bigr]
    = \frac{1}{3} - \frac{1}{24} \cos\varphi + \frac{\sqrt 3}{24} \left( 2q - 1 \right) \sin\varphi.
\end{equation*}
From the perspective of $L$, we can interpret this result as follows.
The operator $\hat Q^{|L} = (\mathsf V^{\to A})^{-1}[Q^{|A}] = \mathsf R_A[\ketbra{+}{+}_B]$
is the POVM element resulting from relationalizing the projector $\sketbra{+}{+}_B$ with respect to the QRF $A$.
It can be seen as the POVM element arising from the quantum operation
defined by the Kraus operator $\hat K^{|L} = (\hat Q^{|L})^{1/2}$.
Note that $\hat K^{|L}$ is not a relational operator, but it belongs to $\A^{B{:}A|L}$.
Likewise, we may interpret $(\hat Q^{|L})^2$ as the POVM element corresponding
to the quantum operation defined by the Kraus operator $(\hat K^{|L})^2$
realised by repeating the afore-mentioned measurement on $B$ relative to $A$.
As $A$ is non-ideal, repeating such a measurement entails a back-reaction of $B$ onto $A$,
implying that information about $\varphi$,
originally stored in $A$ from the perspective of $L$, is accessible in the subsystem $B{:}A|L$.
\medskip
\paragraph*{Perspective of $\boldsymbol B$.}
Similarly, the perspective of $B$ is described by the right diagram
in example~\ref{exa:qubit_qutrit}.
This time, there are four $\bb k$-sectors:
$\bb k_{\pm} = \{\pm 2\}$, $\bb k_\times = \{0\}$, $\bb k_0 = \{-1,1\}$.
The angular momentum charges contained in those sectors are
$\bb\kappa^{-1}(\bb k_{\pm}) = \{\pm 2\}$,
$\bb\kappa^{-1}(\bb k_\times) = \{-1,1\}$,
$\bb\kappa^{-1}(\bb k_0) = \{0\}$.
Note that $\bb k_\times$ consists of two disconnected parts in the figure,
because $\bb\kappa^{-1}(\{0\}) = \{-1,1\}$, which does not contain $r = 0$.

Let $\hat f_A$ be a general operator on $A$ with matrix elements
$f_{ij} = \bra{i}_A \hat f_A \ket{j}_A$,
where $i,j \in \{-,+,0\}$ and $\ket{\pm} := \ket{\pm 1}$.
The corresponding relative operator $\hat F^{|L}$ becomes,
according to proposition~\ref{prop:POV_Abelian},
\begin{align*}
    \hat F^{|B} = \mathsf V^{\to B}[\hat F^{|L}]
    = \ketbra{-2}{-2}_{\tilde B}
    &\otimes f_{--} \ketbra{-1}{-1}_A \\
    + \ketbra{2}{2}_{\tilde B}
    &\otimes f_{++} \ketbra{1}{1}_A \\
    + \bigl(\ketbra{-1}{-1} + \ketbra{1}{1} \bigr)_{\tilde B}
    &\otimes f_{00} \ketbra{0}{0}_A \\
    + \ketbra{0}{0}_{\tilde B}
    &\otimes \bigl( f_{--} \ketbra{-1}{-1} + f_{-+} \ketbra{-1}{1}
    + f_{+-} \ketbra{1}{-1} + f_{++} \ketbra{1}{1} \bigr)_A.
\end{align*}
The first two rows are the sector $\bb k_\mp$,
the third row is $\bb k_\times$
and the last row is $\bb k_0$.
The most general element of $\A^{A:B|B}$ is obtained by
replacing all individual occurrences of $f_{ij}$ with
arbitrary complex numbers.
This way, we see that $\dim\A^{A:B|B} = 7 < 9 = \dim\O_A$.
So observing $A$ relative to the ``smaller'' QRF,
$B$, we effectively lose  degrees of freedom;
the ``back-reaction'' of $B$ onto $A$ does not reveal
enough further degrees of freedom to make up in numbers for the lost ones.

\section{QRFs for the Galilei Group in One Dimension}
\label{app:Galilei}

\noindent
\paragraph*{Galilei Group in One Dimension and its Central Extension.}
The Galilei group $\group{Gal} \iso (\bbR^2,+)$ in one dimension
has the group multiplication $(a,v) \cdot (a',v') = (a+a',v+v')$,
corresponding to addition of translations ($a$, $a'$)
and boosts ($v$, $v'$).
The mass-$m$ representation $\hat U$~\cite{Bargmann1954},
defined in example~\ref{exa:Galilei} and acting on a quantum particle
with Hilbert space $L^2(\bbR)$, is projective:
\begin{equation*}
    \hat U(a',v') \hat U(a,v)
    = \exp\left( \i\frac{m}{2} (av' - a'v)\right) \hat U(a+a',v+v').
\end{equation*}
To make it unitary, one introduces a new generator $m\,\hat\id$
with corresponding parameter $\theta \in \bbR$,
leading to a new class of Galilei-transformations of the particle
besides translations and boosts, which we here call
\emph{$\theta$-transformations}:
\begin{equation*}
    \hat U(\theta,a,v) := \ex^{\i m \theta} \hat U(a,v).
\end{equation*}
Note that the new generator commutes with the other two generators,
it lies in the \emph{centre} of the \emph{universal enveloping algebra}
of the Lie algebra of $\group{Gal}$.
Choosing the new generator to be in the centre makes sense because
its purpose is to cancel a phase in the original representation.
Indeed, we see that for $\hat U(\theta,a,v)$ to be non-projective,
we must have the group multiplication
\begin{equation*}
    (\theta',a',v') \cdot (\theta,a,v)
    = \Bigl(\theta'+\theta+\frac{av'-a'v}{2},a+a',v+v'\Bigr)
\end{equation*}
of a new group consisting of three real parameters.
This new group is the \emph{central extension} $\group{CGal}$ of $\group{Gal}$,
and $\hat U(\theta,a,v)$ is a unitary representation of $\group{CGal}$~%
\cite{LevyLeblond1963,Giulini1996}.
Note that because $\hat U(a,v)$ is irreducible, so is $\hat U(\theta,a,v)$.

Clearly, the transformation $(\theta,a,v)$
is just a regular Galilei transformation $(a,v)$
with some additional phase dictated by $\theta$.
Thus, when building QRFs for $\group{Gal}$,
we may without loss of generality consider $\group{CGal}$ instead.
$\theta$-translations are however not unphysical:
intuitively, $\theta$ is ``conjugate'' to the mass $m$ of the particle
and thus related to shifts in proper time.
This can be made precise by considering Galilei transformations
as the classical limit of Poincaré transformations, of which
$\theta$-transformations then arise as a remnant~\cite{Greenberger2001}.
But since they act on quantum states with fixed $m$ only as a phase
they are not commonly encountered.

\medskip
\paragraph*{Particles as QRFs for Gal.}
A quantum particle with mass-$m$ representation $\hat U(\theta,a,v)$
obviously cannot be an ideal QRF for CGal, since $\theta$-transformations
merely act as a phase.
Furthermore, it is also not an ideal QRF for Gal, \ie,
when $\theta$-transformations are not considered. 

To see this, one can write the most general state $\hat\rho$
of the particle and the most general orientation POVM element $\hat\gamma$
in the improper position basis,
and notice that in order for $\hat\rho$ to have a definite orientation
under translations alone, it must be a position eigenstate.
Similarly, to have definite orientation under boosts,
it must be a velocity eigenstate.
Due to Heisenberg uncertainty $\Delta x\Delta\dot{x} \geq 1/(2m)$, 
this can never be the case simultaneously.
Hence, any QRF for Gal built from a single quantum particle
is always not ideal.

Notice that as $m\to\infty$, the uncertainty relation between
position and velocity disappears,
and we expect that ideal QRFs are possible in this limit.
This only works because the Galilei boosts are changes in velocity
and not momentum:
while it is possible to have states of the particle with
arbitrarily sharp position and velocity at the expense of large $m$,
one can never have states with arbitrarily sharp position and momentum.

\medskip
\paragraph*{Rank-One Particle QRFs.}
To build a QRF for $\group{CGal}$ (and also for $\group{Gal}$)
using $L^2(\bbR)$, the mass-$m$ representation $\hat U(\theta,a,v)$,
and rank-one POVM elements,
we choose a state $\ket{0,0,0} = \ket{e}$
and set $\ket{\theta,a,v} := \hat U(\theta,a,v) \ket{0,0,0}$.
Because the representation is irreducible,
it follows from Schur's lemma that $\int_{\group{CGal}} \der g\,
\ketbra{\theta,a,v}{\theta,a,v} \propto \hat\id$
regardless of which state $\ket{e}$ we choose;
all choices eventually lead to a QRF.
A short computation shows, using the Haar measure
$\der g := \der\theta\,\der a\,\der v$, that
\begin{multline*}
    \int_{\group{CGal}} \der g\, \sketbra{\theta,a,v}{\theta,a,v}
    = \int\der\theta\,\der a\,\der v\,\der x'\,\der x\,
    \sbraket{x'}{0,0,0} \sbraket{0,0,0}{x}
    \ex^{\i mv(x' - x)} \ex^{-\i(a + vt) \hat p} \sketbra{x'}{x}
    \ex^{\i(a + vt) \hat p} \\
    = \int\der\theta\,\der a\,\der v\,\der x'\,\der x\,
    \sbraket{x'}{0,0,0} \sbraket{0,0,0}{x}
    \ex^{\i mv(x'-x)}
    \ketbra{x'+a}{x+a}
    = \frac{2\pi}{m}\int\der\theta\,\der x\,
    \bigl|\sbraket{x}{0,0,0}\bigr|^2\, \hat\id \\
    = \frac{2\pi \braket{e}{e} \int\der\theta}{m} \hat\id,
\end{multline*}
where in the second step we have first substituted
$a \rightsquigarrow a - vt$ before applying the shifts.
Thus, when setting $\hat\gamma_A(\theta,a,v) := \frac{1}{c}
\sketbra{\theta,a,v}{\theta,a,v}$,
the factor $c$ is infinite.
But since this is due to the central extension of the Galilei group,
which we performed not out of physical necessity
but for mathematical convenience,
we will not worry about it and treat it as an infinite symbol.
In the language of~\cite{Perelomov1986},
the orientation POVM is a \emph{coherent state system}.

\medskip
\paragraph*{Squeezed Coherent States.}
A useful choice~\cite{Garmier2023} for the state $\ket{0,0,0}$
is a \emph{squeezed coherent state}:
\begin{equation}
    \sbraket{x}{0,0,0} := C\, \exp\left( -m\frac{x^2}{2\omega^2} \right),
    \label{eq:coherent_seed_state}
\end{equation}
where $\omega > 0$ is the squeezing parameter
and $C = m^{1/4}\omega^{-1/2}\pi^{-1/4}$.
One easily shows that
\begin{equation*}
    \langle \hat x \rangle_{0,0,0}
    = \langle \hat u \rangle_{0,0,0} = 0,
    \qquad \langle \hat{\Delta x}{}^2 \rangle_{0,0,0} = \frac{\omega^2}{2m},
    \qquad \langle \hat{\Delta u}{}^2 \rangle_{0,0,0} = \frac{1}{2m\omega^2},
\end{equation*}
where $\hat u = \hat p / m$ is the velocity operator.
The Heisenberg uncertainty relation is saturated.
In fact, all states which saturate the Heisenberg uncertainty relation
are phase-space-shifted versions of \cref{eq:coherent_seed_state},
and taking all possible values $\omega > 0$ of the squeezing parameter~%
\cite{Ballentine2014}.
All other orientation states are then generated from $\ket{0,0,0}$
through application of $\hat U(\theta,a,v)$,
and one finds (we set $t = 0$)
\begin{equation*}
    \sbraket{x}{\theta,a,v} =
    C \exp\left(
        -m\frac{(x - a)^2}{2\omega^2}
        + \i m \left( vx + \theta - \frac{va}{2} \right)
    \right),
    \qquad \langle \hat x \rangle_{\theta,a,v} = a,
    \qquad \langle \hat u \rangle_{\theta,a,v} = v,
\end{equation*}
while the variances remain unchanged.
The overlaps are
\begin{equation*}
    \sbraket{\theta',a',v'}{\theta,a,v}
    = \exp\left(
        -m \frac{(a - a')^2}{4\omega^2} - m\omega^2\frac{(v - v')^2}{4}
    \right)
    \exp\left(
        \i m \frac{va' - v'a}{2} + \i m (\theta - \theta')
    \right).
\end{equation*}
They approach a $\delta$-distribution in the shift and boost parameters
when taking the limit $m \to \infty$,
yielding an ideal QRF for Gal, as we expected above.

\medskip
\paragraph*{Details for Figure~\ref{fig:Galilei_fuzziness}.}
Finally, we fill in the details leading up the example
shown in figure~\ref{fig:Galilei_fuzziness}.
We use the rank-one orientation POVM elements described above.
From the perspective of the laboratory, we consider a state of the form
$\hat\rho^{|L} = \hat\sigma_A \otimes \hat\varsigma_S$,
where $\hat\varsigma_S$ is arbitrary, and $\hat\sigma_A = \sketbra{\phi}{\phi}_A$ is a superposition
of POVM states:
\begin{equation}
    \hat\sigma_A = \sketbra{\phi}{\phi}_A,
    \qquad
    \ket{\phi}_A = \frac{1}{\sqrt{2}} \bigl( \sket{\theta_1,a_0,v_0}
    + \sket{\theta_2,a_2,v_2} \bigr).
\end{equation}
According to \cref{eq:fuzziness}, the probability
\begin{equation}
    p(\theta,a,v) = \tr\left( \frac{1}{c}\sketbra{\theta,a,v}{\theta,a,v}
    \hat\sigma_A \right)
    = \frac{1}{c} \sbra{\theta,a,v} \hat\sigma_A \sket{\theta,a,v}
\end{equation}
of measuring the orientation $(\theta,a,v)$
captures the more mixed nature of the effective state when observing $S$ relative to $A$
compared to when observing it relative to $L$.
The effective state is
\begin{equation}
    \hat\varsigma'_S = \int\der a\,\der v\, p(a,v)\,
    \mathsf U_S(a,v)[\hat\varsigma_S],
    \qquad p(a,v) = \int\der\theta\, p(\theta,a,v)
    = \frac{2\pi}{m} \bigl| \braket{0,a,v}{\phi} \bigr|^2.
\end{equation}
Note that the integral $\int\der\theta$ in the definition of
$p(a,v)$ nicely cancels the infinity present in $c$, as one would expect.
We compute
\begin{multline*}
    p(a,v) = \frac{\pi}{m} \exp\left( -m\frac{(a-a_1)^2}{2\omega^2}
    - m\omega^2 \frac{(v-v_1)^2}{2}\right)
    + \frac{\pi}{m} \exp\left( -m\frac{(a-a_2)^2}{2\omega^2}
    - m\omega^2 \frac{(v-v_2)^2}{2}\right) \\
    + \frac{2\pi}{m} \exp\left( -m\frac{(a-a_1)^2 + (a-a_2)^2}{4\omega^2}
    - m\omega^2 \frac{(v-v_1)^2 + (v-v_2)^2}{4} \right)
    \cos\left( m\frac{(v_1-v_2)a - v(a_1-a_2)}{2}
    + m(\theta_1 - \theta_2) \right).
\end{multline*}
The probability density consists of two peaks centred around
$(a_1,v_1)$ and $(a_2,v_2)$, with finite cross-sections
due to Heisenberg uncertainty, as well as an interference term due
to relative phases between the two POVM states which are superposed
in the QRF state $\ket{\phi}_A$.
Notice how only the difference $\theta_1 - \theta_2$ enters into $p(a,v)$,
because only this difference has a physical influence on the QRF state
by means of a relative phase.
Finally, note that the width of the peaks vanishes in the limit
$m\to\infty$, and the interference disappears.
This is of course what we would expect from an ideal QRF.

\section{Comparisons with Other Frameworks}
\label{app:comparisons}

\noindent
The various QRF frameworks differ both in interpretation
and mathematical implementation.
Comparing interpretations is typically involved due to conceptual complexity,
and has been done for other frameworks,
including the ideal case of our framework, \eg\
in~\cite{Vuyst2025,CastroRuiz2025}.
Here we focus mostly on mathematical aspects
and compare our framework with others.
For simplicity, we will use our notation throughout.

\medskip
\paragraph*{Perspective-Neutral.}
In the \emph{perspective-neutral framework}~%
\cite{Hoehn2021,Hamette2021,AhmadAli2022},
the potential existence of an external laboratory frame is typically not assumed explicitly \apriori\footnote{
    However, see~\cite{Krumm2021}, where the perspective neutral framework is interpreted operationally in terms of an external laboratory frame.
};
instead, one argues, using ideas from gauge theory,
for the existence of a perspective-neutral structure,
which encodes the perspectives of all QRFs
while not corresponding to any of those perspectives in particular.
Mathematically, the perspective-neutral structure consists of the
\emph{kinematical Hilbert} space $\H_A \otimes \H_S$
carrying a unitary, non-projective representation $\hat U_A \otimes \hat U_S$
of some symmetry group $G$,
and the assumption that the physical states are those which transform
trivially under the representation.
As such, the framework relies on the \emph{coherent} group average.
The perspective-neutral framework can deal with generally non-compact
and non-Abelian symmetry groups,
as well as potentially non-ideal QRFs $A$ with rank-one orientation POVMs,
\ie\ $\hat\gamma_A^{|L}(g) = \frac{1}{c} \ketbra{g}{g}_A$ for $c > 0$.
Furthermore, the role of the QRF jump into the perspective of $A$
is played by the family of
\emph{Schrödinger reduction maps}
\begin{equation*}
    \hat R^{\to A}(g) := \bra{g}_A \otimes \hat\id_S.
\end{equation*}
The parameter $g \in G$ indicates the orientation at which the QRF
is gauge-fixed.
The reduction maps act on physical states in the neutral perspective
to produce states from the perspective of $A$.
The reduction maps are unitary,
and unitary QRF transformations can be defined
by appropriately concatenating reduction maps and their inverses.

In contrast, our framework does not focus only on the trivial representation
sector and as such, relies on the incoherent group average, the $G$-twirl.
Furthermore, we also construct QRF transformations by concatenating
jumps and their inverses;
however, our jumps do not rely on the idea of fixing the QRF orientation
and instead derive from the principles of
\hyperref[ppl:perspective]{perspective}
and \hyperref[ppl:invariance]{invariance}.

The perspective-neutral framework can be compared
more directly to our framework in the case where it makes sense to interpret the
perspective-neutral structure in terms of a laboratory frame~\cite{Krumm2021}.
However, leaving the operational interpretation aside, one can ask how these frameworks are related to each other at the mathematical level. As it turns out, the former can be obtained from the latter essentially by restricting to the trivial representation
(hereafter denoted with ``$r=0$''):
\begin{proposition}{}{PN}
    Let $A$ be a QRF for compact $G$ with rank-one orientation POVM
    $\hat\gamma_A^{|L}(g) = \frac{1}{c}\ketbra{g}{g}_A$, $c > 0$,
    and representation $\hat U_A(g')\ket{g}_A = \sket{g'g}_A$,
    $g,g' \in G$.
    Let $S$ carry the representation $\hat U_S$.
    Assume that the subspace $\H^{|L}_{r=0}$ of $\H^{|L}$
    carrying the trivial representation
    is non-zero; we denote the orthogonal projector onto
    $\H^{|L}_{r=0}$ with~$\hat\pi_{r=0,AS}$.
    Let the map $\hat V^{\to A}_{r=0} :=
    \hat W^\dagger (\hat E \otimes \hat\id_S) \hat\pi_{r=0,AS} :
    \H^{|L}_{r=0} \to \H_{\tilde A} \otimes \H_S$
    be given by the map of proposition~\ref{prop:general_jump}
    (including embedding of $A$ into $\tilde A$),
    restricted to the trivial representation.
    Then:
    \begin{enumerate}
        \medskip
        \item\label{item:PN_projector}
            We have
            \begin{equation}
                \hat\pi_{r=0,AS} = \frac{1}{|G|} \int_G\der g\,
                \hat U_A(g) \otimes \hat U_S(g).
            \end{equation}
            Furthermore,
            \begin{equation}
                \hat V^{\to A}_{r=0} (\hat V^{\to A}_{r=0})^\dagger
                = (\pi_{r=0,\tilde A} \otimes \sfid_S)[\hat\pi]
                := (\hat\pi_{r=0,\tilde A} \otimes \hat\id_S)
                \hat\pi (\hat\pi_{r=0,\tilde A} \otimes \hat\id_S),
                \label{eq:PN_projectors}
            \end{equation}
            where $\hat\pi$ is the orthogonal projector onto the image of
            $\hat W^\dagger(\hat E \otimes \hat\id_S)$
            defined in proposition
           ~\ref{prop:general_jump} and
            \begin{equation}
                \hat\pi_{r=0,\tilde A} = \frac{1}{|G|} \int_G\der g\,
                \hat L_{\tilde A}(g)
            \end{equation}
            is the orthogonal projector onto the trivial representation
            sector of $\tilde A$.

        \medskip
        \item
            It holds that
            \begin{equation}
                \hat V^{\to A}_{r=0} = \frac{1}{\sqrt c}
                \ket{0}_{\tilde A}\! \bra{e}_A \otimes \hat\id_S,
                \qquad \ket{0}_{\tilde A} := \int_G\der g\, \ket{g}_{\tilde A}.
                \label{eq:PN_jump}
            \end{equation}
            The map is invertible on its image by
            \begin{equation}
                \bigl(\hat V^{\to A}_{r=0}\bigr)^{-1}
                = \bigl(\hat V^{\to A}_{r=0}\bigr)^\dagger
                (\pi_{r=0,\tilde A} \otimes \sfid_S)[\hat\pi]
                = \frac{1}{|G| \sqrt c} \int_G \der g\,
                \ket{g}_A\! \bra{0}_{\tilde A} \otimes \hat U_S(g).
                \label{eq:PN_inverse_jump}
            \end{equation}

        \medskip
        \item
            There exists a QRF jump $\hat V^{\to A} : \H^{|L} \to \H^{|A}$,
            including a Hilbert space $\H^{|A} \iso \H^{|L}$,
            such that theorem~\ref{thm:QRF_POV} is satisfied and
            $\hat V^{\to A}_{r=0} = \hat V^{\to A} \hat\pi_{r=0,AS}$.
            Furthermore, the subspace $\H^{|A}_{r=0}
            := \hat V^{\to A}_{r=0} \H^{|L}_{r=0} \subset \H^{|A}$
            contains a single $k$-sector and the commutant subsystem is trivial:
            $\H^{|A}_{r=0} = \H_{r=0} \otimes \H'$,
            where relational operators, restricted to $r=0$ and transformed by
            $\hat V^{\to A}_{r=0}$, act semi-simply
            (even irreducibly) on $\H'$ and $\H_{r=0}$
            is the trivial representation space.
    \end{enumerate}
\end{proposition}

\noindent
Before proving the result, let us interpret it.
Proposition~\ref{prop:PN} shows that a specific choice of QRF jump
$\hat V^{\to A}$ in our framework reproduces the perspective-neutral
reduction map when restricted to $r=0$, with two notable peculiarities:
Firstly, we obtain the reduction map with $g=e$,
\ie\ where $A$ is gauge-fixed at the identity.
Now, we have 
\begin{equation*}
    \hat R^{\to A}(g) = \hat R^{\to A}(e)
    \bigl( \hat U^\dagger_A(g) \otimes \hat\id_S \bigr)
    = \hat U_S(g) \hat R^{\to A}(e)
    \bigl( \hat U^\dagger_A(g) \otimes \hat U^\dagger_S(g) \bigr)
    = \hat U_S(g) \hat R^{\to A}(e),
\end{equation*}
as the domain of the reduction maps is invariant
under global $G$-transformations.
But acting with $\hat U_S(g)$ is just a unitary which does not change the
subsystem structure dictated by theorem~\ref{thm:QRF_POV}.
Hence, by choosing a different QRF jump,
we can obtain any of the reduction maps when restricting to $r=0$.
Secondly, the restriction $\hat V^{\to A}_{r=0}$ includes an additional
pure state factor $\ket{0}_{\tilde A} / \sqrt{c}$;
but since this state is constant, it may be removed via an isomorphism.

Furthermore, perspective-neutral QRF transformations are obtained by
concatenating reduction maps and their inverses~\cite{Hamette2021},
and thus restricting our QRF transformations to $r=0$
reproduces the QRF transformations of the perspective-neutral framework.
Finally, in the case of an Abelian symmetry group,
it has been shown that $\sigma_S \cap (-\sigma_A)$ governs
the perspective of $A$ in the perspective-neutral approach~%
\cite{AhmadAli2022,Hausmann2025};
this is our function $\bb\kappa$ evaluated at $r=0$,
as one would expect.

In summary, we can mathematically think of the perspective-neutral approach
as our approach restricted to $r=0$,
although physically the two approaches are quite distinct.
Inversely, one can describe the transition from perspective-neutral
to our approach by the inclusion of all charges $r$;
to keep unitarity of QRF transformations,
one is forced to now also include the commutant algebra
$\A^{\overline{S:A}|L}$.

\medskip
\begin{proof}[Proof of proposition~\ref{prop:PN}]
    1. Clearly,
    \begin{equation*}
        \frac{1}{|G|} \int_G\der g\, \hat U_A(g) \otimes \hat U_S(g)
        \quad\text{and}\quad
        \frac{1}{|G|} \int_G\der g\, \hat L_{\tilde A}(g)
    \end{equation*}
    map onto coherently $G$-invariant pure states on $\H_{AS}$
    and $\H_{\tilde A}$ respectively, they are orthogonal projectors,
    and act as the identity on already coherently $G$-invariant pure
    states. Hence, they respectively equal $\hat\pi_{r=0,AS}$ and $\hat\pi_{r=0,\tilde A}$.
    Let us calculate $\hat V^{\to A}_{r=0} (\hat V^{\to A}_{r=0})^\dagger$:
    \begin{multline*}
        \hat V^{\to A}_{r=0} (\hat V^{\to A}_{r=0})^\dagger 
        = \frac{1}{|G|} \int_G\der g\,   \mathsf W^\dagger \circ (\mathsf E \otimes \sfid_S)
                \bigl[ \hat U_A(g) \otimes \hat U_S(g) \bigr] \\
        = \frac{1}{|G|} \int_G\der g\,  \hat\pi \bigl( \hat L_{\tilde A}(g) \otimes \hat\id_S \bigr) \hat\pi 
        =  \hat\pi  (\hat\pi_{r=0,\tilde A} \otimes \hat\id_S) =  (\pi_{r=0,\tilde A} \otimes \sfid_S)[\hat\pi]
    \end{multline*}
    where we used statement~\ref{item:commutation_almost_jump} of
    proposition~\ref{prop:general_jump} multiple times.
    Thus, \cref{eq:PN_projectors} follows.

    2.  First, note from that statement~\ref{item:commutation_almost_jump} of
    proposition~\ref{prop:general_jump} and the fact
    that $\hat W^\dagger (\hat E \otimes \hat\id_S)$ is an isometry 
    it follows that $\hat W^\dagger (\hat E \otimes \hat\id_S) \hat\pi_{r=0,AS} 
    = (\pi_{r=0,\tilde A} \otimes \hat\id_S) \hat W^\dagger (\hat E
    \otimes \hat\id_S)$. Using this observation, we compute:
    \begin{multline*}
        \hat V^{\to A}_{r=0} = \hat\pi_{r=0,\tilde A} \hat W^\dagger
        (\hat E \otimes \hat\id_S) \hat\pi_{r=0,AS}
        = \left( \frac{1}{|G|\sqrt c}\int_G\der g\,\der g'\,
        \hat L_{\tilde A}(g) \sket{g'}_{\tilde A}\!\sbra{g'}_A
        \otimes \hat U^\dagger_S(g') \right) \hat\pi_{r=0,AS} \\
        = \ket{0}_{\tilde A} \frac{1}{|G|\sqrt c} \int_G\der g'\, \sbra{e}_A
        \bigl(\hat U_A^\dagger(g') \otimes \hat U^\dagger_S(g')\bigr)
        \hat\pi_{r=0,AS} = \ket{0}_{\tilde A} \frac{1}{|G|\sqrt c}
        \int_G\der g'\, \bigl(\sbra{e}_A \otimes \hat\id_S \bigr)
        \hat\pi_{r=0,AS}
        = \frac{1}{\sqrt c} \ket{0}_{\tilde A}\! \sbra{e}_A \otimes \hat\id_S.
    \end{multline*}
    Because $\hat V^{\to A}_{r=0}$ is an isometry on $\H^{|L}_{r=0}$
    ($\hat\pi_{r=0,AS}$ is the identity on $\H^{|L}_{r=0}$),
    it must be invertible on its image with inverse given by the adjoint.
    According to \cref{eq:PN_projectors} the image of $\hat V^{\to A}_{r=0}$
    is the image of $(\pi_{r=0,\tilde A} \otimes \sfid_S)[\hat\pi]$,
    and hence $(\hat V^{\to A}_{r=0})^{-1} = (\hat V^{\to A}_{r=0})^\dagger
    (\pi_{r=0,\tilde A} \otimes \sfid_S)[\hat\pi]$.
    Thus,
    \begin{equation*}
        (\hat V^{\to A}_{r=0})^{-1} \hat V^{\to A}_{r=0}
        = \hat\pi_{r=0,AS},
        \qquad 
        \hat V^{\to A}_{r=0} (\hat V^{\to A}_{r=0})^{-1}
        = (\pi_{r=0,\tilde A} \otimes \sfid_S)[\hat\pi].
    \end{equation*}
    Let us compute the expression for the inverse.
    For this, we note that
    \begin{equation*}
        \hat\pi_{r=0,\tilde A} \hat R_{\tilde A}(g) \hat\pi_{r=0,\tilde A}
        = \int_G\der g'\, \hat\pi_{r=0,\tilde A}
        \sketbra{g'g^{-1}}{g'}_{\tilde A} \hat\pi_{r=0,\tilde A}
        = \frac{1}{|G|^2} \ketbra{0}{0}_{\tilde A} \int_G\der g'
        = \frac{\ketbra{0}{0}_{\tilde A}}{|G|}.
    \end{equation*}
    Hence,
    \begin{equation}
        (\pi_{r=0,\tilde A} \otimes \sfid_S)[\hat\pi]
        = \frac{\ketbra{0}{0}_{\tilde A}}{|G|} \otimes
        \underbrace{\frac{1}{c} \int_G\der g\, \hat U_S(g) \braket{e}{g}_A}
        _{=:\,\hat\Pi_S},
        \label{eq:tmp_PN_projector}
    \end{equation}
    and
    \begin{multline*}
        (\hat V^{\to A}_{r=0})^{-1}
        = (\hat V^{\to A}_{r=0})^\dagger
        (\pi_{r=0,\tilde A} \otimes \sfid_S)[\hat\pi]
        = \frac{\hat\pi_{r=0,AS}}{c} \int_G\der g\,
        \sket{e}\! \braket{e}{g}_A\!  \bra{0}_{\tilde A} \otimes \hat U_S(g) \\
        = \frac{1}{|G|c} \int_G\der g\,\der g'\,
        \sket{g'}\! \braket{e}{g}_A\!
        \bra{0}_{\tilde A} \otimes \hat U_S(g'g)
        = \frac{1}{|G|c} \int_G\der g\,\der g'\,
        \sket{g'}\! \braket{g'}{g}_A\!
        \bra{0}_{\tilde A} \otimes \hat U_S(g) \\
        = \frac{1}{|G| \sqrt c} \int_G\der g\, \ket{g}_A\!
        \bra{0}_{\tilde A} \otimes \hat U_S(g).
    \end{multline*}
    In the second step we included the projector
    $\hat\pi_{r=0,AS}$ to make clear that with $(\hat V^{\to A}_{r=0})^\dagger$
    we mean the adjoint of the map $\hat V^{\to A}_{r=0}$,
    whose domain is by definition $\H^{|L}_{r=0}$ and not all of $\H^{|L}$.

    3. First note that $\hat V^{\to A}_{r=0}$ satisfies the properties
    of a QRF jump restricted to $r = 0$. \Cref{eq:tmp_PN_projector} shows that
    $(\pi_{r=0,\tilde A} \otimes \sfid_S)[\hat\pi]$
    factorizes in the tensor product $\H_{\tilde A} \otimes \H_S$.
    Hence, the image of $\hat V^{\to A}_{r=0}$ is of the form
    $\H_{r=0} \otimes \H'$ with $\H_{r=0} \subset \H_{\tilde{A}}$
    the one-dimensional trivial representation space spanned by
    $\ket{0}_{\tilde A}$ and $\H' = \hat\Pi_S \H_S \subset \H_S$.
    Any $G$-invariant operator (restricted to $r=0$ and transformed using
    $\hat V^{\to A}_{r=0}$) must thus act on $\H'$.
    Note that this is in line with the $r=0$ subspace
    in the decomposition of theorem~\ref{thm:QRF_POV}.
    Using proposition~\ref{prop:general_jump} it follows that
    a general relational operator, restricted to $r=0$,
    is transformed by $\hat V^{\to A}_{r=0}$ to
    \begin{equation*}
        (\pi_{r=0,\tilde A} \otimes \sfid_S)[\hat\pi]
        (\hat\id_{\tilde A} \otimes \hat f_S)
        (\pi_{r=0,\tilde A} \otimes \sfid_S)[\hat\pi],
    \end{equation*}
    which must be of the form $\ketbra{0}{0}_{\tilde A} \otimes \hat f'$,
    where $\hat f' = \Pi_S[\hat f_S]$ is obtained from $\hat f_S$
    by projection onto $\H'$ because
    $(\pi_{r=0,\tilde A} \otimes \sfid_S)[\hat\pi]$ factorizes.
    Since $\hat f_S$ is arbitrary, so is $\hat f'$.
    In conclusion, relational operators, restricted to $r=0$
    and transformed using $\hat V^{\to A}_{r=0}$,
    act as the algebra of all operators on $\H'$
    and their commutant is trivial. This concludes the proof of the claim
    that $\hat V^{\to A}_{r=0}$ satisfies the properties
    of a QRF jump restricted to $r = 0$.
    
    Let $\hat{U}'^{\to A}$ be a QRF jump whose existence is guaranteed by
    theorem~\ref{thm:QRF_POV}. Note that for this isomorphism it holds
    that $\hat{U}'^{\to A} (\hat \id_{AS} - \hat\pi_{r=0,AS}) =
    (\hat \id - \hat\pi_{r=0}) \hat{U}'^{\to A}$, as 
    \begin{multline*}
        \hat{U}'^{\to A} \hat\pi_{r=0,AS} 
        = \frac{1}{|G|} \int_G\der g\, \hat{U}'^{\to A} (\hat U_A(g) \otimes \hat U_S(g))
        = \frac{1}{|G|} \int_G\der g\, \Big(\bigoplus_{k,r} \hat U_r(g) \otimes
        \hat\id_{r,k} \otimes \hat\id_k\Big)  \hat{U}'^{\to A}
        = \hat\pi_{r=0} \hat{U}'^{\to A}.
    \end{multline*}
    Therefore, the map 
    \begin{equation*}
       \hat U^{\to A} \coloneq \hat{U}'^{\to A} (\hat \id_{AS} - \hat\pi_{r=0,AS}) + \hat V^{\to A}_{r=0}
    \end{equation*}
    is an isomorphism such that $\hat U^{\to A}: \mathcal{H}^{|L} \to \bigoplus_{k,r
    \neq 0} (\H_r \otimes \H_{r,k} \otimes \H_k) \oplus
    (\H_{r=0} \otimes \H')$, and it satisfies the properties of
    theorem~\ref{thm:QRF_POV} by construction. For example, consider
    the algebra $\A^{S:A|L}$, then as required by theorem~\ref{thm:QRF_POV}:
    \begin{multline*}
        \mathsf U^{\to A}[\A^{S:A|L}] = 
        (\hat \id- \hat\pi_{r=0}) \mathsf U'^{\to A}[\A^{S:A|L}] (\hat \id - \hat\pi_{r=0}) 
        + \mathsf V^{\to A}_{r=0}[\A^{S:A|L}] \\
        + (\hat \id - \hat\pi_{r=0}) \hat{U}'^{\to A} \A^{S:A|L} (\hat V^{\to A}_{r=0})^{\dagger}
        + \hat V^{\to A}_{r=0}\A^{S:A|L} ( \hat{U}'^{\to A})^{\dagger}(\hat \id - \hat\pi_{r=0}) \\
        =  (\hat \id - \hat\pi_{r=0}) \mathsf U'^{\to A}[\A^{S:A|L}] (\hat \id - \hat\pi_{r=0})
        + \mathsf V^{\to A}_{r=0}[\A^{S:A|L}] 
       = \Big(\bigoplus_{k,r \neq 0} \hat\id_r \otimes \hat\id_{r,k} \otimes \O_k\Big) \oplus \Big(\hat\id_{r = 0} \otimes \Pi_S[\mathcal{O}_{S}]\Big).
    \end{multline*}    
\end{proof}

\medskip
\paragraph*{Operational Approach.}
The approach of~\cite{Loveridge2018,Glowacki2023a,Glowacki2023,Carette2025},
often termed the \emph{operational approach to QRFs},\footnote{
    This is not to say that other approaches cannot be operational too;
    for example, we also claim our approach to be operational.
}
is a mathematically rigorous\footnote{
    The following comparison between the operational and our approach
    does not aim for the level of rigour of the operational approach
    and hence glosses over many technicalities for ease of understanding.
} approach, which focuses
on traces of the form given in \cref{eq:fuzziness}
such as expectation values or outcome probabilities.
Like our approach, the operational approach relies mathematically
on a laboratory perspective, which may or may not be physical~%
\cite{Loveridge2018}, and can treat quite general QRFs with orientation POVMs.
Unlike our approach, QRF transformations in the operational approach
however do not aim at re-factorizing the Hilbert space,
and instead rely fully on the relationalization maps and their duals:
the dual of the relationalization map, which we give explicitly below,
defines the jump into the perspective of a quantum reference frame.

More precisely, let $A$ and $B$ be QRFs and $S$ a system of interest,
all equipped with a unitary representation of a symmetry group $G$.
Following~\cite{Carette2025} and focusing only on QRF transformations out of regular QRFs,\footnote{
    Strictly speaking, the more technical notion of a
    \emph{localizable principal frame} is used~\cite{Carette2025}.
} 
we assume that $A$ is regular.
The duals of the relationalization
maps $\mathsf R_A$ and $\mathsf R_B$ are
\begin{equation*}
    \mathsf R^\dagger_A[\;\cdot\;]
    := \int_G\der g\, \bigl( \bra{g}_A \otimes \hat U^\dagger_{BS}(g)
    [\;\cdot\;] \ket{g}_A \otimes \hat U_{BS}(g) \bigr), \qquad
    \mathsf R^\dagger_B[\;\cdot\;]
    := \int_G\der g\, \tr_B\bigl( \hat\gamma_B(g)
    \otimes \hat U^\dagger_{AS}(g) [\;\cdot\;] \hat U_{AS}(g) \bigr).
\end{equation*}
We implicitly introduced them in \cref{eq:fuzziness}.
Note the specific form of the partial trace in $\mathsf R^\dagger_A$
due to $A$ being regular.

These maps produce what we called \emph{effective states}
on $BS$ and $AS$ respectively.
In the operational approach, effective states relative to $A$
are called \emph{$A$-relational states}
and they are interpreted as the states from the perspective of $A$.
Similarly, one defines $B$-relational states.

The operational approach then considers an identification
of $A$-relational states resulting in the
\emph{$B$-framed $A$-relational states}:
two $A$-relational states are identified as the same $B$-framed
$A$-relational state if their expectation values agree on all
operators of the form $\int_G\der g\, \hat\gamma_B(g) \otimes \hat f_S(g)$,
with $\hat f_S(g)$ arbitrary.
These operators are called \emph{$B$-framed operators};
the Hermitian ones among them are, in the operational approach,
interpreted as the relevant observables in experiments involving
the QRF $B$~\cite{Carette2025}.
Similarly, one can define $A$-framed $B$-relational states
based on $A$-framed operators.

The QRF transformation from $A$ to $B$ in the operational approach
acts on $B$-framed $A$-relational states 
and produces $A$-framed $B$-relational states.
It can be defined as
\begin{equation*}
    \mathsf S^{A\to B} [\;\cdot\;]
    := \mathsf R^\dagger_B \bigl[ \ketbra{e}{e}_A \otimes [\;\cdot\;] \bigr]
    = \frac{1}{|G|} \mathsf R^\dagger_B \circ \mathsf R_A[\;\cdot\;].
\end{equation*}
The right-hand side is obtained by observing, at least formally, that
$\mathsf R^\dagger_B \circ \mathsf G_{ABS} = \mathsf R^\dagger_B$ and
$\mathsf R_A[\;\cdot\;] = |G|\, \mathsf G_{ABS}\bigl[ \ketbra{e}{e}_A
\otimes [\;\cdot\;] \bigr]$. It can be shown that this transformation
is invertible if $B$ is regular too, but not in general. In the former
case, the transformation also agrees
with the perspective-neutral transformation between regular QRFs $A$ and $B$,
provided that one does the same identifications of states
as in the operational approach.\footnote{
    See~\cite{Carette2025} for the details of this correspondence.
}

Now as is evident from the discussion around \cref{eq:fuzziness}
in the main text, the relationalization maps and their duals are generally
non-invertible even if the QRFs are regular.
Thus, the reversible QRF transformations obtained in the case of regular
$B$ are made possible in the operational approach by \emph{restricting}
the degrees of freedom, specifically to $B$-framed $A$-relational states.
In contrast, our approach \emph{extends} the considered degrees of freedom
to include specifically the commutant in order to obtain reversible QRF transformations;
this allows us to have reversible transformations between general QRFs.

\medskip
\paragraph*{Perspectival Approaches.}
The approach of~\cite{Giacomini2019} for the translation and Galilei
boost groups and its generalization~\cite{Hamette2020} to more general
groups are \emph{perspectival approaches}:
they do not refer to any external structure like a laboratory
or use a perspective neutral structure.
Rather, they postulate the perspective of a QRF as well
as the QRF transformation to other QRFs directly.
While~\cite{Giacomini2019} deals only with
regular QRFs, \cite{Hamette2020} deals with both regular QRFs and with
a special class of non-ideal QRFs, which can only resolve one of the
factors if $G$ is a (semi-direct) product of groups.

It has been shown~\cite{Vanrietvelde2020,Hamette2021} that the
QRF transformations of~\cite{Giacomini2019} and~\cite{Hamette2020}
are mathematically equivalent to those of the perspective-neutral framework
in the case of regular QRFs.
Therefore, they are also obtained mathematically from our framework
in the regular case by restricting to the trivial representation $r = 0$.
In the specific non-ideal cases studied in~\cite{Hamette2020},
the QRF transformations are not invertible,
and hence cannot be straightforwardly compared to our approach.

Regardless of whether the QRFs are regular, perspectival approaches
differ from our approach in that they do not (at least not explicitly)
require compatibility with an external perspective, be this a physical
or purely mathematical perspective.
However, they can be shown to be compatible with our approach in specific cases,
due to their connection with the perspective-neutral framework.

\bibliography{bibliography.bib}

\end{document}